\documentclass[11pt,a4paper]{article} 

\usepackage[a4paper,hmargin=2.8cm,vmargin=3cm]{geometry}
\usepackage[utf8x]{inputenc} 
\usepackage{amsmath,amsthm,amsfonts,amssymb,mathrsfs,dsfont}
\usepackage{authblk}
\usepackage[bookmarksnumbered=true]{hyperref}
\usepackage{dsfont} 
\usepackage{color} 
\usepackage[capitalise]{cleveref}
\crefname{equation}{}{} 

\usepackage{graphicx}
\usepackage{tikz-cd}
\usetikzlibrary{positioning}
  
\newcommand\pgfmathsinandcos[3]{%
  \pgfmathsetmacro#1{sin(#3)}%
  \pgfmathsetmacro#2{cos(#3)}%
}
\newcommand\LongitudePlane[3][current plane]{%
  \pgfmathsinandcos\sinEl\cosEl{#2} % elevation
  \pgfmathsinandcos\sint\cost{#3} % azimuth
  \tikzset{#1/.style={cm={\cost,\sint*\sinEl,0,\cosEl,(0,0)}}}
}

\newcommand\LatitudePlane[3][current plane]{%
  \pgfmathsinandcos\sinEl\cosEl{#2} % elevation
  \pgfmathsinandcos\sint\cost{#3} % latitude
  \pgfmathsetmacro\yshift{\RadiusSphere*\cosEl*\sint}
  \tikzset{#1/.style={cm={\cost,0,0,\cost*\sinEl,(0,\yshift)}}} %
}

\newcommand\DrawLongitudeArc[4][black]{
  \LongitudePlane{\angEl}{#2}
  \tikzset{current plane/.prefix style={scale=1}}
  \pgfmathsetmacro\angVis{atan(sin(#2)*cos(\angEl)/sin(\angEl))} %
  \pgfmathsetmacro\angA{mod(max(\angVis,#3),360)} %
  \pgfmathsetmacro\angB{mod(min(\angVis+180,#4),360} %
  \draw[current plane,#1]  (\angA:\RadiusSphere) arc (\angA:\angB:\RadiusSphere);
}%
\newcommand\DrawLatitudeCircle[2][1]{
  \LatitudePlane{\angEl}{#2}
  \tikzset{current plane/.prefix style={scale=1}}
  \pgfmathsetmacro\sinVis{sin(#2)/cos(#2)*sin(\angEl)/cos(\angEl)}
  \pgfmathsetmacro\angVis{asin(min(1,max(\sinVis,-1)))}
  \draw[current plane] (\angVis:\RadiusSphere) arc (\angVis:-\angVis-180:\RadiusSphere);
}

\newcommand\DrawLatitudeArc[4][black]{
  \LatitudePlane{\angEl}{#2}
  \tikzset{current plane/.prefix style={scale=1}}
  \pgfmathsetmacro\sinVis{sin(#2)/cos(#2)*sin(\angEl)/cos(\angEl)}
  \pgfmathsetmacro\angVis{asin(min(1,max(\sinVis,-1)))}
  \pgfmathsetmacro\angA{max(min(\angVis,#3),-\angVis-180)} %
  \pgfmathsetmacro\angB{min(\angVis,#4)} %

  \draw[current plane,#1] (\angA:\RadiusSphere) arc (\angA:\angB:\RadiusSphere);
}

\newtheorem{thm}{Theorem}[section]

\newtheorem{lem}[thm]{Lemma}

\theoremstyle{remark}

\renewcommand{\theequation}{\thesection.\arabic{equation}}

\numberwithin{equation}{section}

\newcommand{\fock}{\mathcal{F}}		% fock space symbol
\newcommand{\di}{{\textnormal{d}}}		% differential (for integrals)
\newcommand{\rfrak}{\mathfrak{r}}
\newcommand{\Acal}{\mathcal{A}}

\newcommand{\Tbb}{\mathbb{T}}
\newcommand{\W}{W}
\newcommand{\Wt}{\widetilde{W}}
\newcommand{\D}{D}

\newcommand{\Ecal}{\mathcal{E}}
\newcommand{\Hbb}{\mathbb{H}}

\newcommand{\Ncal}{\mathcal{N}}		% calligraphic N
\newcommand{\Hcal}{\mathcal{H}}		% calligraphic H
\newcommand{\Ical}{\mathcal{I}}
\newcommand{\Rcal}{\mathcal{R}}
\newcommand{\Ikp}{\Ical_{k}^{+}}
\newcommand{\Ikm}{\Ical_{k}^{-}}
\newcommand{\Ik}{\Ical_{k}}
\newcommand{\Il}{\Ical_{l}}
\newcommand{\Ilp}{\Ical_{l}^{+}}

\newcommand{\ik}{I}
\newcommand{\Hcorr}{\Hcal_\textnormal{corr}}

\newcommand{\Sbb}{\mathbb{S}}

\newcommand{\Dfrak}{\mathfrak{D}}
\newcommand{\tD}{\widetilde{\mathfrak{D}}}
\newcommand{\tc}{\tilde{c}}

\newcommand{\Ocal}{\mathcal{O}}		% big-O, order-of
\newcommand{\hc}{\textnormal{h.c.}}		%hermitian conjugate

\newcommand{\cc}[1]{\overline{#1}}	% complex conjugate
\newcommand{\Rbb}{\mathbb{R}}		% real numbers
\newcommand{\Cbb}{\mathbb{C}}		% complex numbers
\newcommand{\Nbb}{\mathbb{N}}		% natural numbers
\newcommand{\Zbb}{\mathbb{Z}}

\newcommand{\Xbb}{\mathbb{X}}

\newcommand{\Dbb}{\mathbb{D}}

\renewcommand{\Re}{\operatorname{Re}} 	%RealPart
\newcommand{\id}{\mathbb{I}}
\newcommand{\norm}[1]{\lVert#1\rVert}	%Norm
\newcommand{\Bignorm}[1]{\Big\lVert#1\Big\rVert}
\newcommand{\tr}{\operatorname{tr}}

\newcommand{\sgn}{\operatorname{sgn}}
\newcommand{\tagg}[1]{ \stepcounter{equation} \tag{\theequation} \label{eq:#1} } % add tag and label in align*-environments

\newcommand{\BFc}{B_\F^c}
\newcommand{\BF}{B_\F}

\newcommand{\Efrak}{\mathfrak{E}}

\newcommand{\Kfrak}{\mathfrak{K}}

\newcommand{\kappaf}{\kappa}

\newcommand{\north}{\Gamma^{\textnormal{nor}}}

\newcommand{\F}{\textnormal{F}} % abbreviation for Fermi
\newcommand{\B}{\textnormal{B}} % abbreviation for Bose

\newcommand{\diag}{\operatorname{diag}}
\newcommand{\diam}{\operatorname{diam}}
\newcommand{\supp}{\operatorname{supp}}

\newcommand{\linb}{\textnormal{B}}

\title{Correlation Energy of a Weakly Interacting Fermi Gas}

\author[1,*]{Niels Benedikter}
\author[2]{Phan Th\`anh Nam}
\author[3]{Marcello Porta} 
\author[4]{Benjamin Schlein} 
\author[5]{Robert Seiringer}
\affil[1]{Universit\`a degli Studi di Milano, Dipartimento di Matematica, Via Cesare Saldini 50, 20133 Milano, Italy}
\affil[2]{LMU Munich, Department of Mathematics, Theresienstra{\ss}e 39, 80333 M\"unchen, Germany}
\affil[3]{SISSA, Mathematics Area, Via Bonomea 265, 34136 Trieste, Italy}
\affil[4]{Institute of Mathematics, University of Zurich, Winterthurerstrasse 190, 8057 Zurich, Switzerland}
\affil[5]{IST Austria, Am Campus 1, 3400 Klosterneuburg, Austria}
\affil[*]{corresponding author, \href{mailto:niels.benedikter@unimi.it}{niels.benedikter@unimi.it}}

\begin{document} 

\maketitle
\begin{abstract}
We derive rigorously the leading order of the correlation energy of a Fermi gas in a scaling regime of high density and weak interaction. The result verifies the prediction of the random--phase approximation. Our proof refines the method of collective bosonization in three dimensions. We approximately diagonalize an effective Hamiltonian describing  approximately bosonic collective excitations around the Hartree--Fock state, while showing that gapless and non--collective excitations have only a negligible effect on the ground state energy.  
\end{abstract} 

\tableofcontents

\section{Introduction and Main Result}

In the last thirty years, the study of the quantum many--body problem has made tremendous progress, in particular for weakly interacting regimes where the validity of mean--field theory (or slightly more generally the quasi--free approximation) as an effective theory can be proved. In particular for bosonic systems the mathematical results have been very rich. 
% , especially after the experimental observation of Bose--Einstein condensation in 1995 \cite{AEM+95,DMA+95}.
Just to name some: in the beginning of the 2000s the Gross--Pitaevskii functional for the ground state energy of dilute Bose gases was derived \cite{LSY00,LS02}. Later the time--dependent Gross--Pitaevskii equation was derived \cite{ESY09a,ESY10}; bounds on the rate of convergence were obtained by \cite{BdS15,BS19a}. 
In 2011 validity of the quasi--free approximation for the excitation spectrum of Bose gases in the mean--field regime was proven \cite{Sei11}, thus obtaining also the next--to--leading order of the ground state energy. In contrast, for dilute gases, the quasi--free approximation is not sufficient for obtaining the second order of the energy, although it can be used to derive the leading order with optimal rate of convergence \cite{BBCS18b,BBCS20,NNRT20}. Very recently, results going beyond the quasi--free approximation were obtained: the excitation spectrum for dilute Bose gases was derived \cite{BBCS17,BBCS19}; the Lee--Huang--Yang formula for the second order of the ground state energy was proven \cite{FS19}; and nonlinear classical Gibbs measures were derived as an approximation at positive temperature \cite{LNR20,FKSS20}.   

\medskip
Compared to the development in the theory of bosonic systems, the mathematical progress in the derivation of effective theories for fermionic systems has been lagging behind. For fermions, the mean--field or quasi--free theory leads to the  Hartree--Fock approximation\footnote{In this paper we focus on a setting where the pairing density is not relevant. If the pairing density becomes important, one is lead to Hartree--Fock--Bogoliubov theory or the Bardeen--Cooper--Schrieffer (BCS) theory of superconductivity. Already the study of these quasi--free theories is very challenging. Recently the mathematical properties of BCS theory were extensively analyzed \cite{FHSS16,HHSS08}, and the Ginzburg--Landau theory of superconductivity was derived from BCS theory \cite{FHSS12}.} which is widely used in computational physics and chemistry. The validity of the Hartree--Fock approximation was established for the ground state energy of Coulomb systems in a number of seminal works \cite{FS90,Bac92,GS94}.  Rigorous results taking this analysis beyond the quasi--free effective theory have been notably absent, except for a second--order bound \cite{HPR20} on the many--body correction (called correlation energy) to the ground state energy, inspired by \cite{Hai03, HHS05}. In the present paper we derive an optimal formula for the correlation energy.
 
 Our proof is based on a non--perturbative framework which we started to develop in \cite{BNP+20}. The central concept of our approach is that the dominant degrees of freedom are particle--hole pairs which are delocalized over patches on the Fermi surface in momentum space in such a way that they behave \emph{approximately} as quasi--free bosons. In \cite{BNP+20}, by means of a trial state, we proved that the formula known as the \emph{random phase approximation (RPA)} in physics is an \emph{upper bound} to the correlation energy of a three--dimensional Fermi gas in the mean--field scaling regime (i.\,e., high density and weak interaction) with a regular interaction potential. In the present paper, we again start from the interacting many--body Hamiltonian and prove the matching \emph{lower bound}, thus completely validating the random--phase approximation for the ground state energy of the three--dimensional Fermi gas in the mean--field scaling regime.    

\medskip

The problem of calculating corrections to the Hartree--Fock approximation has a long history in theoretical physics. Already in the early days of quantum mechanics the computation of the correlation energy was attempted using second order perturbation theory \cite{Wig34, Hei47} for a Fermi gas with Coulomb interaction (the electron gas); however, this approach leads to a logarithmically divergent expression due to the long range of the Coulomb potential. It was then noticed \cite{Mac50} that perturbation theory with Coulomb potential becomes even more divergent at higher orders and suggested that a resummation might cure this problem. Then in their seminal work \cite{BP53}, Bohm and Pines developed the RPA: they argued that the Hamiltonian can be partially transformed into normal coordinates which describe collective oscillations screening the long--range of the Coulomb potential, and thus leading to a better behaved perturbative expansion. However, they had to introduce additional bosonic collective degrees of freedom by hand. This was somewhat clarified by \cite{SBFB57,Saw57}, who showed that the collective modes can be understood as a superposition of particle--hole pair excitations. The formulation of the RPA due to Sawada et al.\ has in fact been an important inspiration for our work. Ultimately it was discovered that the RPA can be seen as a systematic partial resummation of perturbation theory; following this line,  one even obtains a more precise result \cite{GB57}. These works have been very influential in the establishment of theoretical condensed matter physics.  

The particle--hole pair bosonization of Sawada et al.\ found application in many settings, for example to describe nuclear rotation and calculate moments of inertia of atomic nuclei \cite{MW69,AP75}. A bosonization method considering only the radial excitations of the Fermi surface was developed by \cite{Lut79}; similar methods applied to systems with square Fermi surface \cite{FSL99,SL05}.  Later, the bosonization of collective excitations of the Fermi surface became an important tool in the context of renormalization group methods \cite{BG90,HM93,HKMS94,Hal94,CF94}. The collective aspect was further emphasized in the operator--formalism by \cite{CF94a,CF95}. In the functional--integral formalism \cite{FGM95,KHS95,Khv95,KC96,KS96,FG97} bosonization was established as a Hubbard--Stratonovich transformation. Despite this popularity, difficulties in judging the quality of the bosonic approximation have been pointed out \cite{Kop97}: ``For example, scattering processes that transfer momentum between different boxes on the Fermi surface and non--linear terms in the energy dispersion definitely give rise to corrections to the free--boson approximation for the Hamiltonian. The problem of calculating these corrections within the conventional operator approach seems to be very difficult.'' As far as the mean--field scaling regime is concerned, with our result we quantify such corrections as being of subleading order.     

\medskip

% EXCLUDED
A different mathematical approach to the fermionic many--body problem has been developed by employing rigorous renormalization group methods to construct convergent perturbative expansions. This allowed the construction of Gibbs states or ground states for two main classes of interacting fermionic models.
  
The first class concerns models in the Luttinger liquid universality class (which was first proposed by Haldane \cite{Hal80, Hal81}), such as interacting fermions or quantum spin chains in one dimension and some two--dimensional models. These models show universal properties agreeing with those of the  Tomonaga--Luttinger model which is solvable in one dimension by an \emph{exact} bosonization method \cite{ML65}. These predictions of bosonization have been verified rigorously, starting from \cite{BG90,BGPS94} to \cite{BM01,BM04,BM11,BFM09a,BFM09b,AMP18,MP17,GMT20}; the proofs however are by detailed analysis of the fermionic theory instead of justifying directly the bosonization. One justification of a bosonization method was achieved by \cite{BW20}, showing equivalence of the massless sine--Gordon model for a special choice of the coupling constant and the massive Thirring model at the free fermion point.   

The second class concerns fermions in two or three dimensions at low temperature. In this context, the use of sectors on the Fermi surface, very similar to the construction of patches we use, has been introduced in \cite{FMRT92} for the program of proving existence of superconductivity \cite{FMRT95}. There, bosonization was implemented as a Hubbard--Stratonovich transformation of sectorized collective excitations. While this ambitious program has not been completed, the sector method was later used to prove Fermi liquid behavior of fermions in two dimensions with uniformly convex Fermi surface at exponentially small positive temperatures (and non--Fermi liquid behavior for fermions with flat Fermi surface) \cite{DR00a,DR00b,Riv02,AMR05a,AMR05b,BGM06}. It furthermore lead to a proof of  convergence for the zero--temperature perturbation theory in a special two--dimensional fermionic model with an asymmetry condition of the Fermi surface; this is a series of eleven papers an overview of which is given in \cite{FKT04}. Partial results have been obtained for fermions in three dimensions at positive temperature \cite{DMR01}. We see our approach, while sharing the `sectorization' or `patches' concept, as providing a complementary point of view on related physical problems, based on different, non--perturbative ideas.    

Finally, our result should also be contrasted to the study of two--dimensional models that have been constructed to be exactly bosonizable. This goes back to a proposal of \cite{Mat87}, who was motivated by high--temperature superconductivity. The analysis and variants of the model were developed by \cite{Lan10a,Lan10b,dL10,dL12a,dL12b,dL14}.  
Furthermore, one may also see similarities (such as the limitation of the number of bosons that can occupy a single bosonic mode) in the bosonization concept to methods such as the Holstein--Primakoff map \cite{CG12,CGS15,Ben17} for spin systems. 

\subsection{Many--Body Hamiltonian in the Mean--Field Regime} 
To describe $N$ spinless fermionic particles on the torus $\Tbb^3 := \Rbb^3/(2\pi \Zbb^3)$, the Hilbert space is the space of totally antisymmetric $L^2$--functions of $N$ variables in $\Tbb^3$,
\begin{equation} \label{eq:antisymm_space}
L^2_\textnormal{a}(\Tbb^{3N}) := \{ \psi \in L^2(\Tbb^{3N}) : \psi(x_{\sigma(1)}, \ldots, x_{\sigma(N)}) = \sgn(\sigma) \psi(x_1,\ldots,x_N)\ \forall \sigma \in \mathcal{S}_N \}\;.
\end{equation}
The Hamiltonian is defined as the sum of Laplacians describing the kinetic energy\footnote{Compared to the mass $m=1$ in \cite{BNP+20}, we now choose $m=1/2$.} and a pair interaction, i.\,e., a multiplication operator defined using a function $V:\Rbb^3 \to \Rbb$,
\begin{equation} \label{eq:HN}
H_N := \hbar^2 \sum_{i=1}^N \left( - \Delta_{x_i} \right) + \lambda \sum_{1 \leq i < j \leq N} V\left( x_i - x_j \right)\;.
\end{equation}
The positive parameters $\hbar$ and $\lambda$ adjust the strength of the kinetic energy and interaction operator, respectively.

In this paper, we assume the interaction potential $V$ to be smooth. Thus the Hamiltonian is bounded from below and its self--adjointness follows from the Kato--Rellich theorem or using the Friedrichs extension. Here we are interested in the infimum of the spectrum (the \emph{ground state energy}) 
\begin{equation} \label{eq:EN-def}
E_N := \inf \operatorname{spec}\left( H \right) = \inf \Big\{\langle \psi, H_N \psi  \rangle: \psi \in L^2_\textnormal{a}(\Tbb^{3N})\;,\ \norm{\psi}_{L^2}=1 \Big\}\;. 
\end{equation}
%The eigenvectors associated with $E_N$ are called ground states of the system.
% According to the Perron--Frobenius theorem, the ground state is up to a phase unique under the assumptions that we are going to impose on the interaction potential $V$.

 In full generality, the computation of $E_N$ is clearly out of reach, simply because the model is too general: it may describe physical systems from superconductors to neutron stars. We thus need to be more specific and consider a particular case of the model, the most accessible case being a mean--field scaling regime: by considering a high density of particles we expect the leading order of the theory to be approximately described by an effective one--particle theory. We thus consider the limit of large particle number on the fixed--size torus. However, kinetic energy and interaction energy in typical states scale differently: the kinetic energy like $N^{5/3}$ due to the Pauli exclusion principle, the interaction energy like $N^2$ since there are $N(N-1)/2$ interacting pairs. To have a chance of obtaining a non--trivial limit we choose to scale the parameters by\footnote{Of course we can also set $\hbar=1$ or $\lambda=1$ and scale only the other parameter. The scaling \cref{eq:choice-hbar-lambda} becomes non--trivial when studying the dynamics, where it relates to a rescaling of time \cite{BPS14}.}
\begin{align} \label{eq:choice-hbar-lambda}
\hbar :=N^{-\frac{1}{3}} \quad \textnormal{and} \quad \lambda:=N^{-1} \quad \textnormal{with }N\to\infty\;.
\end{align}
With this choice the kinetic energy and the interaction energy in typical states close to the ground state have the same order of magnitude (order $N$). This scaling regime couples a semiclassical scaling ($\hbar = N^{-\frac{1}{3}}\to 0$) and a mean--field scaling (coupling constant $\lambda=N^{-1}$).

If the interaction vanishes, $V=0$, then the ground state of the system is exactly given by the Slater determinant (i.\,e., antisymmetrized tensor product) of plane waves
%, where the plane waves are chosen such that the kinetic energy is minimized:
\begin{equation}
\psi_\textnormal{pw} = \bigwedge_{k \in \BF} f_k\;, \qquad f_k(x) = (2\pi)^{-\frac{3}{2}} e^{ik\cdot x}\quad \textnormal{with } k \in \Zbb^3,\ x \in \Tbb^3\;. \label{eq:plane-waves}
\end{equation}
Here the momenta $k$ of the plane waves are chosen such that the expectation value of the kinetic energy operator is minimized. The set of the corresponding momenta $\BF$ is called the Fermi ball. For simplicity we assume that the ball is completely filled, namely we set
\begin{equation}\label{eq:Fermiball}
\BF := \{ k \in \Zbb^3: \lvert k\rvert \leq k_\F\}\;,
\end{equation}
and then define the particle number accordingly as $N := \lvert \BF \rvert$. The limit of large particle number is then realized by considering $k_\F \to \infty$. According to Gauss' classic counting argument we have\footnote{In \cite{BNP+20} we also introduced $\kappa_0 = (3/4\pi)^{\frac{1}{3}}$ and compared explicitly to $\kappaf = \left(3/4\pi\right)^{\frac{1}{3}} + \Ocal(N^{-1/3})$ throughout the paper. However, the estimates for this error are simple to follow and of lower order, so that for readability we do not spell them out in the present paper.}
\[k_\F =  \kappaf N^{\frac{1}{3}} \qquad \textnormal{for} \quad \kappaf = \left(3/4\pi\right)^{\frac{1}{3}} + \Ocal(N^{-1/3})\;.\]

If the system is interacting, $V\not= 0$, the ground state becomes a complicated superposition of Slater determinants. Nevertheless, in Hartree--Fock theory one minimizes only over the set of all Slater determinants. In our setting, the Hartree--Fock energy
\[
E_N^\textnormal{HF}:= \inf \Big\{ \langle \psi, H_N \psi\rangle :  \psi=\bigwedge_{i=1}^N u_i \text{ with } \{u_i\}_{i=1}^N \text{ an orthonormal family in } L^2(\mathbb{T}^3) \Big\}
\]
is attained by the plane waves as in \cref{eq:plane-waves,eq:Fermiball}; see \cref{app:B} for a proof\footnote{This fact is special for the completely filled Fermi ball of a homogeneous gas in finite volume. In general, the plane waves state is not even a local minimum of the Hartree--Fock functional \cite{GL18}.}. Thus in order to gain non--trivial information about the interacting system one must go beyond the Hartree--Fock theory.

Note that by the variational principle, the Hartree--Fock energy $E_N^\textnormal{HF}$ is an upper bound to the ground state energy $E_N$. % 
It follows from the analysis of \cite{Bac92,GS94} that Hartree--Fock theory also provides a good lower bound to the ground state energy. In our setting, the approach of \cite{Bac92,GS94} shows that 
% \[ E_N = E_N^{HF} + o (1) \]
% % 
% For the high--density limit, in \cite{Bac92,GS94} it was proven that Hartree--Fock theory also provides a good lower bound to the true ground state energy, namely\footnote{Actually, in \cite{Bac92,GS94} fermions interacting through the Coulomb potential are considered, where the exchange term (Dirac term) is of order $N^{\frac{1}{3}}$ (after translating to our choice of coupling constants), and their error estimate $o(N^{\frac{1}{3}})$. This changes when the potential is compactly supported in Fourier space, when the exchange term becomes order $1$.} 
\begin{equation}    \label{eq:correl_en}
E_N = E_N^\textnormal{HF} + o(1)  \qquad \textnormal{as $N\to\infty$}\;.
\end{equation}
In particular, both $E_N$ and $E_N^\textnormal{HF}$ contain the Thomas--Fermi energy (in our scaling of order $N$) and the Dirac correction, also know as the exchange term (in our scaling of order $1$).

%Even though the Slater determinant of plane waves in \cref{eq:plane-waves} constitutes a critical point of the Hartree--Fock variational problem, it is known that it is not even a local minimum \cite{GL18}. However, the energy difference to the global Hartree--Fock minimizer is very small. Recently \cite{GHL19} it was proven that \footnote{In \cite{GHL19}, the bound was given as $\Ocal(\exp(-N^{\frac{1}{6}}))$ because they consider Coulomb interaction; for regular potentials one gets the stronger exponent.}
%\begin{equation}
%E_N^\textnormal{HF,pw} := \langle \psi_\textnormal{pw}, H_N \psi_\textnormal{pw}\rangle = E_N^\textnormal{HF}  + \Ocal(e^{-N^{\frac{1}{3}}})\;.
%\end{equation}
%In particular, this implies that in order to gain non--trivial information about the interacting system one must go beyond Hartree--Fock theory. 
%
%thm:HF

From the physical point of view, Slater determinants are as uncorrelated as fermionic states (which have to satisfy the Pauli principle) can be, in the sense that they are just antisymmetrized tensor products. Due to the presence of the interaction, the true ground state will contain non--trivial correlations (i.\,e., it will be a superposition of Slater determinants). Therefore Wigner \cite{Wig34} called the difference
\[
E_N - E_N^\textnormal{HF}
\]
the  \emph{correlation energy}. According to \cref{eq:correl_en} we know that the correlation energy in our scaling is of size $o(1)$ as $N \to \infty$. In the present paper, we are going to determine the leading order of the correlation energy. It is of order $\hbar = N^{-\frac{1}{3}}$ and given by the explicit formula predicted by the random--phase approximation, as obtained by \cite{Mac50,GB57} based on a partial resummation of the perturbation series. We believe that our result is of importance as a rigorous step beyond mean--field theory into the world of interacting quantum systems. Our proof shows that the leading order of the correlation energy can be understood as the ground state energy of an effective quadratic Hamiltonian describing approximately bosonic collective excitations.

\subsection{Main Result}
We write the interaction potential via its Fourier coefficients
\[
V(x) = \sum_{k\in \Zbb^3} \hat V(k) e^{i k \cdot x}\;. 
\]

%are going to denote the Fourier transform of the interaction potential $V: \Rbb^3 \to \Rbb$ by 
%$$\hat{V}(k) = (2\pi)^{-3} \int_{\Rbb^3} \di x\, e^{-ik\cdot x} V(x).$$

\begin{thm}[Main Result]
\label{thm:main}
There exists a $v_0 > 0$ such that the following holds true.
Assume that $\hat{V}: \Zbb^3 \to \Rbb$ is compactly supported, non--negative, satisfies $\hat{V}(k)=\hat{V}(-k)$ for all $k\in \Zbb^3$, and $\norm{\hat{V}}_{\ell^1}< v_0$. For every $k_\F >0$ let the particle number be $N := \lvert\{ k\in \Zbb^3: \lvert k\rvert \leq k_\F \}\rvert$. Then as $k_\F \to \infty$, the ground state energy of the Hamiltonian $H_N$ in \eqref{eq:HN} with $\hbar=N^{-1/3}$ and $\lambda=N^{-1}$ is
%$$
%H_N := \hbar^2 \sum_{i=1}^N \left( - \Delta_{x_i} \right) + \lambda \sum_{1 \leq i < j \leq N} V\left( x_i - x_j \right)\;.
%$$
\begin{equation} \label{eq:Ecorr-main-thm}
E_N = E_N^\textnormal{HF} + E_N^\textnormal{RPA} + \Ocal(\hbar^{1+\frac{1}{16}})\;.
\end{equation}
Here the correlation energy $E_N^\textnormal{RPA}$ is of order $\hbar$ and, with $\kappaf = \left(\frac{3}{4\pi}\right)^{\frac{1}{3}}$, given by
\begin{equation}
E_N^\textnormal{RPA} = \hbar \kappaf \sum_{k \in \Zbb^3} \lvert k\rvert \left( \frac{1}{\pi}\int_0^\infty \log\left[1+2\pi\kappa\hat{V}(k) \left(1-\lambda \arctan\left(\lambda^{-1}\right)\right) \right] \di \lambda - \frac{ \pi}{2}\kappa\hat{V}(k) \right)\;. \label{eq:rpa_energy}
\end{equation}
\end{thm}

The upper bound, $E_N \leq E_N^\textnormal{HF} + E_N^\textnormal{RPA} + \Ocal(\hbar^{1+\frac{1}{9}})$, was proved in \cite{BNP+20}, even without smallness condition on the potential. In the present paper we prove the lower bound. The smallness condition is technical, and we expect that the lower bound is also true without this condition.

\medskip

As already explained in  \cite{BNP+20}, by expanding \cref{eq:rpa_energy} for small $\hat V$, we obtain
\begin{equation} \label{eq:rpa_energy_small_V}
\frac{E_N - E_N^\textnormal{HF}}{\hbar}= m\pi (1-\log (2)) \sum_{k\in \mathbb{Z}^3} \lvert k\rvert \lvert\hat V(k)\rvert^2 \left(1+\mathcal{O}(\hat V(k))\right) + \Ocal(\hbar^{1+\frac{1}{16}})\;. 
\end{equation} 
Thus we recover the result for the weak--coupling limit of \cite{HPR20}. Moreover, the leading order of the correlation energy of the jellium model as given by Gell-Mann and Brueckner \cite[Eq.~(19)]{GB57} (see also \cite[Eq.~(37)]{SBFB57} and \cite{Mac50}) when applied to the case of bounded compactly supported $\hat{V}$ agrees with \cref{eq:rpa_energy}.   
 
\medskip   

Although some tools from the earlier papers \cite{BNP+20,HPR20} will be useful for us, the proof of Theorem \ref{thm:main} requires several important new ingredients. Conceptually, our justification of the random phase approximation is based on the main input that at the energy scale of the correlation energy there are rather few excitations around the Fermi ball. For the upper bound in \cite{BNP+20}, we consider a trial state whose number of excited particles is of order 1, allowing to control most of error terms easily. However, for the lower bound, the best available estimate for the number of excited particles in a ground state is $\mathcal{O}(N^{\frac 1 3})$, thanks to a kinetic inequality from \cite{HPR20}. This weaker input breaks most of the error estimates in the upper bound analysis \cite{BNP+20}, and this is also the reason why a less precise lower bound was obtained in \cite{HPR20}. In fact, using only similar bounds to \cite{BNP+20}, we can at best show that the error terms are of the same order as the correlation energy. In the present paper, we go beyond that and complete the bosonization approach for the first time. 

\medskip  Let us quickly mention the most important new ingredients of the proof; a more detailed explanation will be given in \cref{sec:sketch}. 
 
\begin{itemize}

\item \emph{A refined estimate for the number of bosonic particles.} In \cite{BNP+20}, we control the number of bosonic particles by the fermionic number operator $\Ncal$. This is insufficient here, since the bound $\langle \Ncal\rangle \le CN^{\frac 1 3}$ mentioned above is too weak. It is natural to try to bound all error terms using the kinetic operator $\Hbb_0$, but a serious problem is that $\Hbb_0$ is \emph{not stable} under the Bogoliubov transformation introduced later. Instead, we introduce the \emph{gapped number operator} $\Ncal_\delta$ in \eqref{eq:gappedN2}, which takes into account only the fermionic particles far from the Fermi surface and has a much better bound $\langle \Ncal_\delta\rangle \le CN^\delta$ with $\delta>0$ small. Thus in practice, using  $\Ncal_\delta$ is as good as using the kinetic operator $\Hbb_0$ in many estimates, with the advantage that $\Ncal_\delta$ is \emph{stable} under the Bogoliubov transformation (see \cref{lem:stability}). Since $\Ncal_\delta$ involves the fermionic particles far from the Fermi surface, we have to control separately the contribution from particles close to the Fermi surface, using an improvement of the kinetic inequality in \cite{HPR20} (see \cref{lem:kineticequator}). The latter issue does not appear in  \cite{BNP+20} since for an upper bound we can simply  take a trial state without any contribution from particles close to the Fermi surface.

\item \emph{A refined linearization of the kinetic energy.} Similarly to \cite{BNP+20}, the bosonization approach in the present paper is based on the construction of patches, which allows to \emph{linearize} the fermionic kinetic operator $\Hbb_0$ and relates it to a bosonic operator $\Dbb_\textnormal{B}$. In \cite{BNP+20}, we prove that the expectation value of $\Hbb_0-\Dbb_\textnormal{B}$ against a well--chosen trial state is small, which requires that the number of patches is $M\gg N^{\frac 1 3}$. In the present paper, we only control the commutator of $\Hbb_0-\Dbb_\textnormal{B}$ with bosonic pairs operators (see \cref{lem:lin}). This weaker bound is sufficient to ensure that $\Hbb_0-\Dbb_\textnormal{B}$  is essentially \emph{invariant} under the Bogoliubov transformation (see \cref{lem:lin-2}), and importantly it  requires only $M\gg N^{2\delta}$ with $\delta>0$ small. The possibility of taking a much smaller $M$ is crucial to bound all error terms caused by the Bogoliubov transformation. 

\item \emph{A refined control on the Bogoliubov kernel.} Similarly to \cite{BNP+20}, we will diagonalize the bosonizable part of the Hamiltonian by a Bogoliubov transformation. In \cite{BNP+20} we prove that the kernel of the Bogoliubov transformation is bounded uniformly in the Hilbert--Schmidt topology. This information is sufficient to estimate the error terms when $\langle \Ncal \rangle \sim 1$ (as in the trial state used for the upper bound), but it is insufficient now that there are potentially many excitations.  In the present paper, we will derive an optimal bound for the matrix elements of the Bogoliubov kernel (see \cref{lem:K}). The new estimate encodes that due to the geometry of the Fermi surface, the interaction energy vanishes at the same rate as the kinetic gap closes. This bound is crucial for improving error estimates involving the Bogoliubov transformation (see \cref{lem:Bog-T}), especially for controlling the non--bosonizable terms.

\item \emph{A subtle analysis of the non--bosonizable terms.} As explained in \cite{BNP+20}, the contribution of the non--bosonizable terms can be controlled by $N^{-1}\langle \Ncal^{2}\rangle$. The trial state in \cite{BNP+20} satisfies $\langle \Ncal^{2}\rangle \sim 1$, and hence the non--bosonizable terms are much smaller than the correlation energy. In the present paper, we only know that  $\langle \Ncal \rangle \le CN^{\frac 1 3}$, which is not enough to rule out the possibility that the non--bosonizable terms are comparable to the correlation energy. It turns out that controlling the non--bosonizable terms is highly nontrivial since these terms couple the bosonic degrees of freedom with the uncontrolled low--energy fermions. Our idea is to bound these terms from below by the kinetic operator. Technically, it is easy to establish the lower bound $-C \|\hat V\|_{\ell^1} \Hbb_0$ by completing a square. However, the difficulty here is that we have to validate this bound \emph{after} implementing the Bogoliubov transformation (see \cref{lem:non--bosonizable}). Handling the non--bosonizable terms requires a subtle analysis, using the refined estimate on the Bogoliubov kernel and the smallness assumption on the innteraction potential.

%This requires the optimal bound on the Bogoliubov kernel that we established previously. The error term in this step is of the form  $-C \|\hat V\|_{\ell^1} \Hbb_0$, which will be controlled by $\Hbb_0$ if $\|\hat V\|_{\ell^1}$ is small. 

 \item \emph{Analysis of the diagonalized effective Hamiltonian.} After implementing the Bogoliubov transformation, we obtain the desired correlation energy plus $\Hbb_0-\Dbb_\textnormal{B} + \mathbb{K}$ where $\mathbb{K} = \sum_{k\in \north} \sum_{\alpha,\beta \in \Ik} 2\hbar \kappaf \lvert k\rvert  \Kfrak(k)_{\alpha,\beta} c_\alpha^*(k) c_\beta(k)$ is the diagonalized effective Hamiltonian. Here $\Hbb_0-\Dbb_\textnormal{B}$ remained since it is essentially invariant under the Bogoliubov transformation. For the upper bound in \cite{BNP+20}, the term $\mathbb{K}$ does not cause any problem since its expectation value in the vacuum state is $0$. In the present paper, however, we have to bound it from below as an operator (see \cref{eq:kinetic-Dbb0-final}). This task is nontrivial and we have to use again the refined estimate on the Bogoliubov kernel and the smallness assumption on the innteraction potential.  
%  
% 
% 
% It turns out that the difference $\widetilde{\Dbb}_\textnormal{B}-\Dbb_\textnormal{B}$ can be bounded from below by $-C \|\hat V\|_{\ell^1} \Hbb_0$, which is again controlled by $\Hbb_0$ if $\|\hat V\|_{\ell^1}$ is small. Thus by sacrificing the Bogoliubov excitation spectrum and assuming that the interaction potential is small, we can control error terms from the comparison of the fermionic and the bosonized kinetic energy.
%
\end{itemize}

In summary, in the present paper we provide a complete and unified bosonization approach which can handle the states with a lot of low--energy excitations. We believe that our approach is of general interest and could be useful in other contexts.   

\medskip

We also see our result as a possible starting point for further investigations. For example, our bosonization method is general enough to derive a norm approximation on the many--body \emph{dynamics} \cite{BNPSS21}. Many questions remain; given the historical context of the problem, maybe most importantly the extension to Coulomb interaction, i.\,e., the electron gas, at least in some coupled mean--field/large--volume limit, requiring to optimize our bounds for extensivity. Of course, to reach this goal, we would first need to remove the small--potential condition, which at the moment plays a central role. The next key task is to deal with the divergence at small $k$ which appears in the higher orders of perturbation theory. As the small--$k$ singularity is improved to a logarithmic singularity in \cref{eq:rpa_energy}, we believe that the bosonization method contains intrinsically the necessary ``resummation'' that is responsible for this screening of the potential. Of course, hard technical refinements, e.\,g., in optimizing the $k$--dependence of our estimates will be necessary. Another question concerns the low--energy spectrum of the Hamiltonian: it is believed that a collective plasmon mode can be isolated from the bosonized excitation spectrum, realizing a theory of electrons dressed by a cloud of excitations and supporting the screening concept. Within the bosonic approximation, the emergence of the plasmon mode has been discussed in \cite{Ben19}. We expect that through a detailed analysis of the spectrum, the screening of the Coulomb potential, and the properties of the approximate ground state, the bosonization method may support the future development of a rigorous, non--perturbative Fermi liquid theory. 

% \medskip

 Beyond the mean--field scaling regime and the electron gas, there are other systems of physical interest: for example the helium isotope ${}^3\textnormal{He}$ is fermionic and has short--range isotropic interactions. Furthermore, a high--density limit is particularly important in the description of atomic nuclei; the short--range interactions there are however spin-- and isospin--dependent and anisotropic and furthermore have attractive parts. We conjecture that even with attractive potentials the RPA formula for the correlation energy applies as long as the logarithm in $E^\textnormal{RPA}_N$ does not become ill--defined.
%  
%  It remains unclear whether this instability is to be identified with the superconducting instability; in general, Cooper pairing may arise even with repulsive interactions due to the Kohn--Luttinger effect. Therefore we believe that the absence of a BCS contribution to the energy is not due to the positivity assumption alone but (also) due to our scaling limit suppressing the Cooper channel; in fact, we estimate explicitly (see \cref{sec:nonboson}) the terms of the Hamiltonian that cannot be described by bosonizing particle--hole pairs.
In our scaling, we do not see any contribution from the pairing density related to superconductivity, but one may expect that even if it was non--vanishing, its effect on the energy may be exponentially small. One may speculate that in an appropriate scaling limit the state of a superconductor might be described using a product of a particle--hole pair Bogoliubov transformation as we construct it for the normal phase, times a BCS--type fermionic Bogoliubov transformation.

\subsection{Sketch of the Proof} \label{sec:sketch}

We will use the Fock space formalism. Recall the fermionic Fock space
\begin{equation} \label{eq:Fock-space}
\fock := \bigoplus_{n =0 }^\infty L^2_\textnormal{a}(\Tbb^{3n}) = \mathbb{C} \oplus L^2(\Tbb^{3}) \oplus L^2_\textnormal{a}((\Tbb^{3})^2) \oplus \cdots
\end{equation}
The vector
\[\Omega := (1,0,0,\ldots) \in \fock\]
is called the vacuum. For $\psi = (\psi^{(0)}, \psi^{(1)}, \psi^{(2)}, \ldots) \in \fock$ and $f\in L^2(\Tbb^3)$ we define the creation operators $a^*(f)$ and the annihilation operators $a(f)$ by their actions
\begin{align*}
\left( a^*(f) \psi \right)^{(n)}(x_1,\ldots,x_n) &:= \frac{1}{\sqrt{n}} \sum_{j=1}^n (-1)^{j-1} f(x_j) \psi^{(n-1)} (x_1, \ldots, x_{j-1}, x_{j+1},\ldots,x_n)\;,\\
\left( a(f)\psi\right)^{(n)}(x_1,\ldots, x_n) &:= \sqrt{n+1} \int_{\Tbb^3} \di x \cc{f(x)} \psi^{(n+1)}(x,x_1,\ldots,x_n)\;. 
\end{align*}
Since we will work in the discrete momentum space (Fourier space) $\Zbb^3$, it is convenient to write
\[
a^*_p := a^*(f_p)\;, \quad a_p := a(f_p)\;, \quad \textnormal{where } f_p(x) = (2\pi)^{-\frac{3}{2}} e^{ip\cdot x} \textnormal{ for } p\in \mathbb{Z}^3\;. 
\] 
These operators satisfy the  \emph{canonical  anticommutator relations} (CAR)
\begin{equation}    \label{eq:car}
\{a_p,a^*_{q}\} = \delta_{p,q}\;, \quad \{a_p,a_q\} =0 = \{a^*_p,a^*_q\}\;, \qquad \forall p,q \in \Zbb^3\;. 
\end{equation}

The Hamiltonian $H_N$ in \cref{eq:HN}, originally defined on the $N$-particle sector $L^2_\textnormal{a}((\Tbb^{3})^N) \subset \fock$, can be lifted to an operator on the fermionic Fock space as 
\begin{equation}\label{eq:HcalN}
\Hcal_N = \hbar^2 \sum_{p\in \mathbb{Z}^3}   \lvert p\rvert^2 a^*_p a_p  + \frac{1}{2N} \sum_{k,p,q\in \mathbb{Z}^3} \hat V(k) a^*_{p+k} a^*_{q-k} a_{q} a_{p}\;. 
\end{equation}
Restricted to $L^2_\textnormal{a}((\Tbb^{3})^N) \subset \fock$, $\Hcal_N$ agrees with the Hamiltonian as given in \cref{eq:HN}.

\paragraph{Correlation Hamiltonian.} Now we separate the degrees of freedom described by the Slater determinant of plane waves in \cref{eq:plane-waves}  from non--trivial quantum correlations. Recall the Fermi ball  and its complement
\[
\BF := \{ p \in \Zbb^3: \lvert p\rvert \leq k_\F\}\;, \quad \BFc := \Zbb^3 \setminus \BF\;. 
\]
We define the  \emph{particle--hole transformation} $R: \fock \to \fock$ by 
\begin{equation}
R^* a^*_p R = \left\{ \begin{matrix}{} a^*_p& \textnormal{for } p \in \BFc\\a_p & \textnormal{for } p \in \BF \end{matrix}\right.\;, \qquad  R\Omega := \bigwedge_{p \in \BF} f_p\;.
\label{eq:phtrafo}
\end{equation}
This map is well--defined since vectors of the form $\prod_{j}a^*_{k_j} \Omega$ constitute a basis of $\fock$. Moreover, it is easy to verify that  $R = R^* = R^{-1}$; in particular $R$ is a unitary transformation. (In fact, $R$ is an example of a fermionic Bogoliubov transformation.)

In practice, the action of $R$ on an operator on Fock space is easily computed using the rules \cref{eq:phtrafo} and the CAR \cref{eq:car}. 
For example, consider the particle number operator
\[
\Ncal := \sum_{p\in \mathbb{Z}^3} a^*_p a_p\;.
\]
For $\psi = (\psi^{(0)}, \psi^{(1)}, \ldots) \in \fock$ we have $\Ncal\psi = (0, \psi^{(1)}, 2\psi^{(2)}, 3 \psi^{(3)}, \ldots)$; in particular $\Ncal \psi = N \psi$ is equivalent to the vector belonging to the $N$--particle sector of Fock space, $\psi \in L^2_\textnormal{a}((\Tbb^{3})^N) \subset \fock$. Now  
\begin{equation}
\begin{split}
R^* \Ncal R  &=  \sum_{h \in \BF} a_h a^*_h + \sum_{p \in \BFc} a^*_p a_p  = \sum_{h \in \BF} \left( 1 - a^*_h a_h \right) + \sum_{p \in \BFc} a^*_p a_p\\
&= N + \sum_{p \in \BFc} a^*_p a_p -   \sum_{h \in \BF} a^*_h a_h  =: N +   \Ncal^\textnormal{p} - \Ncal^\textnormal{h}\;.
\end{split}
\label{eq:Nparticle}
\end{equation}
This identity implies that if $R\psi$  is a $N$--particle state, then
\begin{equation}\label{eq:NpminusNh}
(\Ncal^\textnormal{p} - \Ncal^\textnormal{h}) \psi =0\;,
\end{equation}
namely after the transformation $R$ the number of particles is equal to the number of holes. 

The transformed Hamiltonian $R^* \Hcal_N R $ has been computed in \cite{BPS14,BPS14a,BPS14b,BPS16,BSS18}, in a slightly different way for mixed states in \cite{BJP+16}, and in the context of the correlation energy in \cite{HPR20,BNP+20}. Let us therefore just give a short sketch of the transformation of the interaction term; the transformation of the kinetic term uses \cref{eq:NpminusNh} but is otherwise very similar to \cref{eq:Nparticle}. We start by using the CAR once to write
\begin{equation}\label{eq:Hrhorho}
\frac{1}{2N} \sum_{k,p,q\in \mathbb{Z}^3} \hat{V}(k) a^*_{p+k} a^*_{q-k} a_{q} a_{p} = 
\frac{1}{2N} \sum_{k \in \Zbb^3} \hat{V}(k) \rho(k) \rho(-k) - \frac{1}{2N}\sum_{k \in \Zbb^3} \hat{V}(k) \Ncal\;,
\end{equation}
where we introduced
\[\rho(k) := \sum_{p \in \Zbb^3} a^*_{p+k} a_p \;.\]
The second summand of \cref{eq:Hrhorho} equals $-\frac{1}{2}\sum_{k\in \Zbb^3}\hat{V}(k)$, which contributes to the Hartree--Fock energy. For the transformation of the first summand one computes
\[
 R^* \rho(k) R = \mathfrak{D}(k)^* + b^*(k) + b(-k) + N \delta_{k,0}\;,
\]
where we have introduced for any $k\in \Zbb^3$ the particle--hole pair creation operator\footnote{In \cite{BNP+20} this operator was denoted by $\tilde{b}^*_k$.}
\begin{equation} \label{eq:global_pair}
b^*(k) := \sum_{p \in \BFc \cap (\BF +k)}  a^*_p a^*_{p-k}
\end{equation}
and the non--bosonizable operator
\begin{equation} \label{eq:Dkdef}
\Dfrak(k)^* := \sum_{p  \in \BFc \cap (\BFc+k)} a^*_{p} a_{p-k} - \sum_{h  \in \BF \cap (\BF-k)} a^*_{h} a_{h+k} \;.
\end{equation} 
Note that $\Dfrak(k)^* = \Dfrak(-k)$ and $\Dfrak(0)^* = \Ncal^\textnormal{p} - \Ncal^\textnormal{h}$. Observing that the constant terms (i.\,e., not containing any creation or annihilation operator) contribute to the Hartree--Fock energy $E^\textnormal{HF}_N$ and collecting all quadratic terms in the operator $\Xbb$, we arrive at the result
% \footnote{We prefer slightly different conventions in writing the different terms: $Q_\B$ is written in normal order by moving $b^*$--operators to the left of $b$--operators, which gives rise to a commutator of which the quadratic part becomes $\Xbb$ and the constant part contributes to the Hartree--Fock energy. The term $\Ecal_2$ is brought into normal order by moving the $b$--operators to the right and the $b^*$--operators to the left of the $\Dfrak$-- and $\Dfrak^*$--operators (it is convenient to notice that $[b(k),\Dfrak(-k)]=0$). The term $\Ecal_1$ is not normal ordered but rather written as a square, so that it is manifestly positive.}
\begin{equation}\label{eq:Hcorr} 
\Hcal_\textnormal{corr}:= R^* \Hcal_N R - E^\textnormal{HF}_N =  \Hbb_0 + Q_\B + \mathcal{E}_1 + \mathcal{E}_2 + \Xbb
\end{equation}
where the summands are given by 
\begin{align*} 
\Hbb_0 &:= \sum_{k \in \Zbb^3} e(k) a^*_k a_k \quad \textnormal{with dispersion relation } e(k) := \lvert \hbar^2 \lvert k\rvert^2 - \kappaf^2 \rvert\;, \tagg{kinen} \displaybreak[0] \\
  Q_\B & := \frac{1}{N} \sum_{k \in \north } \hat{V}(k) \Big[  b^*(k) b(k) + b^*(-k) b(-k) + b^*(k) b^*(-k) + b(-k) b(k) \Big]\;, \displaybreak[0]\\%\label{eq:QBN} \\
 \mathcal{E}_1 & := \frac{1}{2N} \sum_{k \in \north } \hat{V}(k) \Big[ \Dfrak(k)^* \Dfrak(k) + \Dfrak(-k)^* \Dfrak(-k)\Big]\;, \displaybreak[0]\\
  \mathcal{E}_2 & := \frac{1}{N} \sum_{k \in \north } \hat{V}(k) \Big[  \Dfrak(k)^* b(k)  + \Dfrak(-k)^* b(-k)  + \hc  \Big]\;,    \displaybreak[0]\\
\Xbb &:= - \frac{1}{2N} \sum_{k \in \Zbb^3} \hat{V}(k) \bigg[ \sum_{p \in \BFc \cap (\BF +k)} a^*_p a_p  + \sum_{h \in \BF \cap (\BFc-k)} a^*_{h} a_{h}\bigg]\;.
\end{align*}
Note that we have introduced the set $\north $ of all momenta $k=(k_1,k_2,k_3)$ in $\Zbb^3 \cap  \supp \hat V$ satisfying 
\[
k_3> 0 \text{ or } (k_3=0 \text{ and } k_2>0) \text{ or } (k_2=k_3=0 \text{ and } k_1>0)\;. 
\]
This set is chosen such that 
\[\north  \cap (- \north ) =\emptyset, \quad \north  \cup (- \north ) =\Big( \mathbb{Z}^3 \cap \supp \hat V \Big) \setminus \{0\}\;.\]
 The term $Q_\B$ is the bosonizable part of the interaction and contains only the pair operators. The term $\mathcal{E}_1$ is purely non--bosonizable and $\mathcal{E}_2$ couples bosonizable and non--bosonizable excitations. Note that unlike the other terms $\Ecal_1$ is not normal--ordered (this choice is made so that we have $\Ecal_1 \geq 0$); for this reason $\Xbb$ and $\mathcal{E}_1$ differ slightly from the  expressions given in \cite{BNP+20}.
 
 Since $\Xbb$ is quadratic in fermionic operators, it can be easily bounded using $\Ncal/N$, which will be seen to have expectation value much smaller than the order $\hbar$ of $E_N^\textnormal{RPA}$.
 
 In \cite{BNP+20}, it was proved that $\Hbb_0+Q_\B$ evaluated in a trial state of quasi--free particle--hole pairs gives rise to $E_N^\textnormal{RPA}$ as an upper bound to the correlation energy. Accordingly, an important part of our task will be to prove that the contribution from $\mathcal{E}_1+\mathcal{E}_2$  is negligible. (Whereas this was easily achieved for the upper bound using the explicit form of the trial state, for the lower bound it actually turns out to be a major challenge.)

The rest of the paper is devoted to the proof of the inequality 
\begin{align} \label{eq:Hcorr-lb}
\inf_{\substack{\psi\in \fock\colon \norm{\psi}=1\,,\\(\Ncal^\textnormal{p} - \Ncal^\textnormal{h})\psi=0}} \langle \psi, \Hcal_\textnormal{corr} \psi\rangle \geq E_N^\textnormal{RPA} + \Ocal(\hbar^{1+\frac{1}{16}}) \;.
\end{align}
Thanks to \cref{eq:Hcorr} it directly implies the main result, the lower bound in \cref{thm:main}.
% Note that the constraint  $(\Ncal^\textnormal{p} - \Ncal^\textnormal{h})\psi=0$ is not required in the proof of the lower bound \cref{eq:Hcorr-lb}, but recall that it was used in obtaining the formula for $\Hcorr$.  

In the following we explain the key estimates in our proof. We use the symbol $C$ for positive constants that may change from line to line, but are independent of $N$, $\hbar$, and $M$ (the number of patches, to be introduced in \cref{eq:defM}). The constants $C$ may depend on the momentum $k$, which does not play a role ultimately since we only consider the finitely many $k \in \supp\hat{V}$, i.\,e., we can always take the maximum and so treat all constants as independent of $k$. We generally absorb any dependence on $\hat{V}$ in the constants $C$; \emph{we only write the $\hat{V}$--dependence of estimates explicitly where the smallness condition on $\norm{\hat{V}}_{\ell^1}$ plays a role}.

\paragraph{A priori estimates.} Similarly to \cite{BNP+20,HPR20}, many approximations used in our approach are based on the idea that the relevant quantum states have only few excitations. For the upper bound in \cite{BNP+20}, this fact is easily justified by the strong bound $\langle \Psi_\textnormal{trial}, \Ncal^m \Psi_\textnormal{trial}\rangle \leq C_m$ (for all $m \in \Nbb$) for the trial state used to compute the expectation value of $\Hcorr$.  Compared to that bound, for the ground state we can only derive weaker estimates. In \cref{lem:N-H0} we prove that the particle number operator can be controlled by the kinetic energy (i.\,e., the kinetic energy operator has a tiny gap, of order $\hbar^2$) by 
\begin{align}\label{eq:KEY-01}
 \Ncal \leq 2 N^{\frac{2}{3}} \Hbb_0\;.  
\end{align}
% so that the expectation value of $\Ncal$ in the ground state is found to be of order $N^{\frac{1}{3}}$.
% This is proved using the classic result \cref{eq:classic-lattice} on the counting of lattice points on spheres.  
% 
To avoid the particle number operator, where possible we bound pair operators directly by the kinetic energy, using an inequality from \cite{HPR20},
\begin{align}\label{eq:KEY-00}
\sum_{p \in \BFc \cap (\BF +k)} \norm{a_p a_{p-k} \psi} \leq CN^{\frac{1}{2}} \norm{\Hbb_0^{1/2} \psi}\;, \qquad \forall \psi \in \fock\;.
\end{align}
(The idea of directly using the
kinetic energy for bounds has appeared already in \cite{Hai03, HHS05} in the context of rigorous second order perturbation theory.)
% This way, the kinetic energy can be bounded by the total Hamiltonian, and together with a trivial upper bound (with the plane--wave Slater determinant as trial state) we conclude that the expectation value of $\Hbb_0$ in the ground state is of order $\hbar$.
The bounds \cref{eq:KEY-00} and \cref{eq:KEY-01} imply the rough estimates in \cref{cor:kineticenergy}, as in \cite{HPR20}:
\begin{align}\label{eq:KEY-02}
\frac{1}{2} (\Hbb_0 +\mathcal{E}_1)  -  \hbar \leq \Hcal_\textnormal{corr}  \leq 2 (\Hbb_0 +\mathcal{E}_1+ \hbar)\;.
\end{align}
Together with an upper bound of order $\hbar$ such as the trivial variational one obtained using the trial state $\Omega$ (corresponding to the Slater determinant of plane waves before the particle--hole transformation),  
% proven in \cite{BNP+20,HPR20}
this implies that the ground state $\psi_\textnormal{gs}$ of $\Hcal_\textnormal{corr}$, the minimizer of the expectation value on the left hand side of \cref{eq:Hcorr-lb}, satisfies  
\begin{align}\label{eq:KEY-03}
\langle \psi_\textnormal{gs}, (\Hbb_0+\mathcal{E}_1) \psi_\textnormal{gs} \rangle \leq C\hbar\;, \quad \langle \psi_\textnormal{gs}, \Ncal \psi_\textnormal{gs} \rangle \leq CN^{\frac{1}{3}}\;. 
\end{align}

For technical reasons, we will also need to control the expectation of higher powers of $\Ncal$, which does not follow from  \cref{eq:KEY-00} and \cref{eq:KEY-01}. To overcome this difficulty, in \cref{lem:loc} we replace the ground state $\psi_\textnormal{gs}$ by an approximate ground state $\Psi$ satisfying
\begin{align}\label{eq:KEY-04}
 \langle \Psi, (\Hbb_0 +\mathcal{E}_1) \Psi \rangle \leq C\hbar\;, \quad \Psi= \mathds{1} (\Ncal \leq CN^{\frac{1}{3}})\Psi
\end{align}
while its energy is still close to the ground state energy, i.\,e.,
\[ 
\langle \psi_\textnormal{gs}, \Hcal_\textnormal{corr} \psi_\textnormal{gs}\rangle \geq \langle \Psi, \Hcal_\textnormal{corr} \Psi \rangle - CN^{-1}\;. 
\] 
This is achieved by using the technique of localizing particle number on Fock space, which goes back to Lieb and Solovej \cite{LS01}. In the proof we will use the  formulation from \cite[Proposition 6.1]{LNSS15}. It is the state $\Psi$ that most of our subsequent analysis will be applied to. 

% The construction of such a state $\Psi$ is an important step for our analysis of the energy lower bound. We do not need that for the energy upper bound in \cite{BNP+20} as the trial state there satisfies much stronger a-priori estimates. 

\paragraph{Approximately bosonic creation operators.} When applied to states with few excitations, the pair creation operators behave approximately as \emph{bosonic} creation operators, namely we have to leading order the \emph{canonical commutator relations} (CCR)  
\begin{equation} \label{eq:CCR-bkbl}
[b^*(k),b^*(l)] = 0, \quad  [b(k),b^*(l)] \simeq \delta_{k,l} \times \textnormal{const}\;, \quad \forall k,l \in \Zbb^3\;.
\end{equation}

Unfortunately there is no expression for the kinetic energy $\Hbb_0$ in terms of the $b^\natural(k)$--operators\footnote{The symbol $\natural$ may stand both for ``$*$'' (adjoint in Fock space $\fock$) and for absence of ``$*$''; we use it whenever the choice does not play a role.}. We take inspiration from the solution of the Luttinger model \cite{ML65}: if the dispersion relation were linear, the $b^*(k)$ would create eigenvectors of $\Hbb_0$. Since the dispersion relation $\hbar^2 \lvert k\rvert^2$ is not linear, we will linearize it locally. This is achieved by localizing the creation operators to patches on the Fermi surface. More precisely, we cut the shell of width $R_{\hat{V}} := \diam\supp\hat{V}$ around the Fermi surface into patches $\{B_\alpha\}_{\alpha=1}^M$. The construction of the patches is recalled in  \cref{sec:patches}. As discussed in the introduction, under the name of ``sectors'', this idea has already been employed in the rigorous renormalization group context. 

We consider the pair excitations supported in each patch\footnote{Where confusion may arise, we use the notation $p\colon p \in \BFc \cap B_\alpha,\,p-k\in \BF \cap B_\alpha$ in specifying the range of summation: here it is over all $p \in \Zbb^3$ (but not over $k$) satisfying $p \in \BFc \cap B_\alpha$ and $p-k\in \BF \cap B_\alpha$.} 
\begin{equation}    \label{eq:def_bos_exc-bb}
b^*_\alpha(k) := \frac{1}{m_\alpha(k)} \sum_{\substack{p\colon p \in \BFc \cap B_\alpha\\p-k\in \BF \cap B_\alpha}}  a^*_p a^*_{p-k} \;.
\end{equation}
To normalize the constant in the approximate CCR, the normalization constant $m_\alpha(k)$ should be chosen such that $\norm{b^*_\alpha(k)\Omega} =1$, namely
\begin{equation}    \label{eq:def_bos_exc-mm}
m_\alpha^2(k) = \sum_{\substack{p\colon p \in \BFc \cap B_\alpha\\p-k\in \BF \cap B_\alpha}} 1\;.
\end{equation}
This has the meaning of the number of particle--hole pairs $(p,h) \in \BFc \times \BF$ inside the patch $B_\alpha$ with relative momentum $p-h=k$.
However, this number may be zero! In fact, if $k\cdot \hat{\omega}_\alpha <0$ with $\hat{\omega}_\alpha$ the unit vector pointing in the direction of the patch $B_\alpha$, then a simple geometric consideration shows that the summation domain in \cref{eq:def_bos_exc-mm} and \cref{eq:def_bos_exc-bb} is empty (the condition $k\cdot \hat{\omega}_\alpha <0$ is incompatible with $p \in \BFc$ and $p-k \in \BF$). The same problem occurs for $m_\alpha^2(-k)=0$ if $k\cdot \hat{\omega}_\alpha >0$. 

Furthermore, as suggested by \cite[Chapters 8, 9.2.3, and 9.2.4]{RS05} and \cite{CS16}, bosonization is expected to be a good approximation only if $m_\alpha(k)$ is large. This cannot be ensured for patches where $k \cdot \hat{\omega}_\alpha \approx 0$ (if we think of the direction of $k$ as defining the north pole of the Fermi ball, these are the patches near the equator). However, the momentum $k$ of such excitations is almost tangential to the Fermi surface and thus their energy is very low. In fact, we will be able to show that their contribution to the ground state energy is small and exclude them from the bosonization. To do so, we introduce a cut--off near the equator by defining the index subset $\Ik= \Ikp \cup \Ikm$ where 
\begin{equation}\label{eq:cutoff} 
\begin{split}
\Ikp & := \left\{ \alpha \in \{1,2,\ldots, M\} : k \cdot \hat{\omega}_\alpha \geq N^{-\delta} \right\}\;, \\ \Ikm & := \left\{ \alpha \in \{1,2,\ldots, M\} : k\cdot \hat{\omega}_\alpha \leq - N^{-\delta} \right\}\;.
\end{split}
\end{equation}
We will choose the cut--off parameter $\delta$ and the number of the patches $M$ such that
\begin{equation}
 N^{2\delta} \ll M \ll N^{\frac{2}{3}-2\delta}\;, \quad 0<\delta <\frac{1}{6}\;. \label{eq:defM}
\end{equation}
(Eventually we will choose $M = N^{4\delta}$ and $\delta=\frac{1}{24}$.) Note that unlike \cite{BNP+20} where we require $M\gg N^{\frac 1 3}$, here we allow a much smaller value of $M$, which is important to control the error terms due to the Bogoliubov transformation introduced later. 

Then by  \cite[Proposition~3.1]{BNP+20}, the constant
\[
n_\alpha(k):=\left\{ \begin{array}{cc} m_\alpha(k) & \textnormal{for } \alpha\in \Ikp \\
m_\alpha(-k) & \textnormal{for } \alpha\in \Ikm \end{array}\right.
\]
can be computed to be given by
\begin{align} \label{eq:n-alpha-intro}
n_\alpha(k)^2  =  \frac{4\pi k_\F^2}{M} \lvert k \cdot \hat\omega_\alpha \rvert  \left( 1 + o(1) \right) \gg 1\;. 
\end{align}
(Heuristically, the reader may think of the number of particle--hole pairs as given by the surface area of the patch, $4\pi k_\F^2/M$, times the depth inside the Fermi ball that can be reached by $h$, namely $\lvert k \cdot \hat\omega_\alpha \rvert$. For this counting argument to be justifiable, the diameter of a patch on the Fermi surface may not become too large, requiring $M \gg N^{2\delta}$.)
Consequently, the operators 
\begin{align} \label{eq:c-b-intro}
c^*_\alpha(k) := \left\{ \begin{array}{cc} b^*_\alpha(k) & \textnormal{for }\alpha\in \Ikp\\
b^*_\alpha(-k) & \textnormal{for } \alpha \in \Ikm \end{array}\right.
\end{align}
are well--defined and behave like bosonic creation operators, namely  
\begin{equation} \label{eq:CCR-ckcl}
[c_\alpha^*(k),c_\beta^*(l)] = 0\;, \quad  [c_\alpha(k),c_\beta^*(l)] \simeq \delta_{\alpha,\beta}\delta_{k,l}\;, \quad \forall k,l\in \north,\ \alpha\in \Ik,\ \beta\in \Il\;.
\end{equation}
This is proven in \cref{lem:ccr}, which is a slight extension of \cite[Lemma 4.1]{BNP+20}. 

\paragraph{Gapped Number Operator.}
As we have seen in \cref{eq:KEY-03} we do not have strong control on the particle number operator, due to the possibility of having many small--energy excitations near the Fermi surface; a problem which in the beginning is avoided by directly using $\Hbb_0$ for bounds. However, a serious problem of using $\Hbb_0$ is that it is  \emph{not stable} under the Bogoliubov transformation that we will later introduce to approximately diagonalize the effective Hamiltonian. A way of overcoming this problem, and a key improvement compared to \cite{BNP+20} is that instead of using the full fermionic number operator $\Ncal$ to control error terms, wherever possible we use only the \emph{gapped number operator}   
\begin{equation}
\label{eq:gappedN}
\Ncal_\delta := \sum_{i \in \Zbb^3\colon e(i) \geq \frac{1}{4} N^{-\frac{1}{3} -\delta}} a^*_i a_i\;, 
\end{equation}
which does not count low--energy excitations. 
Here we have used the dispersion relation $e(i)=\lvert \hbar^2 \lvert i\rvert^2 - \kappa^2 \rvert$ introduced in \cref{eq:kinen}, and due to the artificial gap we obtain
\[
\Ncal_\delta\leq N^{\frac{1}{3}+\delta} \Hbb_0\;.
\]
Therefore, \cref{eq:KEY-04} implies that $\langle \Psi, \Ncal_\delta \Psi\rangle \leq CN^\delta$ which is much better than $\langle \Psi, \Ncal \Psi\rangle \leq CN^{\frac{1}{3}}$ in \cref{eq:KEY-03}. Thus in practice, controlling error terms by using $\Ncal_\delta$ is as good as using the kinetic operator $\Hbb_0$. Furthermore, unlike $\Hbb_0$, the gapped number operator $\Ncal_\delta$ is \emph{stable} under the Bogoliubov transformation (see \cref{lem:stability}).

The main instance where $\Ncal_\delta$ finds use is \cref{lem:bosonic-number}, where we bound the approximately bosonic number operator by the fermionic gapped number operator, \begin{align}\label{eq:KEY-08}
\sum_{\alpha\in \Ik} c_\alpha^*(k) c_\alpha(k) \leq C \Ncal_\delta\;. 
\end{align}
This improves \cite[Lemma~4.2]{BNP+20}, where $\Ncal$ was used as the bound. The key insight leading to this improvement is that only bosonic pair operators with $\alpha \in \Ik$ are needed in the effective Hamiltonian \cref{eq:heff-intro} and the diagonalizing Bogoliubov transformation \cref{eq:gs_exact-intro} to obtain the RPA energy \cref{eq:rpa_as_gs}. Since $\alpha \in \Ik$ means $\lvert k\cdot \hat{\omega}_\alpha \rvert \geq N^{-\delta}$, the relative momentum $k$ between particles $p$ and holes $h = p-k$ cannot be tangential to the Fermi surface; i.\,e., $p $ or $h$ (or both) has to lie above the gap $e(i) \geq \frac{1}{4}N^{-\frac{1}{3}-\delta}$. 
This is the reason for the same parameter $\delta > 0$ appearing both in the gapped number operator and in the equator cut--off \cref{eq:cutoff}. The new bound allows us to work with the bosonic pairs at the energy scale relevant for the result, while keeping them as much as possible separate from the low--energy excitations on whose number we do not have strong control. 

In the next steps, we will write the correlation Hamiltonian $\Hcal_\textnormal{corr}$ as a quadratic Hamiltonian in terms of the approximately bosonic operators  $c_\alpha^*(k)$ and $c_\alpha(k)$. 

\paragraph{Bosonization of the interaction energy.} By decomposing 
\begin{align}\label{eq:bR-c-intro}
b(k) \simeq \sum_{\alpha\in \Ikp} n_\alpha(k) c_\alpha(k)\;, \quad b(-k)  \simeq \sum_{\alpha\in \Ikm} n_\alpha(k) c_\alpha(k)
\end{align}
we can write the main interaction term as 
\begin{align}\label{eq:Q-dec-intro}
Q_\B &\simeq    \frac{1}{N} \sum_{k \in \north } \hat{V}(k) \Big[  \sum_{\alpha,\beta \in \Ikp} n_\alpha(k) n_\beta(k) c^*_\alpha(k) c_\beta(k)  + \sum_{\alpha,\beta \in \Ikm} n_\alpha(k) n_\beta(k) c^*_\alpha(k) c_\beta(k) \\
&\qquad\qquad \qquad  + \sum_{\alpha \in \Ikp,\,\beta \in \Ikm} n_\alpha(k) n_\beta(k) c^*_\alpha(k) c^*_\beta(k) + \sum_{\alpha \in \Ikp,\,\beta \in \Ikm} n_\alpha(k) n_\beta(k) c_\beta(k) c_\alpha(k)  \Big]\;. \nonumber
\end{align}
In the approximation \cref{eq:Q-dec-intro} we have ignored all excitations outside the patches. It is justified in \cref{lem:remove-corridors}, where we prove that
\begin{align}\label{eq:KEY-06a}
Q_\B + \mathcal{E}_2 - Q_\B^{\Rcal} - \mathcal{E}_2^{\Rcal}   \geq - C\left( N^{ - \frac{\delta}{2}} + C N^{ -\frac{1}{6} + \frac{\delta}{2} } M^{\frac{1}{4}}\right) \big(\Hbb_0+\mathcal{E}_1+\hbar\big)
 \end{align}
where $Q_\B^{\Rcal}+\mathcal{E}_2^{\Rcal}$ is similar to $Q_\B + \mathcal{E}_2$ but contains only pair excitations in the patches. The proof of \eqref{eq:KEY-06a} requires an improved version of the kinetic inequality \eqref{eq:KEY-00} (see \cref{lem:kineticequator}). Thanks to  \cref{eq:KEY-04}, the error term in \cref{eq:KEY-06a} does not contribute to the leading order of the correlation energy.

Note that the bound \eqref{eq:KEY-06a} is not necessary for the upper bound in \cite{BNP+20} because the trial state there is constructed to contain only pair excitations inside the patches, so that the expectation value of a pair not belonging completely to relevant patches is identically zero.

\paragraph{Bosonization of the kinetic energy.} The bosonization of the fermionic kinetic energy is more complicated. A key observation is that if $\alpha\in \Ikp$, then using the CAR \cref{eq:car} and linearizing the dispersion relation around $k_\F \hat{\omega}_\alpha$, we find 
 \begin{align} \label{eq:lin-intro}
 [\Hbb_0,c^*_{\alpha}(k)] & =  \Big[ \sum_{i\in \mathbb{Z}^3} e(i) a_i^* a_i, \frac{1}{n_{\alpha}(k)} \sum_{\substack{p\colon p\in \BFc \cap B_\alpha\\p-k \in \BF \cap B_\alpha}} a^*_p a^*_{p-k} \Big] \nonumber\\
   & = \frac{1}{n_{\alpha}(k)}   \sum_{\substack{p\colon p\in \BFc \cap B_\alpha\\p-k \in \BF \cap B_\alpha}} (e(p)+e(p-k))a^*_p a^*_{p-k}\nonumber\\
  & \simeq 2\hbar \kappaf \lvert k\cdot \hat{\omega}_\alpha\rvert c_\alpha^*(k)\;.
  \end{align} 
For linearizing the dispersion relation we used the fact that for any $p\in \BFc \cap (\BF+k)\cap B_\alpha$, since $\diam(B_\alpha)\ll k_\F/\sqrt{M}$ we have
\begin{align} \label{eq:ep-epk-intro}
e(p)+e(p-k)= \hbar^2 (2p-k)\cdot k \simeq \hbar^2 (2 k_\F \hat{\omega}_\alpha) \cdot k = 2\hbar \kappaf \lvert k\cdot \hat{\omega}_\alpha\rvert\;. 
\end{align}
Obviously the same holds if $\alpha\in \Ikm$. Therefore, within commutators with pair operators, $\Hbb_0$ can be approximated as in the Luttinger model \cite{ML65} by independent modes (i.\,e., harmonic oscillators) of energies $\hbar \kappaf2 k\cdot \hat{\omega}_\alpha$, namely 
\begin{align} \label{eq:H0-D0-intro}
\Hbb_0 \simeq  2\kappaf \hbar \sum_{k \in \north }  \sum_{\alpha=1}^M  \lvert k\cdot \hat{\omega}_\alpha\rvert c^*_{\alpha} (k) c_\alpha(k)=:\Dbb_\B\;.
\end{align}

A key idea of our analysis is to justify \cref{eq:H0-D0-intro} not by estimating the difference $\Hbb_0-\Dbb_\B$ directly, but rather by proving that it is essentially invariant under the approximate Bogoliubov transformation $T$  which we will introduce below to diagonalize the quadratic bosonized Hamiltonian. More precisely, in \cref{lem:lin-2} we show that with $\psi := T^* \Psi$ we have
\begin{align}\label{eq:KEY-10}
&  \langle \Psi, (\Hbb_0 - \Dbb_\B) \Psi \rangle = \langle  \psi, (\Hbb_0 - \Dbb_\B) \psi \rangle  + \textnormal{error}
\end{align}
where
\begin{equation}\label{eq:bnd}
\begin{split}
\lvert \textnormal{error}\rvert & \leq C\hbar \Big[ M^{-\frac{1}{2}} \norm{  (\Ncal_\delta+1)^{1/2}  \psi  }^2 \\
& \qquad \quad + CM^{\frac{3}{2}}N^{-\frac{2}{3}+\delta}  \norm{ (\Ncal_\delta+1)^{1/2} (\Ncal+1) \psi } \norm{ (\Ncal_\delta+1)^{1/2}   \psi  } \Big]\;.
\end{split}
\end{equation}
Note that only here, in the first error summand, due to the linearization of $\Hbb_0$, does $M$ enter in the denominator. With $\langle \psi, \Ncal_\delta \psi \rangle \leq C N^\delta$ (this bound is stable under the Bogoliubov transformation),
we need to take $M\gg N^{2\delta}$. We will eventually choose  $M=N^{4\delta}$.

The bound \cref{eq:KEY-10} is a crucial improvement over the linearization technique in  \cite{BNP+20} which requires $M\gg N^{\frac{1}{3}}$, a condition that we cannot fulfill due to the second error summand in \cref{eq:bnd} (recall that in our approximate ground state we only know $\Ncal \leq C N^{\frac{1}{3}}$). This improvement is achieved because in \cite{BNP+20} we unnecessarily linearized the expectation value of $\Hbb_0$, whereas in the present paper we only linearize the necessary commutator with a pair operator $c^*_\alpha(k)$. In general, this new possibility of choosing a rather small $M$ means that we gain flexibility in the technical steps because we can afford arbitrarily high powers of $M$ as long as there is a negative power of $N$.

To apply \cref{eq:KEY-10}, prior to using the Bogoliubov transformation, we will decompose 
\begin{align}\label{eq:decomp}
\Hbb_0=(\Hbb_0-\Dbb_\B) + \Dbb_\B\;.
\end{align}
% the term $+\Dbb_\B$ becomes part of the bosonic quadratic operator \cref{eq:heffeff} which will be diagonalized by the Bogoliubov transformation. Under the approximate Bogoliubov transformation the difference $\Hbb_0-\Dbb_\B$ remains essentially unchanged as we just saw in \cref{eq:KEY-10}. Then \emph{after the transformation} we will use the diagonalized bosonic quadratic operator to control the negative contribution $-\Dbb_\B$, up to a loss of $-C\norm{\hat{V}}_{\ell^1} \Hbb_0$ (see \cref{eq:KEY-12}). Thanks to the smallness assumption on the potential, this leaves us the major part of  the fermionic kinetic energy, $(1-C\norm{\hat{V}}_{\ell^1})\Hbb_0 \geq 0$, to be used for a lower bound on the non--bosonizable parts $\Ecal_1+\Ecal_2^{\Rcal}$ of the Hamiltonian later on.  

\paragraph{Diagonalization of the bosonized Hamiltonian.} By combining the approximation \cref{eq:Q-dec-intro} and the operator $+\Dbb_B$ from \cref{eq:decomp}, we find the effective quadratic bosonic Hamiltonian  
\begin{equation}\label{eq:heffeff}
\Dbb_\B + Q_\B^\Rcal =  \sum_{k\in \north} 2\hbar \kappaf \lvert k\rvert h_\textnormal{eff}(k)
\end{equation}
with 
\begin{align} \label{eq:heff-intro}
 h_\textnormal{eff}(k) & := \frac{1}{\lvert k\rvert}\sum_{\alpha \in\Ical_k} \lvert k \cdot \hat{\omega}_\alpha\rvert  c^*_{\alpha} (k) c_{\alpha} (k)  \\
&\quad\; + \frac{\hat V(k)}{2\hbar \kappaf \lvert k\rvert N}  \Big[  \sum_{\alpha,\beta \in \Ikp} n_\alpha(k) n_\beta(k) c^*_\alpha(k) c_\beta(k)  + \sum_{\alpha,\beta \in \Ikm} n_\alpha(k) n_\beta(k) c^*_\alpha(k) c_\beta(k) \nonumber\\
&\qquad\qquad \quad\  + \sum_{\alpha \in \Ikp,\,\beta \in \Ikm} n_\alpha(k) n_\beta(k) c^*_\alpha(k) c^*_\beta(k) + \sum_{\alpha \in \Ikp,\,\beta \in \Ikm} n_\alpha(k) n_\beta(k) c_\beta(k) c_\alpha(k)  \Big]\;.\nonumber
\end{align}
We have arrived at an effective quadratic Hamiltonian in terms of the approximately bosonic creation and annihilation operators. If the effective Hamiltonian were exactly bosonic, it could be diagonalized by a Bogoliubov transformation \cite{Bog47}. While we do not have this tool available since our operators are not exactly bosonic, we can still use the explicit formula as for a true Bogoliubov transformation and define the unitary map 
\begin{equation}\label{eq:gs_exact-intro}
T=\exp \Big( \sum_{k\in \north}\frac{1}{2}\sum_{\alpha,\beta\in\Ik} K(k)_{\alpha,\beta} c^*_\alpha(k) c^*_\beta(k) - \hc \Big) 
\end{equation}
where the real symmetric matrices $K(k)$ are computed as in the exactly bosonic case. The choice of $K(k)$ is the same as in \cite{BNP+20}, following the abstract formulation given in \cite{GS13}. We will quickly recall it in Section \ref{sec:exa-diag}. 

Another key aspect of our proof is the observation that the Bogoliubov kernel $K(k)$ satisfies a refined entry--wise bound,
\begin{equation}\label{eq:Kbound-intro}
\lvert K(k)_{\alpha,\beta}\rvert \leq \frac{C}{M} \min \left\{ \frac{n_\alpha(k)}{n_\beta(k)},\frac{n_\beta(k)}{n_\alpha(k)}\right\} \quad \textnormal{for all }k\in \north \text{ and }\alpha,\beta \in \Ik\;.
\end{equation}
This is proved in \cref{lem:K}. An important role in the proof is played by the fact that due to the geometry of the Fermi surface the normalization factor $n_\alpha(k)^2$ is proportional to $\lvert k\cdot\hat{\omega}_\alpha\rvert$ (see \cref{eq:n-alpha-intro}) which is also the linearization of the dispersion relation (see \cref{eq:H0-D0-intro}), leading to cancellations. This means that as the gap of the kinetic energy closes when we consider particle--hole pairs that are almost tangential to the Fermi surface, the energy gain due to the interaction of such an excitation vanishes at the same rate. While the proof is essentially a detailed computation, it is crucial in controlling the non--bosonizable terms $\Ecal_2$, see \cref{eq:ch-sh-sum}.

In \cref{lem:Bog-T}, we show that $T$ acts approximately as a bosonic Bogoliubov transformation, namely
\begin{align}\label{eq:KEY-09}
T^*_\lambda c_\gamma(l) T_\lambda &= \sum_{\alpha\in\Il}\cosh(\lambda K(l))_{\alpha,\gamma} c_\alpha(l) + \sum_{\alpha\in\Il} \sinh(\lambda K(l))_{\alpha,\gamma} c^*_\alpha(l) + \mathfrak{E}_\gamma(\lambda,l)
\end{align}
where the error operators  satisfy 
\begin{equation}\label{eq:KEY-09bound}
\sum_{\gamma\in \Il} \norm{ \mathfrak{E}_\gamma(\lambda,l) \psi}  \leq C M N^{-\frac{2}{3}+\delta}  \norm{ (\Ncal_\delta+M)^{1/2} (\Ncal +1)\psi} \qquad \forall \psi \in \fock\;. 
\end{equation}
This is an improvement of \cite[Prop 4.4]{BNP+20} in that we replaced some $\Ncal$ by $\Ncal_\delta$. In order to put the error estimate \cref{eq:KEY-09bound} in good use, we need also that the particle number operators be stable under the approximate Bogoliubov transformation; this is the content of \cref{lem:stability}, based on a refinement of the Gr\"onwall argument in \cite{BPS14,BNP+20}. 

To diagonalize the bosonizable part of the Hamiltonian, we insert \cref{eq:KEY-09} in $T^*(\Dbb_0+Q_\B^\Rcal)T$ and write the transformed expression in Wick--normal order (with respect to the approximately bosonic operators). Up to a small error, this produces the ground state energy as desired,
\begin{equation}\label{eq:rpa_as_gs}
\inf \textnormal{spec} \left( \sum_{k\in \north} 2\hbar \kappaf \lvert k\rvert h_\textnormal{eff}(k) \right) = E_N^\textnormal{RPA} + o(\hbar)\;. 
\end{equation}
Additionally we obtain the (up to a one--particle unitary) diagonalized quadratic Hamiltonian which in exact Bogoliubov theory would be the excitation spectrum; for some explicit matrix $\Kfrak(k)_{\alpha,\beta}$ it has the form
\[
\sum_{k\in \north} 2\kappaf \hbar \lvert k\rvert \sum_{\alpha,\beta \in \Ik} \Kfrak(k)_{\alpha,\beta} c^*_\alpha(k)c_\beta(k)\;.
\]
As mentioned before, a further new idea of our proof is that we sacrifice the positive contribution of the excitation spectrum to control the negative term $-\Dbb_\B$ left from the comparison of the fermionic and bosonic kinetic energy \cref{eq:KEY-10}. In fact, we will prove that (see \cref{eq:kinetic-Dbb0-final}) 
\begin{align}\label{eq:KEY-12}
\sum_{k\in \north} 2\kappaf \hbar \lvert k\rvert \sum_{\alpha,\beta\in \Ik} \Kfrak(k)_{\alpha,\beta}  c^*_\alpha(k)c_\beta(k) \geq \Dbb_\B - C \norm{\hat V}_{\ell^1} \Hbb_0\;.
\end{align} 
The proof of \cref{eq:KEY-12} is based on an explicit computation of the operator $\Kfrak(k)$ and the nice property \cref{eq:Kbound-intro} of the Bogoliubov kernel. 

When $\norm{\hat V}_{\ell^1}$ is small, the error term $-\norm{\hat V}_{\ell^1} \Hbb_0$ in \cref{eq:KEY-12} is controlled by the positive term $\Hbb_0$ left from the comparison of the fermionic and bosonic kinetic energy \cref{eq:KEY-10}. 

\paragraph{Controlling non--bosonizable parts of the Hamiltonian.} We still have to show that the non--bosonizable terms $\Ecal_1+\Ecal_2^{\Rcal}$ have only a small effect on the ground state energy. As explained in \cite{BNP+20} these error terms can be easily controlled by $\Ncal^2/N$. In the trial state in \cite{BNP+20}, the expectation value of $\Ncal^2$ is of order $1$, so that $\Ncal^2/N$ is a small error. For the lower bound however, in the (approximate) ground state we only know that $\Ncal$ is of order $N^{\frac{1}{3}}$, so that $\Ncal^2/N$ would be of the same order $\hbar=N^{-\frac{1}{3}}$ as the correlation energy. Another way to see the difficulty in dealing with these terms is to observe that $\Ecal_2^\Rcal$ couples the ``good'' bosonic degrees of freedom with ``bad'' uncontrolled fermions near the Fermi surface (the latter were, by construction, absent in the trial state used for the upper bound).   

Thus the non--bosonizable parts require a subtle analysis. The following argument relies on the fact that $\Ecal_1$ is non--negative (as $\hat{V}(k) \geq 0$), which helps us in obtaining a lower bound for $\Ecal_2^\Rcal$. By the Cauchy--Schwarz inequality and the kinetic energy estimate \cref{eq:KEY-00}, it is easy to see that 
\[
\Ecal_1+\Ecal_2^\Rcal \geq - C \norm{\hat V}_{\ell^1} \Hbb_0\;. 
\]
Of course, this bound is useless because $\Hbb_0$ is of the same order as $\Hcal_\textnormal{corr}$. However, we are able to rescue this idea by proving a similar lower bound for the \emph{transformed} operator $T^*(\Ecal_1+\Ecal_2^\Rcal)T$. In fact, in \cref{lem:non--bosonizable} we prove that, with $\psi=T^*\Psi$, 
\begin{align}\label{eq:KEY-11}
 \langle \Psi, (\Ecal_1+ \Ecal_2^\Rcal) \Psi \rangle &\geq - C \norm{\hat V}_{\ell^1} \norm{ \Hbb_0^{1/2} \psi}^2 - CN^{-\frac{1}{2}} \norm{\psi}\norm{\Hbb_0^{1/2} \Psi} \nonumber \\
 &\quad - C N^{-\frac{5}{3}+2\delta} M \norm{(\Ncal_\delta+M)^{1/2} (\Ncal +1)\psi}^2\;.
\end{align} 
The bound \cref{eq:KEY-11} is one of the most subtle estimates of our analysis and does not have any counterpart in the proof of the upper bound. Note that on the right hand side, once and only once the vector $\Psi$ appears. The proof of this bound relies on the nice property \cref{eq:Kbound-intro} of the Bogoliubov kernel.

 The second and third summand on the right hand side of \eqref{eq:KEY-11} are simply bounded by the a--priori estimates \cref{eq:KEY-04}. Unlike $CN^{-\frac{1}{2}} \norm{\psi}\norm{\Hbb_0^{1/2} \Psi}$ with its small pre--factor $N^{-\frac{1}{2}}$ the expectation value $-\norm{\hat V}_{\ell^1} \langle \psi,\Hbb_0 \psi\rangle$ has to be controlled differently: since $\norm{\hat V}_{\ell^1}$ is assumed to be small, we can control it using the positive term $\Hbb_0$ left after the Bogoliubov transformation of the difference of fermionic and bosonic kinetic energy, see the right hand side of \cref{eq:KEY-10}.

Eventually we will take the parameters $M=N^{4\delta}$ and $\delta=\frac{1}{24}$, resulting in the total error $\mathcal{O}(\hbar^{1+\frac{1}{16} })$ to the correlation energy. 
This completes the sketch of the proof.

\section*{Acknowledgements}
We thank Christian Hainzl for helpful discussions and a referee for very careful reading of the paper and many helpful suggestions. NB and RS were supported by the European Research Council (ERC) under the European Union’s Horizon 2020 research and innovation programme (grant agreement No.~694227). Part of the research of NB was conducted on the RZD18 Nice--Milan--Vienna--Moscow. NB thanks Elliott~H.\ Lieb and Peter Otte for explanations about the Luttinger model.
PTN has received funding from the Deutsche Forschungsgemeinschaft (DFG, German Research Foundation) under Germany's Excellence Strategy (EXC-2111-390814868). 
MP acknowledges financial support from the European Research Council
(ERC) under the European Union’s Horizon 2020 research and innovation programme (ERC StG MaMBoQ,
grant agreement No.~802901).
BS gratefully acknowledges financial support from the NCCR SwissMAP, from the Swiss National Science Foundation through the Grant ``Dynamical and energetic properties of Bose-Einstein condensates'' and from the European Research Council through the ERC-AdG CLaQS (grant agreement No.~834782).  
All authors acknowledge support for workshop participation from Mathematisches Forschungsinstitut Oberwolfach (Leibniz Association). NB, PTN, BS, and RS acknowledge support for workshop participation from Fondation des Treilles.

\section{Kinetic Estimates}\label{sec:kin_est}

Our goal is to derive some rough estimates on the correlation Hamiltonian $\Hcal_\textnormal{corr}$ in \cref{eq:Hcorr} using the kinetic energy
\[
\Hbb_0 = \sum_{p\in \mathbb{Z}^3} e(p) a^*_p a_p\;, \quad e(p)= \lvert \hbar^2 \lvert p\rvert^2 -\kappaf^2\rvert\;, \qquad \kappaf=  \left(\frac{3}{4\pi}\right)^{\frac{1}{3}}\;. 
\]
The main result of this section is the following estimate for $\Hcal_\textnormal{corr}$. The proof is based on the estimates of \cite[Section~4.1]{HPR20}, which we shall review for the convenience of the reader.
\begin{lem}[A--Priori Estimates for the Correlation Hamiltonian]\label{cor:kineticenergy}
There exists a $v_0 > 0$ such that for $\hat{V}: \Zbb^3 \to \Rbb$ compactly supported, non--negative, with $\hat{V}(k)=\hat{V}(-k)$ for all $k\in \Zbb^3$, and $\norm{\hat{V}}_{\ell^1}< v_0$, the following holds true:
\[
\frac{1}{2} (\Hbb_0 +\mathcal{E}_1)  -  \hbar \leq \Hcal_\textnormal{corr}  \leq 2 (\Hbb_0 +\mathcal{E}_1+ \hbar)\;. 
\]
\end{lem}

Before coming to the proof of this lemma at the end of the section we need a couple of auxiliary lemmas.
We start by recalling \cite[Lemma~4.7]{HPR20}, which allows us to control the pair operators $b(k)$ using the kinetic energy $\Hbb_0$.  
\begin{lem} [Kinetic Bound for Pair Operators] \label{lem:kinetic}  For every $k\in \mathbb{Z}^3$ and $\psi \in \fock$ we have   
\[
\sum_{p \in \BFc \cap (\BF +k)} \norm{a_p a_{p-k} \psi} \leq CN^{\frac{1}{2}} \norm{\Hbb_0^{1/2} \psi}\;.
\]
\end{lem}

Since we will use this bound several times and we also need a modified version in \cref{sec:patches}, a simplified proof of \eqref{lem:kinetic} is provided in \cref{app:A} for the reader's convenience. 

%will be used many times in our paper, and 
%As this is an important and non--trivial bound, let us recall the proof.
% 
As a consequence of \cref{lem:kinetic} we can bound the pair operators by the kinetic energy.
\begin{lem}[Kinetic Bound for $b(k)^\natural$] \label{lem:bD}
For all $k\in \Zbb^3$ we have
\[
 b^*(k)  b(k) \leq C N \Hbb_0\;, \quad  b(k)  b^*(k) \leq C N (\Hbb_0+ \hbar)\;. 
\]
\end{lem}

\begin{proof}
For every $\psi \in \fock$, from \cref{lem:kinetic} and the triangle inequality we have
\[
\norm{b(k) \psi} \leq \sum_{p\in \BFc \cap (\BF +k)}  \norm{a_p a_{p-k}\psi} \leq C N^{\frac{1}{2}} \norm{\Hbb_0^{1/2} \psi}\;.
\]
This is equivalent to $b^*(k)b(k) \leq C N \Hbb_0$. To estimate $b(k) b^*(k)$, we use the CAR \cref{eq:car}
\begin{align}
[b(k),b^*(k)] &= \sum_{p,q\in \BFc \cap (\BF +k)}  [a_p a_{p-k}, a_{q-k}^* a_q^*] =  \sum_{p\in \BFc \cap (\BF +k)}  \Big( 1 - a^*_p a_p -  a^*_{p-k} a_{p-k}\Big) \nonumber  \\
&\leq \lvert\{ p\in \BFc \cap (\BF +k) \} \rvert \leq CN^{\frac{2}{3}} = C N \hbar\;. % \label{eq:counting-0}
\end{align}
%%%%%%%%%%%%%%%%%%%%%%%%%%%%%%
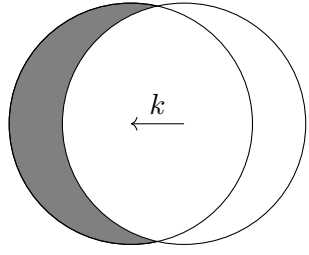
\begin{figure}\centering
\begin{tikzpicture}
 \draw[fill=gray] (4,2) circle (1.6cm);
 \draw[fill=white] (4.7,2) circle (1.6cm);
 \draw (4,2) circle (1.6cm);
 \node [above] at (4.35,2) {$k$}; 
 \draw [->] (4.7,2)--(4,2);
\end{tikzpicture}
 \caption{The grey area represents the set $\BFc \cap (\BF  + k) \subset \Zbb^3$, i.\,e., momenta which are affected by creation of particle--hole pairs with relative momentum $k$.}\label{fig:normalization}
\end{figure}
%%%%%%%%%%%%%%%%%%%%%%%%%%%%
The last estimate follows from a simple counting argument: the set $\BFc \cap (\BF +k)$ (sketched in grey in \cref{fig:normalization}) is contained in the volume obtained by extending an area of size $\Ocal(N^{\frac{2}{3}})$ on the Fermi surface to a shell of thickness of order $\Ocal(1)$. Therefore, this set contains no more than $C N^{\frac{2}{3}}$ points of $\Zbb^3$. Thus 
\[
b(k)  b^*(k) \leq b^*(k)  b(k)  + C N \hbar  \leq C N (\Hbb_0+ \hbar)\;. \qedhere
\]
\end{proof}

The following new bound is our main tool to control the number of excited fermions in the ground state. It is based on the observation that because of the discreteness of momentum space the kinetic energy operator has a tiny (i.\,e., order $\hbar^2$) gap. 

\begin{lem}[Kinetic Bound for Number of Fermions] \label{lem:N-H0}
Let $\psi \in \fock$ satisfy $\left( \Ncal^\textnormal{p} - \Ncal^\textnormal{h} \right) \psi = 0$. Then we have
\[
\langle \psi, \Ncal \psi \rangle \leq 2 N^{\frac{2}{3}} \langle \psi, \Hbb_0 \psi \rangle \;. 
\]
% For any $\theta \in (0, \frac{2}{3}]$ and any $\bourgain > 0$ we have
% \[
%  \Ncal \leq C_\bourgain N^{1-\theta+\bourgain } + N^{\theta} \Hbb_0\;. 
%  \]
\end{lem}
\begin{proof}
 Consider any $k_2 \in B_\F^c$ and $k_1 \in B_\F$. By definition of the Fermi ball $\lvert k_2 \rvert > \lvert k_1 \rvert$, and since $k_1, k_2 \in \Zbb^3$ we have $\lvert k_2 \rvert^2 - \lvert k_1 \rvert^2 \in \Nbb$; thus
 \[
  \inf_{k \in B_\F^c} \lvert k \rvert^2 - \sup_{k \in B_\F}\lvert k \rvert^2 \geq 1\;.
 \]
 Define 
 \[c_0 := \frac{1}{2} \inf_{k \in B_\F^c} \lvert k\rvert^2 + \frac{1}{2} \sup_{k \in B_\F} \lvert k\rvert^2 \;.\]
 Obviously $\sup_{k \in B_\F} \lvert k\rvert^2 \leq c_0 \leq \inf_{k \in B_\F^c} \lvert k\rvert^2$, and since $B_\F \cup B_\F^c = \Zbb^3$ we get
 \[\inf_{k \in \Zbb^3} \lvert \lvert k \rvert^2 - c_0 \rvert \geq \frac{1}{2}\;.\]
 Moreover, using $\left(\Ncal^\textnormal{p} - \Ncal^\textnormal{h}\right) \psi = 0$ we find
 \begin{align*}
  \Hbb_0 \psi & = \sum_{p \in B_\F^c} \left( \hbar^2 \lvert p\rvert^2 - \kappa^2 \right) a^*_p a_p \psi - \sum_{h \in B_\F} \left( \hbar^2 \lvert h\rvert^2 - \kappa^2 \right) a^*_h a_h \psi \\
  & = \sum_{p \in B_\F^c} \left( \hbar^2 \lvert p\rvert^2 - \hbar^2 c_0 \right) a^*_p a_p \psi - \sum_{h \in B_\F} \left( \hbar^2 \lvert h\rvert^2 - \hbar^2 c_0 \right) a^*_h a_h \psi + \left( \hbar^2 c_0 - \kappa^2 \right) \left( \Ncal^{\textnormal{p}} - \Ncal^\textnormal{h} \right) \psi \\
  & = \sum_{k \in \Zbb^3} \hbar^2 \lvert \lvert k\rvert^2 - c_0\rvert a^*_k a_k \psi\;.
 \end{align*}
Thus we conclude
\[
\langle \psi, \Hbb_0 \psi \rangle \geq \sum_{k \in \Zbb^3} \langle \psi , \frac{\hbar^2}{2} a^*_k a_k \psi\rangle = \frac{\hbar^2}{2} \langle \psi, \Ncal \psi \rangle \;. \qedhere 
\]
\end{proof}

As a consequence of \cref{lem:N-H0} we can easily prove that the exchange term has a very small contribution to the ground state energy, namely bounded as in the following lemma. 
\begin{lem}[Bound for $\Xbb$] \label{lem:exchange} Let $\psi \in \fock$ satisfy $\left( \Ncal^\textnormal{p} - \Ncal^\textnormal{h} \right) \psi = 0$. Then 
\[
\lvert \langle \psi, \Xbb \psi\rangle \rvert \leq C N^{-\frac{1}{3}} \langle \psi, \Hbb_0 \psi \rangle \;.   
\]
\end{lem} 

\begin{proof} This follows from the simple estimate\footnote{We use the notation $\pm A \leq B$ for two self--adjoint operators $A$ and $B$ to indicate that both $A \leq B$ and $-A \leq B$ hold.} $\pm \Xbb \leq C\norm{\hat V}_{\ell^1} \Ncal/N$ and \cref{lem:N-H0}.  
\end{proof}

Now we are ready to prove the main result of the section.
\begin{proof}[Proof of \cref{cor:kineticenergy}] By the Cauchy--Schwarz inequality and  \cref{lem:bD}
\begin{align*} 
\pm \Big( \Dfrak(-k)^* b(k) + \hc  \Big)   &\leq \frac{1}{2}\Dfrak(-k)^* \Dfrak(-k) + 2 b^*(k) b(k)\;,\\
\pm \Big( b^*(k) b^*(-k) + \hc  \Big) &\leq b^*(k) b(k) + b(-k) b^*(-k) \leq CN (\Hbb_0 +\hbar)\;. 
\end{align*}
Combining this with \cref{lem:exchange} we obtain 
\[
\Hbb_0 + \frac{1}{2} \mathcal{E}_1 -  C \norm{\hat V}_{\ell^1} (\Hbb_0+ \hbar) \leq  \Hcal_\textnormal{corr} \leq \Hbb_0 + 2 \mathcal{E}_1+ C \norm{\hat V}_{\ell^1} (\Hbb_0 + \hbar)\;.
\]
The desired result follows from the smallness condition on $\hat V$.
\end{proof}

\section{Localization of Particle Number} 

From the previous kinetic energy estimates it is possible to derive a--priori bounds for the ground states of $\Hcal_\textnormal{corr}$ (see \cref{lem:gs} below). For example, we can control the expectation value of the particle number $\Ncal$ in a ground state using \cref{lem:N-H0}. To estimate also the expectation values of higher powers of $\Ncal$, inspired by \cite{LS01} we use IMS localization with respect to particle number to construct an \emph{approximate} ground state which has energy close to the ground state energy and at the same time fulfills the desired bounds for powers of the number operator. This is the main outcome of the section, given in the following lemma.
\begin{lem}[Localization in Particle Number] \label{lem:loc} Let $\psi_\textnormal{gs}$ be a ground state vector for $\Hcal_{corr}$ satisfying $(\Ncal^\textnormal{p} - \Ncal^\textnormal{h})\psi_\textnormal{gs} = 0$, i.\,e., a minimizer of the left hand side of \cref{eq:Hcorr-lb}. Then there exists a normalized vector  $\Psi\in \fock$ such that 
\[
\langle \Psi,  (\Hbb_0 +\mathcal{E}_1) \Psi\rangle \leq C\hbar\;, \qquad \Psi= \mathds{1} (\Ncal \leq CN^{\frac{1}{3}})\Psi 
\] 
(i.,e., $\Psi$ lives in the Fock space sectors with particle number less or equal to $N^{\frac{1}{3}}$) and
\begin{align}  \label{eq:energyprop}
\langle \psi_\textnormal{gs}, \Hcal_\textnormal{corr} \psi_\textnormal{gs}\rangle \geq \langle \Psi, \Hcal_\textnormal{corr} \Psi \rangle - CN^{-1}\;.
\end{align} 
Furthermore $(\Ncal^\textnormal{p} - \Ncal^\textnormal{h})\Psi = 0$.
\end{lem}

As a first ingredient for the proof of \cref{lem:loc}, we have a--priori estimates based on \cref{sec:kin_est}. Note that for $\psi= \Omega$ we have $\langle \Omega, \Hcorr \Omega\rangle =0 \leq C\hbar$, and thus for any ground state $\psi_\textnormal{gs}$ of $\Hcorr$ we have $\langle \psi_\textnormal{gs},\Hcorr \psi_\textnormal{gs} \rangle \leq 0$ by the variational principle. Thus we can apply the following lemma to $\psi_\textnormal{gs}$. 

\begin{lem}[A--Priori Estimates] \label{lem:gs}
Let $\psi \in \fock$ such that $\langle \psi, \Hcal_\textnormal{corr} \psi\rangle \leq C\hbar$. Then we have 
\[
\langle \psi,  (\Hbb_0 +\mathcal{E}_1) \psi\rangle \leq C\hbar\;, 
\qquad  \langle \psi, \Ncal \psi \rangle \leq CN^{\frac{1}{3}}\;.
\]
\end{lem} 

\begin{proof} From the lower bound in \cref{cor:kineticenergy} and the assumption of \cref{lem:gs} we have 
\[
\frac{1}{2}\langle \psi,  (\Hbb_0 +\mathcal{E}_1)  \psi \rangle - \hbar \leq \langle \psi, \Hcal_\textnormal{corr} \psi \rangle \leq C\hbar\;. 
\]
This implies $\langle \psi ,  (\Hbb_0 +\mathcal{E}_1)  \psi \rangle \leq C \hbar$. The bound on $\Ncal$ follows from \cref{lem:N-H0}.
\end{proof}
% 
% \Cref{lem:gs} applies to the ground state $\psi_\textnormal{gs}$ of $\Hcal_\textnormal{corr}$, thanks to  the variational principle
% \[
% \langle \psi_\textnormal{gs}, \Hcal_\textnormal{corr} \psi_\textnormal{gs} \rangle = E_N -  E_N^\textnormal{HF,pw} \leq 0\;. 
% \]

Next, we localize the particle number using a suitable localization formula on Fock space. This technique goes back to \cite[Theorem A.1]{LS01}. The following general statement is taken from \cite[Proposition 6.1]{LNSS15}. 

\begin{lem}[Localization on Fock Space]
\label{lem:ims}
Let $\Acal$ be a non--negative operator on $\mathcal{F}$ such that $P_i D(\Acal)\subset D(\Acal)$ and $P_i \Acal P_j=0$ if $|i-j|>\ell$, where $P_i=\mathds{1}(\Ncal=i)$.  Let $f,g: [0,\infty) \to [0,1]$ be smooth functions such that $f^2 + g^2=1$, $f(x) = 1$ for $ x \leq 1/2$ and $f(x)=0$ for $ x \geq 1$. For any $L \geq 1$ define the operators
\[
f_L:= f(\Ncal/L)\;, \qquad g_L:= g(\Ncal/L)\;. 
\]
Then 
\[
- \frac{C_f \ell^3}{L^2}  \Acal_\textnormal{diag}  \leq \Acal -  f_L \Acal f_L- g_L \Acal g_L  \leq \frac{C_f \ell^3}{L^2} \Acal_\textnormal{diag}
\]
where $\Acal_\textnormal{diag}=\sum_{i=0}^{\infty} P_i \Acal P_i$ and $C_f= 2 (\norm{f'}_{L^\infty}^2 + \norm{g'}_{L^\infty}^2)$.  
\end{lem}

The proof of \cref{lem:ims} in \cite{LNSS15} is based on the double commutator identity 
\[
\Acal -  f_L \Acal f_L- g_L \Acal g_L = \frac{1}{2} [ f_L, [ f_L, \Acal]] + \frac{1}{2} [ g_L, [ g_L, \Acal]] \quad\text{when}\quad f_L^2+g_L^2=1\;. 
\]
This is an analogue of the standard IMS localization formula in position space \cite{Sim83}
\[
\Delta -  f \Delta  f- g \Delta g =  \lvert \nabla f\rvert^2 +   \lvert\nabla g\rvert^2\;, \quad f^2 + g^2 =1 \;.
\]

Now we are ready to prove the main result of the section.
\begin{proof}[Proof of \cref{lem:loc}] We will apply \cref{lem:ims} for $\Acal = \Hcal_\textnormal{corr} + \hbar$. We can take $\ell=4$ as the Hamiltonian $ \Hcal_\textnormal{corr}$ changes particle number by at most $\pm4$.  By \cref{cor:kineticenergy}, we have
\[
0 \leq \Acal \leq C (\Hbb_0 + \mathcal{E}_1 + \hbar )
\]
which also implies that 
\[
\Acal_\textnormal{diag} \leq C (\Hbb_0 + \mathcal{E}_1 + \hbar )
\]
because $\Ncal$ commutes with $\Hbb_0$ and $ \mathcal{E}_1$. Thus by \cref{lem:gs} we get (with the constant $C_0 > 0$ fixed for reference in the further proof)
\begin{align} \label{eq:gs-property}
\langle \psi_\textnormal{gs}, \Acal_\textnormal{diag} \psi_\textnormal{gs}\rangle  \leq C \hbar\;, \quad \langle \psi_\textnormal{gs}, \Ncal \psi_\textnormal{gs}\rangle \leq C_0 N^{\frac{1}{3}}\;.
\end{align}
Now applying \cref{lem:ims}, for all $L\geq 1$ we can bound 
\[
\langle  \psi_\textnormal{gs}, \Acal   \psi_\textnormal{gs} \rangle \geq \langle f_L  \psi_\textnormal{gs}, \Acal f_L  \psi_\textnormal{gs} \rangle + \langle g_L  \psi_\textnormal{gs}, \Acal g_L  \psi_\textnormal{gs} \rangle - C\hbar L^{-2} \;.
\]
Combining with the variational principle (and using that $\psi_\textnormal{gs}$ is a ground state for $\Acal$)
\[
\langle g_L  \psi_\textnormal{gs}, \Acal g_L  \psi_\textnormal{gs} \rangle \geq \norm{g_L \psi_\textnormal{gs} }^2 \langle  \psi_\textnormal{gs}, \Acal  \psi_\textnormal{gs}\rangle 
\] 
and then together with $f_L^2+g_L^2=1$ we obtain 
\begin{align} \label{eq:loc-f2-faf}
\norm{f_L \psi_\textnormal{gs}}^2 \langle  \psi_\textnormal{gs}, \Acal   \psi_\textnormal{gs} \rangle \geq \langle f_L  \psi_\textnormal{gs}, \Acal f_L  \psi_\textnormal{gs} \rangle - C\hbar L^{-2}\;.
\end{align}
Choosing $L:=4 C_0 N^{\frac{1}{3}}$, with $C_0$ the constant fixed in \cref{eq:gs-property}, we get 
\[
\norm{f_L \psi_\textnormal{gs} }^2 = 1 - \norm{g_L \psi_\textnormal{gs} }^2 \geq 1 - \frac{2\langle \psi_\textnormal{gs}, \Ncal \psi_\textnormal{gs}\rangle }{L} \geq  \frac{1}{2}\;. 
\]
Consequently,  \cref{eq:loc-f2-faf} implies 
\[
 \langle  \psi_\textnormal{gs}, \Acal   \psi_\textnormal{gs} \rangle \geq \langle \Psi, \Acal \Psi \rangle - CN^{-1} 
\]
with 
\[
\Psi:= \frac{f_L  \psi_\textnormal{gs}}{ \norm{ f_L \psi_\textnormal{gs}} }\;. 
\]
Since $\Acal=\Hcal_\textnormal{corr}+\hbar$ and $\Psi$ and $\psi_\textnormal{gs}$ are normalized, the previous inequality is equivalent to 
\begin{align*}
\langle  \psi_\textnormal{gs}, \Hcal_\textnormal{corr} \psi_\textnormal{gs} \rangle \geq \langle \Psi, \Hcal_\textnormal{corr} \Psi\rangle -C N^{-1}\;. 
\end{align*}
Finally, from the definition of $\Psi$ and $0\leq f_L \leq \mathds{1} (\Ncal \leq L)$ we get $\Psi= \mathds{1} (\Ncal \leq L) \Psi$. Since $\Ncal$ and $\Ncal^\textnormal{p} - \Ncal^\textnormal{h}$ commute, it follows also that $(\Ncal^\textnormal{p} - \Ncal^\textnormal{h})\Psi = 0$.  
\end{proof} 

\section{Reduction to Pair Excitations on Patches} \label{sec:patches}

Our bosonization method is based on decomposing the pair excitations $b^*(k)$ into smaller pieces localized in disjoint patches on the Fermi surface. This procedure has been introduced in the context of the renormalization group \cite{BG90,HM93,HKMS94,Hal94,CF94}. In this section we will define the patches precisely. Moreover, we prove that the correlation Hamiltonian $H_\textnormal{corr}$ can be properly represented by the pair excitations in patches up to an explicitly estimated error. 

First, we  will decompose the Fermi surface into patches; a partition of the Fermi surface is sketched in Figure~\ref{fig:blub}. Then we thicken the patches on the Fermi surface by allowing a relative momentum of order $\Ocal(1)$. With a parameter $\delta \in (0,1/6)$ that will eventually be optimized, the number $M$ of patches will be chosen in the range 
\begin{equation}\label{eq:defM2}
N^{2\delta} \ll M \ll N^{\frac{2}{3} -2 \delta}\;.
\end{equation}
(We will eventually take $\delta=1/24$ and $M=N^{4\delta}$.) The parameter $\delta$ is the same as will appear in the gapped number operator \cref{eq:gappedN2} and in the equator cut--off \cref{eq:cutoff2}. The details of the construction are given in the following paragraphs, leading to the patch definition \cref{eq:patch}. 
\begin{figure}\centering
\begin{tikzpicture}[scale=0.7]
% \clip(-5,-2) rectangle (5,5);
\def\RadiusSphere{4} % sphere radius
\def\angEl{20} % elevation angle
\def\angAz{-20} % azimuth angle

\filldraw[ball color = white] (0,0) circle (\RadiusSphere);

\DrawLatitudeCircle[\RadiusSphere]{75+2}
\foreach \t in {0,-50,...,-250} {
  \DrawLatitudeArc{75}{(\t+50-4)*sin(62)}{\t*sin(62)}
 \DrawLongitudeArc{\t*sin(62)}{50+2}{75}
 \DrawLongitudeArc{(\t-4)*sin(62)}{50+2}{75}
  \DrawLatitudeArc{50+2}{(\t+50-4)*sin(62)}{\t*sin(62)}
 }
 \foreach \t in {0,-50,...,-300} {
   \DrawLatitudeArc{50}{(\t+50-4)*sin(37)}{\t*sin(37)}
 \DrawLongitudeArc{\t*sin(37)}{25+2}{50}
  \DrawLongitudeArc{(\t-4)*sin(37)}{25+2}{50}
   \DrawLatitudeArc{25+2}{(\t+50-4)*sin(37)}{\t*sin(37)}
 }
 \DrawLatitudeArc{50}{(-300-4)*sin(37)}{-330*sin(37)}
 \foreach \t in {0,-50,...,-450} {
    \DrawLatitudeArc{25}{(\t+50-4)*sin(23)}{\t*sin(23)}
 \DrawLongitudeArc{\t*sin(23)}{00+2}{25}
 \DrawLongitudeArc{(\t-4)*sin(23)}{00+2}{25}
 \DrawLatitudeArc{00+2}{(\t+50-4)*sin(23)}{\t*sin(23)}
 }
     \DrawLatitudeArc{25}{(-450-4)*sin(23)}{-500*sin(23)}

%center of polar cap
\fill[black] (0,3.75) circle (.075cm);

% first row from the pole
\fill[black] (1.72,3.08) circle (.075cm);
\fill[black] (.76,2.73) circle (.075cm);
\fill[black] (-.66,2.73) circle (.075cm);
\fill[black] (-1.73,3.04) circle (.075cm);
% \fill[black] (-1.81,3.54) circle (.075cm);

% second row from the pole
\fill[black] (2.25,1.5) circle (.075cm);
\fill[black] (.8,1.2) circle (.075cm);
\fill[black] (-.85,1.22) circle (.075cm);
\fill[black] (-2.27,1.5) circle (.075cm);
\fill[black] (-3.09,1.97) circle (.075cm);
\fill[black] (3.09,1.97) circle (.075cm);

% third row from the pole
\fill[black] (2.57,-.15) circle (.075cm);
\fill[black] (1.43,-.37) circle (.075cm);
\fill[black] (.155,-.48) circle (.075cm);
\fill[black] (-1.17,-.41) circle (.075cm);
\fill[black] (-2.35,-.2) circle (.075cm);
\fill[black] (-3.26,0.1) circle (.075cm);
\fill[black] (-3.79,.55) circle (.075cm);
\fill[black] (3.37,.18) circle (.075cm);
\fill[black] (3.85,.57) circle (.075cm);

%%%%%%%%%%%%%%%%%%%%%%%%%%
\end{tikzpicture}
 \caption{Patch decomposition of the northern half of the unit sphere: a spherical cap is placed at the pole; then collars along the latitudes are introduced and split into patches, separated by corridors. The vectors $\hat{\omega}_\alpha$ are picked as centers of the patches, marked in black. Finally patches are reflected by the origin to the southern half sphere.}\label{fig:blub}
\end{figure}
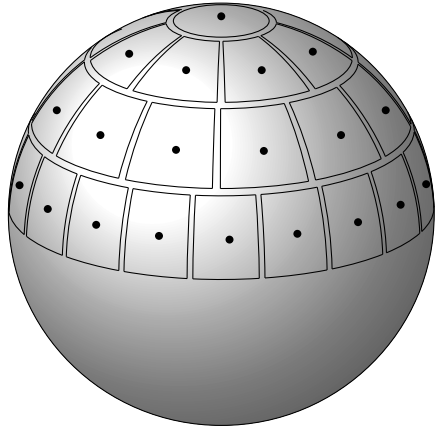

\paragraph{Patches on the unit sphere.} We start our construction on the unit sphere (following, e.\,g., \cite{Leo06}) and later scale up to the Fermi sphere of radius $k_\F = \kappaf N^{\frac{1}{3}}$. We use standard spherical coordinates: for $\hat\omega \in \Sbb^2$, denote by $\theta$ the inclination  (measured between $\hat\omega$ and $e_3 = (0,0,1)$) and by $\varphi$ the azimuth (measured between $e_1 = (1 , 0 , 0)$ and the projection of $\hat{\omega}$ onto the plane perpendicular to $e_3$). We write $\hat{\omega}(\theta,\varphi)$ to specify a vector on the unit sphere by its inclination and azimuth.

We place a spherical cap centered at $e_3$ with opening angle $\Delta\theta_0 := D /\sqrt{M}$, $D > 0$ chosen such that the area of the cap is $4\pi / M$. Then we decompose the remaining part of the northern half sphere into $\sqrt{M}/2$ (rounded to the next integer) collars; the $i$--th collar consists of all $\hat{\omega}(\theta,\varphi)$ with $\theta \in [\theta_i - \Delta\theta_i,\theta_i+\Delta\theta_i)$ and arbitrary $\varphi$. The inclination of every collar extends over $\Delta\theta_i \sim 1/{\sqrt{M}}$; the proportionality constant is adjusted so that the number of collars is integer.

Observe that the circle $\left\{ \hat\omega(\theta_i,\varphi): \varphi \in [0,2\pi) \right\}$ 
has circumference proportional to $\sin(\theta_i)$; therefore we split the $i$--th collar into $\sqrt{M}\sin(\theta_i)$ (rounded to the next integer) patches. This implies that the $j$--th patch in the $i$--th collar covers an azimuth $\varphi \in [\varphi_{i,j}-\Delta\varphi_{i,j},\varphi_{i,j} + \Delta\varphi_{i,j})$, where
\begin{equation}\Delta\varphi_{i,j} \sim \frac{1}{{\sin(\theta_i)\sqrt{M}}}\;.
\end{equation}
We fix the proportionality constants demanding that all patches have area $4\pi/M$.

Since $\hat{V}$ is compactly supported, we can set 
\[
R_{\hat{V}} := \diam \supp \hat{V}\;.
\]
Next we introduce \emph{corridors} between the patches by redefining 
\begin{equation}    \label{eq:tildeD}
\Delta\widetilde{\theta}_i := \Delta\theta_i - \tilde{D} R_{\hat{V}} N^{-\frac{1}{3}}, \quad \Delta\widetilde{\varphi}_{i,j} := \Delta\varphi_{i,j} - \tilde{D} R_{\hat{V}} N^{-\frac{1}{3}}/\sin(\theta_i)\;.
\end{equation}
 The constant $\tilde{D}$ is chosen such that when scaled up to the Fermi sphere adjacent patches are separated by corridors of width strictly larger than $2 R_{\hat{V}}$.

We then define $p_1$ as the spherical cap with opening angle $\Delta\widetilde{\theta}_0$ centered at $e_3$ and the other $\frac{M}{2} - 1$ patches as
\begin{equation} \label{eq:pij}
p_{i,j} := \big\{ \hat{\omega}(\theta,\varphi): \theta \in [\theta_i-\Delta\widetilde{\theta}_i,\theta_i+\Delta\widetilde{\theta}_i) \text{ and } \varphi \in [\varphi_{i,j}-\Delta\widetilde{\varphi}_{i,j},\varphi_{i,j} + \Delta\widetilde{\varphi}_{i,j})\big\}\;.
\end{equation}
 Patches on the southern half sphere are defined through reflection by the origin. Finally we enumerate the patches by  $\alpha \in \{1, \ldots , M\}$ and obtain the collection $\{p_\alpha\}_{\alpha=1}^M$ from \cref{eq:pij}. This completes the construction of patches $\{p_\alpha\}_{\alpha=1}^M$ for the  unit sphere.

\paragraph{Patches on the Fermi sphere.} Next, 
using the Fermi momentum
\[
k_\F = \kappaf N^{\frac{1}{3}}
\]
we scale the patches $\{p_\alpha\}_{\alpha=1}^M$ from the unit sphere up to the Fermi sphere by defining 
\[
P_\alpha := k_\F p_\alpha \qquad \forall \alpha\in\{1,2,\ldots,M\} \;. 
\]
By construction, we have the following properties for $\{P_\alpha\}_{\alpha=1}^M$: 
\begin{enumerate}

\item Reflection symmetry: $-P_\alpha = P_{\alpha+\frac{M}{2}}$ for all $\alpha = 1,\ldots ,\frac{M}{2}$.

\item The area of every patch is $
4\pi k_\F^2{M^{-1}}  \Big( 1 + \Ocal\big(M^{\frac{1}{2}}N^{-\frac{1}{3}}\big)\Big).
$ Moreover, the diameter of each patch is bounded by $
\diam(P_\alpha) \leq C N^{\frac{1}{3}}M^{-\frac{1}{2}}.
$
In words: patches do not degenerate into elongated thin strips. 

\item The patches are separated by corridors of width at least $2 R_{\hat{V}}$. The area of the union of all corridors is bounded by $CN^{\frac{1}{3}}M^{\frac{1}{2}}$.
\end{enumerate}

\paragraph{Extended patches around the Fermi sphere.} 
Next we extend the patch decomposition radially, over the shell around the Fermi surface that is affected by the interaction with momenta $k \in \supp\hat{V}$.  For any $\alpha\in \{1,2,\ldots,M\}$ we introduce the extended patch
\begin{equation} \label{eq:patch}
B_\alpha := \left\{ q \in \Zbb^3 :  k_\F - R_{\hat{V}} \leq \lvert q \rvert \leq k_\F + R_{\hat{V}}  \right\} \bigcap \Bigg( \bigcup_{r \in (0,\infty)} r P_\alpha \Bigg)\;. 
\end{equation}
Thus we have the following properties for the patch decomposition $\{B_\alpha\}_{\alpha=1}^M$: 
\begin{enumerate}

\item Reflection property: $-B_\alpha = B_{\alpha+\frac{M}{2}}$ for all $\alpha = 1,\ldots ,\frac{M}{2}$.

\item The diameter of each patch is bounded by $C N^{\frac{1}{3}}/\sqrt{M}$. 

\item The patches $\{B_\alpha\}_{\alpha=1}^M$ are pairwise disjoint and separated by corridors of width $2 R_{\hat{V}}$. (If the separation of patches on the Fermi surface is $S$, then below the Fermi surface, at distance $k_\F - R_{\hat{V}}$ from the origin, their separation is $S-\Ocal(N^{-\frac{1}{3}})$. Since by construction \cref{eq:tildeD} $S$ is strictly larger than $2R_{\hat{V}}$, also $S-\Ocal(N^{-\frac{1}{3}}) > 2 R_{\hat{V}}$ for large enough $N$.)
\end{enumerate}

\paragraph{Removing patches near the equator.} 
Now we assign a unit vector $\hat{\omega}_\alpha$ to every patch on the northern half such that $k_\F\hat{\omega}_\alpha\in P_\alpha \subset B_\alpha$. Reflecting the construction to the southern half sphere, the vectors $\hat{\omega}_{\alpha}$ inherit the reflection symmetry
\[
\hat \omega_{\alpha + M/2} = -\hat\omega_\alpha, \quad \forall \alpha = 1, \ldots , M/2\;.
\]
For any momentum $k\in \mathbb{Z}^3\setminus \{0\}$, we are only interested in a subset of the constructed patches, as labeled by the index set (the parameter $\delta \in (0,1/6)$ is the same as in \cref{eq:defM2})
\begin{equation} \label{eq:cutoff2}
\Ik := \{\alpha\in \{1,2, \ldots, M\} :  \lvert k \cdot \hat{\omega}_\alpha\rvert \geq N^{-\delta} \}\;.
\end{equation}
Pair excitations near the equator $k \cdot \hat{\omega}_\alpha \approx 0$ are almost tangential to the Fermi surface and cannot be treated with the bosonization technique. Fortunately their contribution to the energy turns out to be small.

  For any $k\in \mathbb{Z}^3\setminus \{0\}$ we define the operators without the corridors and the excitations near the equator as
\begin{align} 
b^\Rcal(k) &:= \sum_{\alpha \in \Ik} \sum_{\substack{p\in \BFc \cap B_\alpha \\ p-k\in \BF \cap B_\alpha}} a_{p-k}a_p \;,   \label{eq:QBNr1} \\
Q_\B^{\Rcal} & := \frac{1}{N} \sum_{k \in \north } \hat{V}(k) \Big[  b^\Rcal(k) ^* b^\Rcal(k) + b^\Rcal(-k)^* b^\Rcal(-k) + b^\Rcal(k)^* b^\Rcal(-k)^* + b^\Rcal(-k) b^\Rcal(k) \Big]\;, \label{eq:QBNr3} \\
 \mathcal{E}_2^{\Rcal} & := \frac{1}{N} \sum_{k \in \north } \hat{V}(k) \Big[  \Dfrak(k)^* b^\Rcal(k)  + \Dfrak(-k)^* b^\Rcal(-k) + \hc \Big]\;.   \label{eq:QBNr2}
\end{align}
The main result of this section is the following lemma, which takes care of estimating the difference to the original $\Hcal_\textnormal{corr}$.
\begin{lem}[Reduction to Pair Excitations on Patches] \label{lem:remove-corridors}
We have
\[
 Q_\B + \mathcal{E}_2 - Q_\B^{\Rcal} - \mathcal{E}_2^{\Rcal}   \geq - C\Big( N^{ - \frac{\delta}{2}} + C N^{ -\frac{1}{6} + \frac{\delta}{2} } M^{\frac{1}{4}}\Big)  \left(\Hbb_0+\mathcal{E}_1+\hbar \right)\;. 
\]
\end{lem}

In order to prove \cref{lem:remove-corridors} we will need the following modified version of the kinetic energy estimate in \cref{lem:kinetic}. 
\begin{lem}[Kinetic Bound for Pairs near the Equator]
\label{lem:kineticequator}
Let $\delta \in (0,77/624)$. Then for all $k\in \Zbb^3$ we have
 \begin{equation} \label{eq:kineticequator} \sum_{\substack{p\colon p \in \BFc \cap (\BF +k)\\ e(p)+e(p-k) \leq 4 N^{-\frac{1}{3} -\delta}  }} \norm{a_{p}a_{p-k}\psi} \leq CN^{\frac{1}{2} -\frac{\delta}{2}} \norm{\Hbb_0^{1/2} \psi} \qquad \forall \psi \in \fock\;.\end{equation}
\end{lem}
The condition $e(p)+e(p-k) \leq 4 N^{-\frac{1}{3} -\delta}$ implies that the momentum $p$ is located near the equator of the Fermi surface (if we think of $k$ as defining the direction of the north pole). This is easily seen expanding $e(p) + e(p-k) = \hbar^2 2 p\cdot k - \hbar^2 k^2$; because $\lvert p \rvert \sim N^{\frac{1}{3}}$we then have $k\cdot\hat{p} < C N^{-\delta}$. The idea of \cref{lem:kineticequator} is that the estimate in \cref{lem:kinetic} can be improved since here we sum only over a ribbon parallel to the equator on the Fermi surface. The ribbon covers a fraction of order $N^{-\delta}$ of the Fermi surface, explaining the improvement from $N$ to $N^{1-\delta}$ in \eqref{eq:kineticequator}. The proof of \cref{lem:kineticequator} can be found in \cref{app:A}.

\begin{proof}[Proof of \cref{lem:remove-corridors}]
For every $k\in \mathbb{Z}^3$, recall from \cref{eq:QBNr1} that
\[
b^\Rcal(k) := \sum_{\alpha \in \Ik} \sum_{\substack{p\in \BFc \cap B_\alpha \\ p-k\in \BF \cap B_\alpha}} a_{p-k}a_p \;,\]
i.\,e., the summation is over all $p$ in the set
 \[\bigcup_{\alpha \in \Ik} \left( \BFc \cap B_\alpha \right) \cap \left( (\BF \cap B_\alpha) +k \right) = \BFc \cap (\BF + k) \cap \bigcup_{\alpha \in \Ik} \left(B_\alpha \cap (B_\alpha +k)\right)\;.\]
Thus the error term compared to the full pair operator becomes
\begin{equation} \label{eq:badpairs}
\rfrak^\Rcal(k) := b(k)- b^\Rcal(k) = \sum_{p\in U} a_{p-k} a_{p} 
\end{equation}
where, with $A_1 := \BFc \cap (\BF +k)$ and $A_2 := \bigcup_{\alpha\in \Ik}  \Big(  B_\alpha \cap (B_\alpha+k) \Big)$, we define the set
\[
U := A_1 \setminus (A_1 \cap A_2) \;. 
\] 
In words: $U$ consists of all those particle momenta $p \in \BFc$ that correspond to a kinematically permitted particle--hole pair (i.\,e., $h := p-k$ is inside the Fermi ball)  but do not belong to any included patch. Thus in \cref{eq:badpairs} we sum over pairs belonging to a corridor between patches or to the cut--off equator region (i.\,e., they belong to a patch $B_\alpha$ but $\alpha \not\in \Ik$).
To estimate $\rfrak^\Rcal(k)$, let $\psi \in \fock$. By the triangle inequality we can bound
\[
\norm{\rfrak^\Rcal(k) \psi} \leq \sum_{p\in U} \norm{a_{p-k} a_{p} \psi} \leq \sum_{p\in Y} \norm{a_{p-k} a_{p} \psi} + \sum_{p\in U\setminus Y} \norm{a_{p-k} a_{p} \psi}
\]
where 
\[
Y :=  \{ p   \in U: e(p)+e(p-k) \leq 4 N^{-\frac{1}{3} -\delta} \}\;.
\]
To make contact with our earlier heuristic explanations, note that $Y$ corresponds to the region (both patches and corridors between patches) of the Fermi surface near the equator, i.\,e., where $k \cdot \hat{p} \approx 0$, and $U \setminus Y$ to the corridors between the patches on the rest of the Fermi surface. This may be seen by expanding $e(p) + e(p-k) = \hbar^2 2 p\cdot k - \hbar^2 k^2$; because $p$ is close to the Fermi surface we have $\lvert p \rvert \sim N^{\frac{1}{3}}$, so that $k\cdot\hat{p} < C N^{-\delta}$.

On the set $Y$, by \cref{lem:kineticequator}  we get 
\begin{align} \label{eq:reduce-patches-01}
\sum_{p\in Y} \norm{a_{p} a_{p-k} \psi} \leq C N^{\frac{1}{2} - \frac{\delta}{2}}  \norm{\Hbb_0^{1/2} \psi}\;.
\end{align} 
We turn to the set $U \setminus Y$.
By the Cauchy--Schwarz inequality we have
\begin{equation}
 \sum_{p\in U\setminus Y} \norm{a_p a_{p-k}\psi} \leq \sqrt{ \sum_{p\in U\setminus Y}  \norm{a_p a_{p-k}\psi}^2} \sqrt{ \sum_{p\in U\setminus Y}  1} \;.
\end{equation}
To estimate the first factor, note that $p \not\in Y$ implies $e(p)+e(p-k) \geq 4N^{-\frac{1}{3}-\delta}$, so that we have $e(p) \geq 2N^{-\frac{1}{3}-\delta}$ or $e(p-k)\geq 2N^{-\frac{1}{3}-\delta}$. Consequently 
\begin{align} \label{eq:reduce-patches-02}
\sum_{p\in U\setminus Y} \norm{a_p a_{p-k} \psi}^2 &\leq  \sum_{p\in U\setminus Y}    \min\{\norm{a_p \psi}^2, \norm{a_{p-k} \psi}^2 \} \nonumber\\
&\leq  \sum_{q \in \Zbb^3\colon e(q)\geq  2N^{-\frac{1}{3}-\delta} }  \norm{a_q \psi}^2 \leq \frac{1}{2}N^{\frac{1}{3} +\delta} \norm{\Hbb_0^{1/2}\psi}^2\;. 
\end{align} 
% Next, we count the lattice points of $\Zbb^3$ inside $U\setminus Y$. Assume $p\in U\setminus Y$. Let $B_\alpha$ be the patch closest to $p$. Using  the condition $p\notin Y$, we can write 
% \[
% 4 N^{-\frac{1}{3} -\delta} \leq e(p)+e(p-k)=\hbar^2 ( \lvert p\rvert^2 - \lvert p-k\rvert^2 ) = \hbar^2 \Big( 2k_\F \hat{\omega}_\alpha \cdot k +   2 (p -k_\F \hat{\omega}_\alpha) \cdot k -\lvert k\rvert^2 \Big)\;. 
% \]
% Moreover, by construction of the patches we have  
% \[
% \lvert p-k_\F \hat{\omega}_\alpha\rvert \leq C N^{\frac{1}{3}}M^{-\frac{1}{2}} \ll N^{\frac{1}{3}-\delta}\;.
% \]
% Thus we deduce that $\hat{\omega}_\alpha \cdot k > N^{-\delta}$, or in other words, $\alpha\in \Ik$. The condition $p\in U$  then implies that $p$ or $p-k$ does not belong to $B_\alpha$. 
To estimate the second factor, note that the number of lattice points of $\Zbb^3$ in $U\setminus Y$ can be bounded by the number of lattice points in the corridors between \emph{all} patches  
\begin{align} \label{eq:reduce-patches-03}
\lvert U\setminus Y \rvert \leq C\frac{N^{\frac{1}{3}}}{\sqrt{M}} \times M =  CN^{\frac{1}{3}}M^{\frac{1}{2}}\;. 
\end{align} 
(Here we used that the length of a corridor surrounding a patch is of order $N^{\frac{1}{3}} M^{-\frac{1}{2}}$, its width of order one, and the number of patches is $M$.)
Having estimated both factors, we get
\begin{align} \label{eq:reduce-patches-04}
\sum_{p\in U\setminus Y} \norm{a_p a_{p-k}\psi} \leq C N^{\frac{1}{3} +  \frac{\delta}{2}} M^{\frac{1}{4}} \norm{\Hbb_0^{1/2}\psi}\;.
\end{align} 
%%%%%%%%%%%%%%
Putting \cref{eq:reduce-patches-04,eq:reduce-patches-01} together we arrive at
\begin{align} \label{eq:reduce-patches}
\norm{\rfrak^\Rcal(k) \psi} \leq C\Big( N^{\frac{1}{2} - \frac{\delta}{2}} + C N^{\frac{1}{3} + \frac{\delta}{2} } M^{\frac{1}{4}}\Big) \norm{\Hbb_0^{1/2}\psi}\;, \qquad \forall k\in \mathbb{Z}^3\;.  
\end{align} 

Now we turn to the Hamiltonian. Expanding $b(k)= b^\Rcal(k) + \rfrak^\Rcal(k)$ in the formula for $\Hcal_\textnormal{corr}$ in \cref{eq:Hcorr}, we get
\begin{align*} 
Q_B + \mathcal{E}_2 - Q_\B^\Rcal - \mathcal{E}_2^\Rcal &=   \frac{1}{2N} \sum_{k\in \mathbb{Z}^3 \setminus \{0\}} \hat V(k) \Big[ b^*(k) \rfrak^\Rcal(k) + \rfrak^\Rcal(k)^* b^\Rcal(k) + \nonumber \\
&\qquad + \Big( b(k) \rfrak^\Rcal(-k) +  b^\Rcal (-k) \rfrak^\Rcal(k) + 2 \Dfrak(k)^* \rfrak^\Rcal(k) + \hc \Big)   \Big]\; . 
\end{align*}
It is easy to see that the kinetic bounds in \cref{lem:bD} hold also with $b(k)$ replaced by $b^\Rcal(k)$. Therefore, in combination with \cref{eq:reduce-patches}, by the Cauchy--Schwarz inequality we get
\begin{align*}
&\Big\lvert\langle \psi,  (Q_B + \mathcal{E}_2 - Q_\B^\Rcal - \mathcal{E}_2^\Rcal ) \psi \rangle \Big\rvert \\
&\leq CN^{-1} \sum_{k\in \mathbb{Z}^3 \setminus \{0\}}  \hat V(k) \norm{\rfrak^\Rcal(k) \psi} \Big( \norm{ b^\Rcal(k) \psi} + \norm{ b^\Rcal(k)^* \psi} + \norm{\Dfrak(k) \psi} \Big)\\
&\leq C\Big( N^{ - \frac{\delta}{2}} + C N^{ -\frac{1}{6} + \frac{\delta}{2} } M^{\frac{1}{4}}\Big) \langle \psi, (\Hbb_0+\mathcal{E}_1+\hbar)\psi\rangle\;.  \qedhere
\end{align*}
\end{proof}

\section{Approximately Bosonic Creation Operators}

Let $\{B_\alpha\}_{\alpha=1}^M$ be the patches constructed as in the previous section and let $k_\F \hat{\omega}_\alpha\in B_\alpha$. Recall that  for every $k\in \mathbb{Z}^3 \setminus\{0\}$ we have defined 
\[
\Ik := \left\{\alpha\in \{1,2, \ldots, M\} :  \lvert k \cdot \hat{\omega}_\alpha\rvert \geq N^{-\delta} \right\}\;.
\]
By the reflection symmetry we decompose further $\Ik := \Ikm \cup \Ikp$ where 
\[\begin{split}
\Ikp & := \left\{ \alpha \in \{1,2,\ldots, M\} : k \cdot \hat{\omega}_\alpha \geq N^{-\delta} \right\}\;,\\
\Ikm & := \left\{ \alpha \in \{1,2,\ldots, M\} : k \cdot \hat{\omega}_\alpha \leq - N^{-\delta} \right\}\;.
\end{split}
\]
Then we define the local pair excitations $\{c_\alpha^*(k)\}_{\alpha\in \Ik}$ by 
\[            
c^*_\alpha(k) := \frac{1}{n_\alpha(k)} \sum_{\substack{p\colon p \in \BFc \cap B_\alpha\\p-k\in \BF \cap B_\alpha}}  a^*_p a^*_{p-k}\;, \qquad n_\alpha(k)^2 := \sum_{\substack{p\colon p \in \BFc \cap B_\alpha\\p-k\in \BF \cap B_\alpha}} 1\;, \qquad \text{if }\alpha\in \Ikp
\]
and
\[            
c^*_\alpha(k) := \frac{1}{n_\alpha(k)} \sum_{\substack{p\colon p \in \BFc \cap B_\alpha\\p+k\in \BF \cap B_\alpha}}  a^*_p a^*_{p+k}\;, \qquad n_\alpha(k)^2 := \sum_{\substack{p\colon p \in \BFc \cap B_\alpha\\p+k\in \BF \cap B_\alpha}} 1\;,\qquad \text{if }\alpha\in \Ikm\;.
\]
Thus, for all $k \in \north$, the operator $b^\Rcal(k)$ in \cref{lem:remove-corridors} can be decomposed as
\begin{align}\label{eq:bR-c}
b^\Rcal(k) = \sum_{\alpha\in \Ikp} n_\alpha(k) c_\alpha(k)\;, \qquad b^\Rcal(-k)  = \sum_{\alpha\in \Ikm} n_\alpha(k) c_\alpha(k)\;.
\end{align}

The quantity $n_\alpha(k)^2$ counts the number of particle--hole pairs of relative momentum $k$ belonging to patch $B_\alpha$. We cite the result from \cite[Proposition~3.1]{BNP+20}. The condition $M\gg N^{\frac{1}{3}}$ in \cite{BNP+20} is not necessary, $M \gg N^{2\delta}$ is sufficient.  
\begin{lem}[Normalization Constant] 
\label{lem:counting} Assume that $N^{2\delta} \ll M \ll N^{\frac{2}{3}-2\delta}$. Then for all $k \in \north$ and $\alpha \in \Ik$, we have 
\[
n_\alpha(k)^2  =  \frac{4\pi k_\F^2}{M} \lvert k \cdot \hat\omega_\alpha \rvert  \left( 1 + o(1) \right)\;.
\]
\end{lem}
(This lemma may heuristically be understood as follows: the surface area covered by the patch is $4\pi k_\F^2/M$. We think of the particle--hole pairs $(p,h) \in B_\F^c \times B_\F$ with $p-h =k$ as organized on lines through the lattice $\Zbb^3$ parallel to $k$. To count how many lines intersect the Fermi surface we project the patch onto a plane, picking up the factor $\lvert k \cdot \hat\omega_\alpha \rvert$. The condition $M \gg N^{2\delta}$ ensures that patches remain so small that even near their boundaries the assumption $\lvert k\cdot \hat{\omega}_\alpha \rvert \geq N^{-\delta}$ implies that $k$ points from inside the Fermi ball to outside. The error term arises since we may miscount a pair when one of its components falls into the surrounding corridor, so it is proportional to the circumference of the patch. We write only $1+o(1)$ because the precise estimate is not important for us.) 
 
A crucial idea of our analysis is that the local pair excitation operators $\{c_\alpha^*(k)\}_{\alpha\in\Ik}$ behave similarly to bosonic creation operators. More precisely, we have approximate canonical commutator relations as given by the following lemma. The lemma is a simple consequence of \cite[Lemma 4.1]{BNP+20}, but since it is a key idea of the collective bosonization concept we provide a self--contained proof again.

\begin{lem}[Approximate CCR]\label{lem:ccr}
Let $k\in\north $ and $l\in \north $. Let $\alpha \in \Ik$ and $\beta \in \Il$. The operators $c_\alpha(k)$ and $c_\beta^*(l)$ satisfy the following commutator relations:
\begin{equation}
[c_\alpha(k),c_\beta(l)] = 0 = [c^*_\alpha(k),c^*_\beta(l)]\;, \quad [c_\alpha(k),c^*_\beta(l)]  = \delta_{\alpha,\beta}\big( \delta_{k,l} + \Ecal_\alpha(k,l) \big)\;. 
\label{eq:approximateCCR}
\end{equation}
The operator $\Ecal_\alpha(k,l)=\Ecal_\alpha(l,k)^*$ commutes with $\Ncal$ and, for any $\gamma \in \Ik\cap\Il$, satisfies the operator inequalities 
\begin{align} \label{eq:ccr-0}
\lvert \Ecal_\gamma(k,l) \rvert^2 \leq \sum_{\alpha\in \Ik \cap \Il} \lvert \Ecal_\alpha(k,l) \rvert^2  \leq C (M N^{-\frac{2}{3}+\delta} \Ncal)^2\;. 
\end{align}
Furthermore, for all $\psi \in \fock$ we have
\begin{align} \label{eq:ccr-1}
\sum_{\alpha\in \Ik \cap \Il} \norm{ \Ecal_\alpha(k,l) \psi }  \leq C M^{\frac{3}{2}}N^{-\frac{2}{3}+\delta} \norm{\Ncal \psi}\;. 
\end{align}
\end{lem}

\begin{proof}  By the CAR \cref{eq:car} it is easy to see that
\[
[c_\alpha(k),c_\beta(l)] = 0 = [c^*_\alpha(k),c^*_\beta(l)]\;.
\]
Moreover, if $\alpha\ne \beta$, then $B_\alpha \cap B_\beta= \emptyset$, and hence $[c_\alpha(k),c^*_\beta(l)]  =0$. 

Now let us focus on the case $\beta=\alpha$ and compute $[c_\alpha(k),c^*_\alpha(l)]$. We only consider the case $\alpha \in \Ikp \cap \Ilp$ (the other cases are simple variations). By the CAR \cref{eq:car} it is straightforward to compute that  
\begin{align}\label{eq:Ekk}
\Ecal_\alpha(k,l) &= - \frac{1}{n_\alpha(k)n_\alpha(l)} \Bigg[  \sum_{\substack{p\colon p \in \BFc \cap B_\alpha\\p-k,\,p-l \in \BF \cap B_\alpha}} a^*_{p-l} a_{p-k} + \sum_{\substack{h\colon h\in \BF\cap B_\alpha \\ h+l,\,h+k\in \BFc \cap B_\alpha\\}} a^*_{h+l} a_{h+k} \Bigg] \\
&=: \Ecal^{(1)}_\alpha(k,l) + \Ecal^{(2)}_\alpha(k,l)\; .\nonumber
 \end{align}
  Let us focus on $\lvert \Ecal^{(1)}_\alpha(k,l)\rvert^2$; the term $\lvert\Ecal^{(2)}_\alpha(k,l)\rvert^2$ can be bounded similarly, and the mixed terms are controlled by the Cauchy--Schwarz inequality. Symmetrizing, we find
\begin{align*}
\lvert  \Ecal^{(1)}_\alpha(k,l) \rvert^2 &= \frac{1}{2n_\alpha(k)^2n_\alpha(l)^2}  \sum_{\substack{p,\,q\colon p,\,q \in \BFc \cap B_\alpha\\p-k,\,p-l,\,q-k,\,q-l \in \BF \cap B_\alpha}}  \left(  a^*_{p-k} a_{p-l} a_{q-l}^* a_{q-k} +\hc\right)\;.
 \end{align*}
By \cref{lem:counting} we have  $n_{\alpha}(k) n_{\alpha}(l) \geq C^{-1} N^{\frac{2}{3}-\delta}/M$. Moreover, by the Cauchy--Schwarz inequality, 
 \begin{align*}
 \pm \left( a^*_{p-k} a_{p-l} a_{q-l}^* a_{q-k} +\hc\right) 
 & = \pm  \left( \delta_{p,q}a^*_{p-k}a_{p-k} - a^*_{p-k} a_{q-l}^* a_{p-l}  a_{q-k} + \hc \right)\\
 &\leq \left(2\delta_{p,q}a^*_{p-k}a_{p-k} + a^*_{p-k} a_{q-l}^* a_{q-l}  a_{p-k} + a^*_{q-k}a^*_{p-l}a_{p-l}  a_{q-k} \right)\;.
\end{align*}
% Note also that if $p=q\in B_\alpha\cap B_\gamma$, then $\gamma=\alpha$ since the patches are disjoint.
Therefore
\begin{align*}
& \sum_{\alpha\in \Ikp \cap \Ilp} \lvert \Ecal^{(1)}_\alpha(k,l) \rvert^2 \displaybreak[0]\\
&\leq  C\left(M N^{-\frac{2}{3}+\delta}\right)^2 \sum_{\alpha\in \Ikp \cap \Ilp}  \sum_{\substack{p,\,q\colon p,\,q \in \BFc \cap B_\alpha\\p-k,\,p-l,\,q-k,\,q-l \in \BF \cap B_\alpha}}    \left(  \delta_{p,q}a^*_{p-k}a_{p-k} + a^*_{p-k} a_{q-l}^* a_{q-l}  a_{p-k}\right) \displaybreak[0]\\
&\leq C\left(M N^{-\frac{2}{3}+\delta}\right)^2 \sum_{\alpha \in \Ikp \cap \Ilp} \sum_{\substack{p\colon p \in \BFc \cap B_\alpha\\p-k,\,p-l \in \BF \cap B_\alpha}} \left(  a^*_{p-k}a_{p-k} + a^*_{p-k} \Ncal a_{p-k}\right) \displaybreak[0]\\
&= C \left(M N^{-\frac{2}{3}+\delta}\right)^2 \sum_{\alpha \in \Ikp \cap \Ilp} \sum_{\substack{p\colon p \in \BFc \cap B_\alpha\\p-k,\,p-l \in \BF \cap B_\alpha}} a^*_{p-k} a_{p-k} \Ncal \\
&\leq C \left(M N^{-\frac{2}{3}+\delta}\right)^2  \Ncal^2\;.
 \end{align*}

 The first bound in \cref{eq:ccr-0} (without the summation) is a trivial consequence. The bound \cref{eq:ccr-1} follows from \cref{eq:ccr-0}  and  the Cauchy--Schwarz inequality.
\end{proof}

In the following, we show that the approximately bosonic number operator can be controlled by a fermionic number operator. One of our main technical improvements compared to \cite[Lemma~4.2]{BNP+20} is that instead of using the full $\Ncal$ we use only the gapped number operator 
\begin{equation}\label{eq:gappedN2}
 \Ncal_\delta := \sum_{i \in \Zbb^3\colon e(i) \geq \frac{1}{4} N^{-\frac{1}{3} -\delta}} a^*_i a_i\;.  
\end{equation}
The parameter $\delta>0$ is the same as that in the cut--off parameter $N^{-\delta}$ defining the index set $\Ik$ of relevant patches. Compared to $\Ncal$ as in \cref{lem:N-H0}, the gain in using the gapped number operator $\Ncal_\delta$ is that it can be controlled by $\langle \Psi, \Ncal_\delta \Psi \rangle \leq C N^{\delta}$ in an approximate ground state $\Psi$.

\begin{lem}[Bosonic Number Operator]\label{lem:bosonic-number} 
For all $k \in \north $ we have
\begin{align}
\label{eq:bosebound0}
\sum_{\alpha\in \Ik} c_\alpha^*(k) c_\alpha(k) \leq \Ncal_\delta\;.  
\end{align}
Consequently, for all $\psi\in \fock$, 
\begin{align}   \label{eq:boseboundX}
\sum_{\alpha\in \Ik} \norm{c_\alpha(k) \psi} \leq M^{\frac{1}{2}}\norm{\Ncal_{\delta}^{1/2} \psi}, \quad \sum_{\alpha\in \Ik} \norm{c_\alpha^* (k) \psi} \leq M^{\frac{1}{2}}\norm{(\Ncal_{\delta}+M)^{1/2} \psi}
 \end{align}
and for $f \in \ell^2(\Ik)$ also 
\begin{align}
  \norm{ \sum_{\alpha \in \Ik} f_\alpha c^*_{\alpha}(k) \psi} & \leq \norm{f}_{\ell^2} \norm{\left( \Ncal_{\delta} +1 \right)^{1/2}\psi}\;. \label{eq:bosebound1}
 \end{align}
\end{lem}

\begin{proof} First we take $\alpha \in \Ikp$ (the case $\alpha \in \Ikm$ is similar). For any $\psi\in \fock$,  by the triangle and Cauchy--Schwarz  inequalities, 
\begin{align*}
\norm{c_\alpha(k) \psi} &= \frac{1}{n_\alpha(k)}  \Big\lVert \sum_{\substack{p\colon p\in \BFc \cap B_\alpha \\ p-k\in \BF \cap B_\alpha } } a_{p} a_{p-k} \psi \Big\rVert \leq  \frac{1}{n_\alpha(k)}  \sum_{\substack{p\colon p\in \BFc \cap B_\alpha \\ p-k\in  \BF \cap B_\alpha } }   \norm{a_{p} a_{p-k} \psi}  \\
&\leq  \frac{1}{n_\alpha(k)} \Bigg( \sum_{\substack{p\colon p\in \BFc \cap B_\alpha \\ p-k\in  \BF \cap B_\alpha } }  1\Bigg)^{1/2} \Bigg(\sum_{\substack{p\colon p\in \BFc \cap B_\alpha \\ p-k\in \BF \cap B_\alpha } } \norm{a_{p} a_{p-k} \psi} ^2 \Bigg)^{1/2}\;. 
\end{align*}
Using the definition of $n_\alpha(k)$ and the fermionic property $\norm{a_i}_\textnormal{op} \leq 1$ we deduce that 
\[
 \norm{c_\alpha(k) \psi}^2 \leq \sum_{p\in \BFc \cap B_\alpha \cap (\BF +k) }    \norm{a_{p} a_{p-k} \psi} ^2 \leq  \sum_{p\in \BFc \cap B_\alpha \cap (\BF +k) }    \min\{\norm{a_p \psi}^2, \norm{a_{p-k} \psi}^2 \}\;. 
\]
For all $p\in \BFc \cap B_\alpha \cap (\BF+k)$ we have 
\[
\lvert p-k_\F \hat{\omega}_\alpha \rvert \leq \diam(B_\alpha) \leq CN^{\frac{1}{3}}M^{-\frac{1}{2}} \ll N^{\frac{1}{3}-\delta}\;,
\]
and the condition $\alpha\in \Ikp$ ensures that  $k\cdot \hat{\omega}_\alpha\geq N^{-\delta}$; hence
\[
e(p) + e(p-k) =\hbar^2(\lvert p\rvert^2- \lvert p-k\rvert^2)=  \hbar^2 \Big( 2 k_\F \hat{\omega}_\alpha \cdot k + 2(p-k_\F \hat{\omega}_\alpha) \cdot k   - \lvert k\rvert^2 \Big) \geq  \frac{1}{2} N^{-\frac{1}{3}-\delta}\;.
\]
Consequently, we have $e(p) \geq \frac{1}{4} N^{-\frac{1}{3}-\delta}$ or $e(p-k)\geq  \frac{1}{4} N^{-\frac{1}{3}-\delta}$. Thus
\[
 \norm{c_\alpha(k) \psi}^2 \leq  \sum_{p\in \BFc \cap B_\alpha \cap (\BF +k) }    \min\{\norm{a_p \psi}^2, \norm{a_{p-k} \psi}^2 \} \leq \sum_{\substack{q\in B_\alpha\colon e(q)\geq   \frac{1}{4} N^{-\frac{1}{3}-\delta} }}  \norm{a_q \psi}^2\;.
\]
By the same method we obtain the same bound when $\alpha \in \Ikm$. Thus by the definition of the gapped number operator we can bound
\[
\sum_{\alpha \in \Ik} \norm{c_\alpha(k) \psi}^2 \leq \sum_{\alpha \in \Ik}  \sum_{\substack{q\in B_\alpha\colon e(q)\geq   \frac{1}{4} N^{-\frac{1}{3}-\delta} }}  \norm{a_q \psi}^2 \leq \norm{\Ncal_{\delta}^{1/2}\psi}^2
\]
which is equivalent to \cref{eq:bosebound0}. Moreover, it can be seen from $\Ecal_\alpha(k,k) \leq 0$ in \cref{eq:Ekk} that
\begin{equation}    \label{eq:aacomm}
[c_\alpha(k), c^*_\alpha(k)]\leq 1 \;.
\end{equation}
Thus
\[
\sum_{\alpha \in \Ik} \norm{c^*_\alpha(k) \psi}^2 \leq \sum_{\alpha\in \Ik} ( \norm{c_\alpha(k) \psi}^2 +\norm{\psi}^2) \leq \norm{(\Ncal_\delta + M)^{1/2}\psi}^2\;. 
\]
By the Cauchy--Schwarz inequality we obtain
\[
\sum_{\alpha\in \Ik} \norm{c_\alpha(k) \psi} \leq M^{\frac{1}{2}}\norm{\Ncal_{\delta}^{1/2} \psi}\;, \quad \sum_{\alpha\in \Ik} \norm{c_\alpha^* (k) \psi} \leq M^{\frac{1}{2}}\norm{(\Ncal_{\delta}+M)^{1/2} \psi}\;. 
\]
Using that $[c_\alpha(k), c^*_\beta(k)]$ vanishes for $\alpha \neq \beta$, by \cref{eq:aacomm} we obtain
 \begin{align*}
  \norm{\sum_{\alpha \in \Ik} f(\alpha) c^*_\alpha(k) \psi}^2 & = \sum_{\alpha,\beta \in \Ik} \cc{f(\alpha)} f(\beta) \langle \psi, c^*_\beta(k) c_\alpha(k) \psi\rangle \\
  & \quad + \sum_{\alpha,\beta \in \Ik} \cc{f(\alpha)} f(\beta) \langle \psi, [c_\alpha(k),c^*_\beta(k)]  \psi\rangle \\
  & \leq \sum_{\alpha,\beta \in \Ik} \lvert f(\alpha)\rvert^2 \norm{c_\beta(k) \psi}^2 + \sum_{\alpha \in \Ik} \lvert f(\alpha)\rvert^2 \norm{\psi}^2\\
  &\leq  \sum_{\alpha \in \Ik} \lvert f(\alpha)\rvert^2 \norm{(\Ncal_\delta+1)^{1/2}\psi}^2\;.    
 \end{align*}
 This concludes the proof.
\end{proof}

For further application, it is useful to extend the definition of $c_\alpha(k)$ to include a weight function. Given $g: \mathbb{Z}^3\times \mathbb{Z}^3 \to \mathbb{R}$, we define weighted pair operators
\begin{align}
c^g_{\alpha}(k) &:= \frac{1}{n_{\alpha}(k)}\sum_{\substack{p\colon p \in \BFc \cap B_\alpha\\p-k \in \BF \cap B_\alpha}} g(p,k) a_{p-k} a_p \qquad \text{if } \alpha\in \Ikp\;,  \label{def:weightedpair} \\
c^g_{\alpha}(k) &:= \frac{1}{n_{\alpha}(k)}\sum_{\substack{p\colon p \in \BFc \cap B_\alpha\\p+k \in \BF \cap B_\alpha}} g(p,k) a_{p+k} a_p \qquad \text{if } \alpha\in \Ikm\;. \nonumber
\end{align}

The weighted pair operators satisfy similar bounds as the simple pair operators.

\begin{lem}[Weighted Pair Operators]\label{lem:bosonic-number-weight} For all $k \in \north $ and $\psi\in \fock$ we have 
\[
\sum_{\alpha\in \Ik} \norm{c^g_\alpha(k) \psi} \leq C M^{\frac{1}{2}} \norm{g}_{\ell^\infty} \norm{\Ncal_{\delta}^{1/2} \psi}\;, \quad \sum_{\alpha\in \Ik} \norm{c_\alpha^{g*} (k) \psi} \leq CM^{\frac{1}{2}} \norm{g}_{\ell^\infty} \norm{(\Ncal_{\delta}+M)^{1/2} \psi}\;, 
\]
and for all $f \in \ell^2(\Ik)$ also
\[
  \norm{ \sum_{\alpha \in \Ik} f_\alpha c^{g*}_{\alpha}(k) \psi}  \leq \norm{f}_{\ell^2} \norm{g}_{\ell^\infty} \norm{\left( \Ncal_{\delta} +1 \right)^{1/2}\psi}\;. 
\]
\end{lem}
\begin{proof}
 Compared to \cref{lem:bosonic-number} the only non--trivial modification is that we now use
 \begin{align*}
 [c_\alpha^g(k),c_\alpha^{g*}(k)] & = \frac{1}{n_\alpha(k)^2} \sum_{\substack{p \colon p \in \BFc \cap B_\alpha\\p-k \in \BF \cap B_\alpha}} \lvert g(p,k) \rvert^2 \left( 1 - a^*_p a_p - a^*_{p-k} a_{p-k} \right)\\
 & \leq \frac{1}{n_\alpha(k)^2} \sum_{\substack{p \colon p \in \BFc \cap B_\alpha\\p-k \in \BF \cap B_\alpha}} \lvert g(p,k) \rvert^2  \qquad \leq \norm{g}_{\ell^\infty}^2
 \end{align*}
 where before we used $[c_\alpha(k),c^*_\alpha(k)] \leq 1$. We omit the further details.
\end{proof}

\section{Bogoliubov Kernel} \label{sec:exa-diag}

In this section we study the Hamiltonian $h_\textnormal{eff}(k)$ introduced in \cref{eq:heff-intro}. Let us use $\hat{k} := k/\lvert k\rvert$. It is convenient to write 
\begin{equation}    \label{eq:heff}
h_\textnormal{eff}(k)  := \sum_{\alpha,\beta \in \Ik} \Big[ \big( \D(k) + \W(k) \big)_{\alpha,\beta} c_\alpha^*(k) c_\beta(k) +  \frac{1}{2} \Wt(k)_{\alpha,\beta} \big( c^*_\alpha(k) c^*_\beta(k) + c_\beta(k)  c_\alpha(k) \big)  \Big]
\end{equation}
where $\D(k)$, $\W(k)$, and $\Wt(k)$ are $\Ik \times \Ik$ real symmetric matrices with elements 
\begin{equation}
\label{eq:blocks1}\begin{split}
\D (k)_{\alpha,\beta} &: =  \delta_{\alpha,\beta} \lvert \hat{k} \cdot \hat{\omega}_\alpha\rvert \qquad \forall \alpha,\beta \in \Ik\;, \\
\W(k)_{\alpha,\beta}  & := \frac{\hat V(k)}{2\hbar \kappaf N \lvert k\rvert} \times \left\{ \begin{array}{cl} n_\alpha(k) n_\beta(k) & \text{ if } \alpha,\beta \in \Ikp \text{ or } \alpha,\beta \in \Ikm \\
 0 & \text{ otherwise}\;,\end{array} \right. \\
\Wt(k)_{\alpha,\beta} & := \frac{\hat V(k)}{2\hbar \kappaf N \lvert k\rvert} \times \left\{ \begin{array}{cl} 
0  & \text{ if } \alpha,\beta \in \Ikp \text{ or } \alpha,\beta \in \Ikm 
\\ n_\alpha(k) n_\beta(k)  & \text{otherwise}\;.
\end{array} \right.
\end{split}
\end{equation}

If $c_\alpha^*(k)$ were exactly bosonic creation operators, then the quadratic Hamiltonian $h_\textnormal{eff}(k)$ could be diagonalized by a Bogoliubov transformation of the form 
\begin{equation}\label{eq:bogformula}
\exp \Big( \frac{1}{2}\sum_{\alpha,\beta\in\Ik} K(k)_{\alpha,\beta} c^*_\alpha(k) c^*_\beta(k) -\hc \Big)\;. 
\end{equation} 
The matrix $K(k)$ (also called the Bogoliubov kernel) achieving this can be computed from $\D(k)$, $\W(k)$, $\Wt(k)$; we refer to \cite[Appendix A.1]{BNP+20} for a detailed derivation. Here let us just state the result. We introduce the $\Ik \times \Ik$ matrices 
\begin{equation} \label{eq:defE}
E(k) := \big[ \big(\D(k)+\W(k)-\Wt(k)\big)^{1/2} (\D(k)+\W(k)+\Wt(k)) \big(\D(k)+\W(k)-\Wt(k)\big)^{1/2} \big]^{1/2}
\end{equation}
and
\begin{align}\label{eq:defS1}
S_{1}(k)  := (\D(k)+\W(k)-\Wt(k))^{1/2} E(k)^{-1/2}    \;.
\end{align} 
(Formulas \cref{eq:block1,eq:block2} below show that the square roots here involve only positive matrices, so that $E(k)$ and $S_1(k)$ are well--defined.)
Then the Bogoliubov kernel $K(k)$ is
\begin{equation} 
\label{eq:Kk}
K(k) := \log \lvert S_1 (k)^\intercal \rvert =\frac{1}{2} \log  \Big( S_1(k) S_1(k)^\intercal \Big)\; .
\end{equation}

The following lemma provides strong estimates for the matrix elements of $K(k)$. While in most parts the simpler bound $\lvert K(k)_{\alpha,\beta}\rvert \leq C \hat{V}(k)/M$ is sufficient for our analysis, the sharp bound of the lemma is crucial for controlling the non--bosonizable terms $\Ecal_2$; see \cref{eq:ch-sh-sum}. The simpler bound can be proved without smallness assumption on the potential; the sharp bound requires the smallness because we prove it using a power series expansion in \cref{eq:powerseries}.

The proof of the lemma is a lengthy but mostly straightforward computation. A key role in the proof is played by the fact that the factor $k\cdot\hat{\omega}_\alpha$ arising from the linearized kinetic energy (through the matrix $D(k)$) appears also for the independent geometric reason of \cref{lem:counting} in the normalization factor $n^2_\alpha(k)$ in the bosonized interaction, i.\,e., in $W(k)$ and $\tilde{W}(k)$. The geometry of the Fermi surface implies that the excitation--creating operators of the interaction vanish at the same rate $u_\alpha= \sqrt{\lvert \hat{k} \cdot \hat{\omega}_\alpha\rvert}$ as the leading order of their kinetic energy when we move toward ``tangential'' excitations.

\begin{lem}[Bogoliubov Kernel]\label{lem:K} Let $K(k)$ be defined in \cref{eq:Kk}. If $\norm{\hat{V}}_{\ell^\infty}$ is sufficiently small, then for any $k\in \north$, $K(k)$ is a real symmetric matrix satisfying
\begin{equation}
\label{eq:Kbound}
\lvert K(k)_{\alpha,\beta}\rvert \leq \frac{C \hat{V}(k)}{M} \min \left\{ \frac{n_\alpha(k)}{n_\beta(k)},\frac{n_\beta(k)}{n_\alpha(k)}\right\} \quad \textnormal{for all }\alpha,\beta \in \Ik\;.
\end{equation}
\end{lem}
\begin{proof}  In the following, we frequently drop the $k$--dependence from the notation for simplicity.
Let us 
introduce
\[
g:=\frac{1}{2}\kappaf \hat V(k)\;, \qquad u_\alpha:= \sqrt{\lvert \hat{k} \cdot \hat{\omega}_\alpha\rvert}\;, \qquad v_\alpha:= \frac{\hbar}{\kappaf \sqrt{\lvert k\rvert}} n_\alpha(k)\;, \qquad \forall \alpha \in\Ik\;. 
\]
Recall that $\kappaf=(\frac{3}{4\pi})^{\frac{1}{3}}$. By definition of the index set $\Ik$, we have $1\geq u_\alpha^2 \geq CN^{-\delta}$ for all $\alpha \in \Ik$. Moreover, \cref{lem:counting} implies the important relation $v_\alpha \simeq u_\alpha\sqrt{ 4\pi/M}$ (up to a lower order error term) between the normalization factor $n_\alpha(k)$ and the linearization of the kinetic energy, which will be used repeatedly for cancellations.

Due to the reflection symmetry 
\[
B_{\alpha+M/2} = - B_\alpha, \quad \omega_{\alpha+M/2} = - \omega_\alpha \qquad \forall \alpha\in\{1,2,\ldots, M/2\} 
\]
we can denote $\ik := \lvert \Ikp\rvert = \lvert\Ikm\rvert \leq M/2$ and map the indices $\Ikp$ to $\{1,\ldots, \ik\}$, and the indices $\Ikm$ to $\{\ik+1,\ldots 2\ik\}$. 
Obviously
\[
u_\alpha= u_{\alpha+\ik}, \quad  v_\alpha = v_{\alpha+\ik} \qquad \forall \alpha\in\{1,2,\ldots,\ik\}\;. 
\]
Therefore, the matrices in \cref{eq:blocks1} can be written in the $2\ik \times 2\ik$ block form 
\[
\D = \begin{pmatrix} d & 0 \\ 0 & d \end{pmatrix}, \quad \label{eq:Wdef}\W = \begin{pmatrix} b & 0 \\ 0 & b \end{pmatrix}\;, \quad \Wt = \begin{pmatrix} 0& b\\ b & 0\end{pmatrix}\;,
\]
where $d=\diag(u_\alpha^2, \alpha = 1,\ldots,\ik)$ and $b=g\lvert v\rangle \langle v\rvert$ (i.\,e., a rank--one operator) with $v = (v_1, \cdots , v_{\ik})$. 

As in \cite{GS13}, denoting by $\id$ the $\ik\times \ik$--identity matrix, we define
\begin{equation}
\label{eq:blockdiagonalization}
U = \frac{1}{\sqrt{2}}\begin{pmatrix} \id & \id \\ \id & -\id \end{pmatrix}\;.
\end{equation}
Obviously $U^\intercal = U = U^{-1}$ and it simultaneously block--diagonalizes $\D+\W+ \Wt$ and $\D+\W- \Wt$, namely
\begin{equation}\label{eq:block1}
U^\intercal (\D+\W+\Wt) U = \begin{pmatrix} d+2b & 0 \\ 0 & d \end{pmatrix}\;, \quad U^\intercal (\D+\W-\Wt) U = \begin{pmatrix} d & 0\\ 0 & d+2b \end{pmatrix}\;.\end{equation}
Recall the matrix
\[
E = \left( (\D+\W-\Wt)^{1/2} (\D+\W+\Wt) (\D+\W-\Wt)^{1/2} \right)^{1/2} \in \Cbb^{2\ik\times 2\ik}\;;
\]
applying the block--diagonalization we find 
\begin{equation}\label{eq:block2}
U^\intercal E U = \begin{pmatrix} \left[ d^{1/2}(d+2b) d^{1/2} \right]^{1/2} & 0 \\ 0 & \left[ (d+2b)^{1/2} d (d+2b)^{1/2}\right]^{1/2} \end{pmatrix}\;.
\end{equation}
%Note that all matrices  $\D,\W,\Wt$ and $E$ are real and symmetric. 

Now defining the matrix $L := S_1 S_1^\intercal - \id$ with $S_1$ as in \cref{eq:defS1}, we find
\begin{align}\label{eq:matrixseries}
K = \frac{1}{2}\log(S_1 S_1^\intercal) = \frac{1}{2}\log\big(\id + L\big) = \frac{1}{2}\sum_{n=1}^\infty \frac{(-1)^{n+1} L^n}{n} \;.
\end{align}
We are going to prove
\begin{equation}\label{eq:Lbound}
\lvert L_{\alpha,\beta} \rvert \leq C \frac{\hat{V}(k)}{M} \min \left\{ \frac{u_\alpha}{u_\beta},\frac{u_\beta}{u_\alpha} \right\} \quad \textnormal{for all }\alpha,\beta \in \Ik\;.
\end{equation}
In particular, thanks to the assumption of $\norm{\hat{V}}_{\ell^\infty}$ being small, also, e.\,g., the Hilbert--Schmidt norm of $L$ can be assumed to be uniformly smaller than $1$, which is sufficient to ensure convergence of the matrix power series \cref{eq:matrixseries}.

\paragraph{From $L$ to $K$.}
Let us deduce \cref{eq:Kbound} by assuming \cref{eq:Lbound}. Spelling out the matrix product
\begin{equation*}
( L^n )_{\alpha,\beta} = \sum_{\alpha_1 \in \Ik} \sum_{\alpha_2 \in \Ik} \cdots \sum_{\alpha_{n-1} \in \Ik} L_{\alpha,\alpha_1} L_{\alpha_1,\alpha_2} \cdots L_{\alpha_{n-1},\beta}
\end{equation*}
we obtain (recall that $\lvert \Ik \rvert =2\ik \leq M$)
\begin{align*}
\lvert ( L^n )_{\alpha,\beta} \rvert & \leq \sum_{\alpha_1 \in \Ik} \sum_{\alpha_2 \in \Ik} \cdots \sum_{\alpha_{n-1} \in \Ik} \left( \frac{C\hat{V}(k)}{M} \frac{u_\alpha}{u_{\alpha_1}} \right) \left( \frac{C\hat{V}(k)}{M} \frac{u_{\alpha_1}}{u_{\alpha_2}} \right) \cdots \left( \frac{C\hat{V}(k)}{M} \frac{u_{\alpha_{n-1}}}{u_{\beta}} \right) \\
& \leq \left( \frac{C\hat{V}(k)}{M} \right)^n \lvert \Ik \rvert^{n-1} \frac{u_\alpha}{u_\beta} \leq \frac{(C\hat{V}(k))^n}{M}  \frac{u_\alpha}{u_\beta}\;. \tagg{Lnbound}
\end{align*}
The same holds with exchanged roles of $u_\alpha$ and $u_\beta$. 
From \cref{eq:matrixseries} we obtain
\begin{equation}    \label{eq:powerseries}
\begin{split}
\lvert K_{\alpha,\beta}\rvert & \leq \frac{1}{2} \sum_{n=1}^\infty \frac{\lvert (L^n)_{\alpha,\beta}\rvert}{n} \leq \frac{1}{2} \sum_{n=1}^\infty \frac{(C\hat{V}(k))^n}{2n M}  \min \left\{ \frac{u_\alpha}{u_\beta}, \frac{u_\beta}{u_\alpha} \right\}  \\
& = -\frac{\log( 1 - C\hat{V}(k))}{2M}  \min \left\{ \frac{u_\alpha}{u_\beta}, \frac{u_\beta}{u_\alpha} \right\} \leq  C\frac{\hat{V}(k)}{M} \min \left\{ \frac{u_\alpha}{u_\beta}, \frac{u_\beta}{u_\alpha} \right\}\;.
\end{split}
\end{equation}
The convergence of the series of the logarithm follows from the assumption that $|\hat{V}(k)|$ is small. This implies \cref{eq:Kbound}, thanks to \cref{lem:counting}. 
%concludes the proof of the bound for $K$, provided that \cref{eq:Lbound} holds true. 

\paragraph{Bound for $L$.}
We now prove \cref{eq:Lbound}. The $2\ik \times 2\ik$--matrix $L$ can be block--diagonalized using the orthogonal matrix $U$ from \cref{eq:blockdiagonalization}, i.\,e.,
 \begin{align}\label{eq:rel}
U \left( S_1 S_1^\intercal -\id \right) U^\intercal & = U (\D+\W-\Wt)^{1/2} E^{-1} (\D+\W-\Wt)^{1/2} U^\intercal -\id  = \begin{pmatrix} L_1 & 0 \\ 0 & L_2 \end{pmatrix}
 \end{align}
 with $\ik \times \ik$--blocks 
  \begin{equation}
  \label{eq:Ldef}
  \begin{split}
  L_1 & := d^{1/2}\big[d^{1/2}(d+2b)d^{1/2}\big]^{-1/2}d^{1/2} - \id\;, \\
  L_2 & := (d+2b)^{1/2} \big[ (d+2b)^{1/2}d(d+2b)^{1/2} \big]^{-1/2} (d+2b)^{1/2} - \id\;.
  \end{split}\end{equation}
 Inverting \cref{eq:rel} we obtain
 \begin{align}
 L = U^\intercal \begin{pmatrix} L_1 & 0 \\ 0 & L_2 \end{pmatrix} U = \frac{1}{2}\begin{pmatrix} L_1 + L_2 & L_1 - L_2 \\ L_1 - L_2 & L_1 + L_2 \end{pmatrix}\,.
 \end{align}
 Thus, with the matrix indices on $L_1$ and $L_2$ to be read as $\alpha \mod \ik$ and $\beta \mod \ik$, we have
 \begin{equation}
 \lvert L_{\alpha,\beta} \rvert \leq  \lvert (L_1)_{\alpha,\beta} \rvert + \lvert (L_2)_{\alpha,\beta} \rvert \;.
 \end{equation}

  \paragraph{Estimating $L_1$.}
 In the square brackets in the definition \cref{eq:Ldef} of $L_1$ we have a rank--one perturbation of a diagonal matrix, namely defining the vector $\tilde{v} := d^{1/2} v$ we have
\begin{equation}d^{1/2}(d+2b)d^{1/2} = d^2 + 2g \lvert \tilde{v}\rangle\langle \tilde{v} \rvert\;.\label{eq:resol}\end{equation}
Using the Sherman--Morrison formula for the resolvent of any invertible matrix $A$ with rank--one perturbation given by vectors $x$, $y$ such that $1+ \langle y, A^{-1} x \rangle \not= 0$, i.\,e.,
\begin{equation} \label{eq:A}
\left( A + \lvert x \rangle \langle y \rvert \right)^{-1} = A^{-1} - \frac{A^{-1} \lvert x \rangle \langle y \rvert A^{-1}}{1+ \langle y, A^{-1} x \rangle}
\end{equation}
we explicitly calculate the resolvent of \cref{eq:resol} and then enter with it in the integral representation $A^{-1/2} = \frac{2}{\pi}\int_0^\infty \di\lambda (A + \lambda^2)^{-1}$ (for any positive matrix $A$), with the result that
\[
 \left( d^2 + 2 g \lvert \tilde{v}\rangle\langle \tilde{v} \rvert \right)^{-1/2} = d^{-1} -  \frac{4g}{\pi} \int_0^\infty  \frac{\di \lambda}{f(\lambda)} (d^2 + \lambda^2)^{-1} \left\lvert \tilde{v} \right\rangle\left\langle \tilde{v} \right\rvert (d^2 + \lambda^2)^{-1}\;.
\]
The function $f(\lambda)$ here is given by
\[
f(\lambda) = 1 + 2g \sum_{\alpha=1}^{\ik} \frac{u_\alpha^2 v_\alpha^2}{u_\alpha^4 + \lambda^2}\;.
\]
Multiplying from both sides by $d^{1/2}$, and subtracting the identity matrix, we obtain
\[
 (L_1)_{\alpha,\beta} = -  \frac{4g}{\pi} \int_0^\infty \di\lambda \frac{1}{f(\lambda)} \frac{u_\alpha^2 v_\alpha}{u_\alpha^4 + \lambda^2}  \frac{u_\beta^2 v_\beta}{u_\beta^4 + \lambda^2}\;.
\]
Recall $v_\alpha \leq u_\alpha CM^{-\frac{1}{2}}$ and observe that $1/f(\lambda) \leq 1$, so we get
\begin{align}
\lvert \left(L_1\right)_{\alpha,\beta} \rvert & \leq \frac{Cg}{M}  \int_0^\infty \di \lambda \frac{1}{f(\lambda)} \frac{u_\alpha^3}{u_\alpha^4 + \lambda^2} \frac{u_\beta^3}{u_\beta^4 + \lambda^2} \leq \frac{Cg}{M}  u_\alpha^3 u_\beta^3 \int_0^\infty \di\lambda \frac{1}{u_\alpha^4 + \lambda^2} \frac{1}{u_\beta^4 + \lambda^2} \nonumber
\\
& = \frac{Cg}{M}  u_\alpha^3 u_\beta^3 \frac{\pi}{2(u_\alpha^4 u_\beta^2 + u_\alpha^2 u_\beta^4)}  = \frac{Cg}{M} \frac{u_\alpha u_\beta}{u_\alpha^2 + u_\beta^2} \leq \frac{Cg}{M} \min \left\{ \frac{u_\alpha}{u_\beta}, \frac{u_\beta}{u_\alpha} \right\}\;. \label{eq:L1}
\end{align}

\paragraph{Estimating $L_2$.} Recall that
\begin{align}
L_2 = (d+2b)^{1/2} \left[ (d+2b)^{2} -(d+2b)^{1/2} 2b (d+2b)^{1/2} \right]^{-1/2} (d+2b)^{1/2} - \id\;. \label{eq:defL2}
\end{align}
Here $-(d+2b)^{1/2} 2b (d+2b)^{1/2} = -2g \lvert (d+2b)^{1/2} v\rangle \langle (d+2b)^{1/2} v \rvert$ is a rank--one perturbation of $(d+2b)^2$, so by employing again the integral expansion as used for $L_1$ we obtain
\begin{align}\label{eq:L2}
L_2 & = \frac{4g}{\pi}\!\int_0^\infty\!\!\!\!\! \di \lambda \left[ 1 - 2g\big\langle v, \frac{d+2b}{(d+2b)^2 + \lambda^2} v\big\rangle \right]^{-1} \left\lvert \frac{d+2b}{(d+2b)^2+\lambda^2} v\right\rangle \left\langle  \frac{d+2b}{(d+2b)^2+\lambda^2} v \right\rvert.
\end{align}
Now consider the function $\tilde{f}(\lambda) := 1 - 2g\big\langle v, \frac{d+2b}{(d+2b)^2 + \lambda^2} v\big\rangle$, the inverse of which is appearing in the integral. For $\lambda = 0$,  using the Sherman--Morrison formula \cref{eq:A}, this time expanding $d+2b$ around $d$, we find
\begin{align*}
\tilde{f}(0) & = 1 - 2g\langle v,\bigg[ d^{-1} - 2g\frac{d^{-1} \lvert v\rangle \langle v\rvert d^{-1}}{1+ 2g\langle v, d^{-1} v\rangle} \bigg] v\rangle 
%              & = 1 - 2g \langle v,d^{-1} v\rangle + \frac{(2g \langle v,d^{-1} v\rangle)^2}{1+2g \langle v,d^{-1} v\rangle}  
             = 1 - \frac{2g \langle v,d^{-1} v\rangle}{1+2g \langle v,d^{-1} v\rangle}\;.
\end{align*}
Since
\begin{equation}
2g \langle v,d^{-1} v\rangle = 2g \sum_{\alpha=1}^{\ik} \frac{v_\alpha^2}{u_\alpha^2} \leq C g \sum_{\alpha=1}^{\ik} \frac{1}{M} \frac{u_\alpha^2}{u_\alpha^2} \label{eq:A23}
\end{equation}
is uniformly bounded we have $\tilde{f}(0) > 0$, strictly and uniformly in $k$ and $M$. Furthermore $\lambda \to \left\langle v, \frac{d+2b}{(d+2b)^2 + \lambda^2} v\right\rangle$ is monotone decreasing for all $\lambda \geq 0$, thus $\tilde{f}(\lambda) \geq \tilde{f}(0)$ for all $\lambda \geq 0$.
We expand 
$\frac{d+2b}{(d+2b)^2+\lambda^2}= (d+2b)(d+2b+i\lambda)^{-1} (d+2b -i\lambda)^{-1} $ and use the Sherman--Morrison formula separately for both the resolvents $(d+2b\pm i\lambda)^{-1}$. Using the Dirac ket notation, this results in
\begin{align*}
& \left\lvert \frac{d+2b}{(d+2b)^2+\lambda^2} v\right\rangle \\
& = (d+2b) \bigg[ 1 - \frac{2g(d+i\lambda)^{-1}\lvert v\rangle\langle v\rvert }{1+2g \langle v, (d+i\lambda)^{-1} v\rangle}\bigg] (d^2+\lambda^2)^{-1}  \bigg[ 1 - \frac{2g\lvert v\rangle\langle v\rvert (d-i\lambda)^{-1}}{1+2g \langle v, (d-i\lambda)^{-1} v\rangle}\bigg] \lvert v\rangle \\
& = (d+2b) \bigg[ 1 - 2g\frac{(d+i\lambda)^{-1}\lvert v\rangle\langle v\rvert }{1+2g \langle v, (d+i\lambda)^{-1} v\rangle}\bigg] (d^2+\lambda^2)^{-1}  \lvert v\rangle  \frac{1}{1+2g \langle v, (d-i\lambda)^{-1} v\rangle} \displaybreak[0]\\
& = d \bigg[ 1 - 2g\frac{(d+i\lambda)^{-1}\lvert v\rangle\langle v\rvert }{1+2g \langle v, (d+i\lambda)^{-1} v\rangle}\bigg] (d^2+\lambda^2)^{-1}  \lvert v\rangle  \frac{1}{1+2g \langle v, (d-i\lambda)^{-1} v\rangle} \\
& \quad + 2g \lvert v\rangle \frac{1}{1+2g \langle v, (d+i\lambda)^{-1} v\rangle} \langle v, (d^2+\lambda^2)^{-1}  v\rangle  \frac{1}{1+2g \langle v, (d-i\lambda)^{-1} v\rangle} 
\end{align*}
where for the last line we used $b = g\lvert v\rangle \langle v\rvert$.
Keeping the $1$ from the big square bracket separate while combining the other terms, this simplifies to
\begin{align}
 \left\lvert \frac{d+2b}{(d+2b)^2+\lambda^2} v\right\rangle 
& = \frac{d}{d^2+\lambda^2}  \lvert v\rangle  \frac{1}{1+2g \langle v, (d-i\lambda)^{-1} v\rangle} \label{eq:real1}\\
& \quad + 2g \frac{i\lambda}{d+i\lambda} \lvert v\rangle \frac{1}{\lvert 1+2g \langle v, (d+i\lambda)^{-1} v\rangle\rvert^2} \langle v, (d^2+\lambda^2)^{-1}  v\rangle
\;. \label{eq:real2}
\end{align}
The vector $\big\lvert \frac{d+2b}{(d+2b)^2+\lambda^2} v\big\rangle$ has real elements since $v$ is a real vector and $\frac{d+2b}{(d+2b)^2+\lambda^2}$ is a real matrix. However, \cref{eq:real1,eq:real2} are not explicitly real (by choosing an order out of the two options $(d+2b)^2 + \lambda^2 = (d+2b+i\lambda)^{-1} (d+2b -i\lambda)^{-1}$ and $(d+2b)^2 + \lambda^2 = (d+2b-i\lambda)^{-1} (d+2b +i\lambda)^{-1}$ we have broken this symmetry).
% So sad.
To make the expression explicitly real again, let us add the complex conjugate, yielding
\begin{align*}
2\left\lvert \frac{d+2b}{(d+2b)^2+\lambda^2} v\right\rangle &
% & = \frac{d}{d^2+\lambda^2}  \lvert v\rangle  \frac{1}{1+2g \langle v, (d-i\lambda)^{-1} v\rangle} + \frac{d}{d^2+\lambda^2}  \lvert v\rangle  \frac{1}{1+2g \langle v, (d+i\lambda)^{-1} v\rangle}\\
% & \quad + 2g \frac{i\lambda}{d+i\lambda} \lvert v\rangle \frac{\langle v, (d^2+\lambda^2)^{-1}  v\rangle}{\lvert 1+2g \langle v, (d+i\lambda)^{-1} v\rangle\rvert^2}  + 2g \frac{-i\lambda}{d-i\lambda} \lvert v\rangle \frac{\langle v, (d^2+\lambda^2)^{-1}  v\rangle}{\lvert 1+2g \langle v, (d+i\lambda)^{-1} v\rangle \rvert^2} \\
 = \frac{d}{d^2+\lambda^2}  \lvert v\rangle  \frac{2+2g \langle v, (d+i\lambda)^{-1} v\rangle +2g \langle v, (d-i\lambda)^{-1} v\rangle}{\lvert 1+2g \langle v, (d+i\lambda)^{-1} v\rangle\rvert^2}\\
& \quad + 2g \frac{2\lambda^2}{d^2+\lambda^2} \lvert v\rangle \frac{\langle v, (d^2+\lambda^2)^{-1}  v\rangle}{\lvert 1+2g \langle v, (d+i\lambda)^{-1} v\rangle\rvert^2}\;.
\end{align*}
Using $\langle v, (d+i\lambda)^{-1} v\rangle +  \langle v, (d-i\lambda)^{-1} v\rangle =   \langle v, \frac{2d}{d^2+\lambda^2} v\rangle$
(and dividing by $2$) this becomes
\begin{align}
& \left\lvert \frac{d+2b}{(d+2b)^2+\lambda^2} v\right\rangle \label{eq:A37} \\
& = \bigg[ \frac{d}{d^2+\lambda^2}  \lvert v\rangle   \left( 1+2g \langle v, \frac{d}{d^2 + \lambda^2} v\rangle \right)  + 2g \frac{\lambda^2}{d^2+\lambda^2} \lvert v\rangle \langle v, (d^2+\lambda^2)^{-1}  v\rangle \bigg] \times \label{eq:melem}\\
& \quad \times \frac{1}{1+ 2g \langle v, \frac{2d}{d^2+\lambda^2} v\rangle +4g^2 \lvert \langle v, (d+i\lambda)^{-1} v\rangle\rvert^2}\;. \label{eq:frac}
\end{align}
In the denominator $1+ 2g \langle v, \frac{2d}{d^2+\lambda^2} v\rangle +4g^2 \lvert \langle v, (d+i\lambda)^{-1} v\rangle\rvert^2 \geq 1$; hence \cref{eq:frac} can be dropped for an upper bound. Inserting \cref{eq:melem} in \cref{eq:L2} we obtain
\begin{align}
&\lvert (L_2)_{\alpha,\beta} \rvert \label{eq:L2bound}\\
& \leq gC \int_0^\infty\!\!\! \di \lambda\, \Big\lvert \frac{u_\alpha^2 v_\alpha}{u_\alpha^4+\lambda^2}   \left( 1+2g \langle v, \frac{d}{d^2 + \lambda^2} v\rangle \right) + 2g \frac{\lambda^2 v_\alpha}{u_\alpha^4 + \lambda^2} \langle v, (d^2+\lambda^2)^{-1}  v\rangle \Big\rvert \times \nonumber \\ 
& \hspace{4.45em} \times \Big\lvert\frac{u_\beta^2 v_\beta}{u_\beta^4+\lambda^2}   \left( 1+2g \langle v, \frac{d}{d^2 + \lambda^2} v\rangle \right) + 2g\frac{\lambda^2 v_\beta}{u_\beta^4+\lambda^2} \langle v, (d^2+\lambda^2)^{-1}  v\rangle \Big\rvert \;. \nonumber
\end{align}
% (The sub-index $\alpha$ indicates the element of the vector given in ket notation.)

For the following estimates, note that 
$\langle v, \frac{d}{d^2+\lambda^2} v\rangle \leq \langle v, \frac{1}{d} v\rangle \leq C$ according to \cref{eq:A23}.

%%%%%%%%%%%%%%%%%%%%%%%%%%%%%%%%
\paragraph{First summand times first summand in \cref{eq:L2bound}.} Consider the product of the first summands from inside each of the absolute values. This is of the same type as \cref{eq:L1}, so
\begin{align*}
& gC \int_0^\infty \di \lambda \frac{u_\alpha^2 v_\alpha}{u_\alpha^4+\lambda^2}   \left( 1+2g \langle v, \frac{d}{d^2 + \lambda^2} v\rangle \right)^2 \frac{u_\beta^2 v_\beta}{u_\beta^4 + \lambda^2} \\
& \leq \frac{gC}{M}
% (1+8\pi g)^2
u_\alpha^3 u_\beta^3 \int_0^\infty \di \lambda \frac{1}{u_\alpha^4+\lambda^2} \frac{1}{u_\beta^4 + \lambda^2} \leq \frac{gC}{M}
% (1+8\pi g)^2
\min\left\{ \frac{u_\alpha}{u_\beta}, \frac{u_\beta}{u_\alpha} \right\}\;.
\end{align*}

%%%%%%%%%%%%%%%%%%%%%%%%%%%%
\paragraph{Second summand times second summand in \cref{eq:L2bound}.}
We have $\frac{1}{2 d \lambda} \geq \frac{1}{d^2+\lambda^2}$, which implies $\langle v, (d^2+\lambda^2)^{-1} v\rangle \leq (2\lambda)^{-1} \langle v,d^{-1} v\rangle \leq C \lambda^{-1}$. Thus
\begin{align}
g^3 C \int_0^\infty \di \lambda \frac{\lambda^2 v_\alpha}{u_\alpha^4 + \lambda^2} \langle v, (d^2+\lambda^2)^{-1}  v\rangle^2 \frac{\lambda^2 v_\beta}{u_\beta^4 + \lambda^2} \leq \frac{g^3 C}{M} \int_0^\infty \di\lambda \frac{u_\alpha}{u_\alpha^4+\lambda^2} \frac{u_\beta}{u_\beta^4 + \lambda^2} \lambda^2\;. \label{eq:662}
\end{align}
Assuming without loss of generality $u_\alpha \geq u_\beta$, by dropping a non--negative $u_\beta^4$ from the numerator we find
\begin{align*}
\textnormal{\cref{eq:662}} & \leq \frac{g^3 C}{M} \int_0^\infty \di\lambda  \frac{u_\alpha}{u_\alpha^4+\lambda^2} \frac{u_\beta}{\lambda^2} \lambda^2 = \frac{g^3 C}{M} u_\beta \int_0^\infty \di\lambda  \frac{u_\alpha}{u_\alpha^4+\lambda^2}\\
& = \frac{g^3 C}{M} \frac{u_\alpha u_\beta}{u_\alpha^4} \int_0^\infty u_\alpha^2 \di \rho \frac{1}{1+\rho^2} = \frac{g^3 C}{M} \frac{u_\beta}{u_\alpha}\;.
\end{align*}

%%%%%%%%%%%%%%%%%%%%%%%%%%%%%%%%%%%%
\paragraph{Mixed term in \cref{eq:L2bound}.}
We turn to the remaining two terms obtained from the product in the integral, for which we have to estimate
\begin{align}
& g^2C \int_0^\infty \di \lambda \frac{u_\alpha u_\beta (u_\alpha^2 + u_\beta^2)}{(u_\alpha^4 + \lambda^2)(u_\beta^4 + \lambda^2)} \lambda^2 \langle v, \frac{1}{d^2+\lambda^2} v\rangle \nonumber \\
& \leq \frac{g^2C}{M} u_\alpha u_\beta (u_\alpha^2 + u_\beta^2) \sum_{\gamma=1}^{\ik} \frac{C}{M} u_\gamma^2 \int_0^\infty \di\lambda \frac{\lambda^2}{(u_\alpha^4 + \lambda^2)(u_\beta^4 + \lambda^2)(u_\gamma^4 + \lambda^2)} \;. \label{eq:fish}
\end{align}
The integral is
\[
\int_0^\infty \di\lambda \frac{\lambda^2}{(u_\alpha^4 + \lambda^2)(u_\beta^4 + \lambda^2)(u_\gamma^4 + \lambda^2)} = \frac{\pi}{2} \frac{1}{(u_\alpha^2 + u_\beta^2)(u_\alpha^2 + u_\gamma^2)(u_\beta^2 + u_\gamma^2)}\;.
\]
Without loss of generality $u_\alpha \leq u_\beta$; then
\[
\textnormal{\cref{eq:fish}}  = \frac{g^2C}{M}\frac{u_\alpha}{ u_\beta} \frac{1}{M} \sum_{\gamma=1}^{\ik} \frac{u_\gamma^2}{u_\gamma^2 + u_\alpha^2} \frac{u_\beta^2}{u_\beta^2 + u_\gamma^2} \leq \frac{g^2C}{M} \frac{u_\alpha}{u_\beta} \;. 
\]
In the last step we used that both fractions in the sum are bounded by $1$. This completes the proof of \cref{lem:K}. 
\end{proof}

\section{Approximate Bogoliubov Transformation}

Given the Bogoliubov kernel $K(k)$ in \cref{eq:Kk}, for any $\lambda \in \Rbb$ we define a unitary transformation $T_\lambda: \fock \to \fock$ by
\begin{equation}\label{eq:Bog-T}
T_\lambda := \exp \Big( \lambda \sum_{k\in \north }\frac{1}{2}\sum_{\alpha,\beta\in\Ik} K(k)_{\alpha,\beta} c^*_\alpha(k) c^*_\beta(k) - \hc \Big)\;.
\end{equation}

The following lemma is the main result of this section, showing that $T_\lambda$ acts approximately as a bosonic Bogoliubov transformation. 
 \begin{lem}[Approximate Bogoliubov Transformation]\label{lem:Bog-T}
 For all $\lambda\in[-1,1]$, $l \in\north $ and $\gamma \in \Il$, we have
\begin{align} \label{eq:Bog_Tlambda}
T^*_\lambda c_\gamma(l) T_\lambda &= \sum_{\alpha\in\Il}\cosh(\lambda K(l))_{\alpha,\gamma} c_\alpha(l) + \sum_{\alpha\in\Il} \sinh(\lambda K(l))_{\alpha,\gamma} c^*_\alpha(l) + \mathfrak{E}_\gamma(\lambda,l)
\end{align}
where the error operator $\mathfrak{E}_\gamma(\lambda,l)$ satisfies (with a constant $C$ independent of $\lambda$ and $l$)
\begin{align*}
\sum_{\gamma\in \Il} \norm{\mathfrak{E}_\gamma(\lambda,l) \psi}  \leq C M N^{-\frac{2}{3}+\delta}  \norm{(\Ncal_\delta+M)^{1/2} (\Ncal +1)\psi}\;, \qquad \forall \psi \in \fock\;. 
\end{align*}
The same estimate holds for $\mathfrak{E}^*_\gamma(\lambda,l)$. 
\end{lem}

Here the matrices $\cosh(K(k))$ and $\sinh(K(k))$ are defined by functional calculus, or more explicitly by the series
\begin{align*}
\cosh(K(k)) = \sum_{n=0}^\infty \frac{K(k)^{2n}}{(2n)!}\;, \qquad \sinh(K(k)) = \sum_{n=0}^\infty \frac{K(k)^{2n+1}}{(2n+1)!}\;.
\end{align*}
As a consequence of \cref{lem:K}, by a calculation similar to that in \cref{eq:powerseries} using the power series, we can verify that, for all $\lambda\in [-1,1]$ and $k\in \north$, 
\begin{align} \label{eq:ch-sh}
|\cosh(\lambda K(k))_{\alpha,\beta}-  \delta_{\alpha,\beta}| + \lvert\sinh(\lambda K(k))_{\alpha,\beta}\rvert \leq \frac{C \hat V(k)}{M}\min\left\{ \frac{n_\alpha(k)}{n_\beta(k)} , \frac{n_\beta(k)}{n_\alpha(k)}\right\}\;. 
\end{align}

In order to prove \cref{lem:Bog-T}, we need to show that the fermion number is stable under the approximate Bogoliubov transformation. This is the content of the next lemma.

\begin{lem}[Stability of Fermion Number] \label{lem:stability} Let $T_\lambda$ be the approximate Bogoliubov transformation defined in \cref{eq:Bog-T}. For all $\lambda \in [-1,1]$ and all $\psi \in \fock$ we have
 \begin{align}
%   T^*_\lambda \Ncal_{\delta} T_\lambda   &\leq C  (\Ncal_{\delta}+1) \\
 T^*_\lambda \Ncal^m T_\lambda  &\leq C_m (\Ncal+1)^m  ,\quad\forall m\geq 1\;. 
 \label{eq:furtherbound1}\\
  T^*_\lambda \Ncal_{\delta} \Ncal^m T_\lambda  &\leq C_m  (\Ncal_{\delta}+1)(\Ncal+1)^m \;,\quad\forall m\geq 0\;. \label{eq:stability1} 
\end{align}
\end{lem}

\begin{proof}  Let $\psi \in \fock$. We use a variation of the Gr\"onwall argument in \cite{BPS14,BNP+20}. 
\paragraph{Proof of \cref{eq:furtherbound1}.} 
For any function $f$ we have $f(\Ncal) c_\alpha^*(k) c_\beta^*(k)   = c_\alpha^*(k) c_\beta^*(k) f(\Ncal+4)$. Thus  
\begin{align*}
&[(\Ncal +4)^m, c^*_\alpha(k)c^*_\beta(k)  ]  =   \Big( (\Ncal +4)^m- \Ncal^m\Big) c^*_\alpha(k)c^*_\beta(k)  \\
&=  \Big( (\Ncal +4)^m- \Ncal^m\Big)^{1/2}  c^*_\alpha(k)c^*_\beta(k)  \Big( (\Ncal +8)^m- (\Ncal +4)^m\Big)^{1/2}\;.
\end{align*}
 Recall from \cref{lem:K} that $\lvert K(k)_{\alpha,\beta} \rvert\leq C/M$, so that 
%  \[
%  \norm{K(k)}\HS = \Big(\sum_{\alpha,\beta \in \Ik} \lvert K(k)_{\alpha,\beta} \rvert^2 \Big)^{1/2} \leq C\;.\]
 \cref{eq:bosebound1} from \cref{lem:bosonic-number} implies
\begin{align*}
\Bignorm{ \sum_{\beta\in \Ik} K(k)_{\alpha,\beta} c^*_{\beta}(k) T_\lambda \psi} &\leq CM^{-\frac{1}{2}}  \norm{(\Ncal_{\delta} + 1)^{1/2} T_\lambda\psi}\;.
% \Bignorm{ \sum_{\beta\in \Ik} K(k)_{\alpha,\beta} c^{g*}_{\beta}(k) T_\lambda \psi} &\leq CM^{-\frac{1}{2}}  \norm{(\Ncal_{\delta} + 1)^{1/2} T_\lambda\psi} \;,\\
\end{align*}
Furthermore, by \cref{eq:boseboundX} from \cref{lem:bosonic-number} we have
\begin{align*}
  \sum_{\alpha \in \Ik} \norm{ c_{\alpha}(k)T_\lambda \psi}  &\leq C M^{\frac{1}{2}} \norm{\Ncal_{\delta}^{1/2} T_\lambda \psi} \;.
%   \\
%     \sum_{\alpha \in \Ik} \norm{ c^g_{\alpha}(k)T_\lambda \psi}  &\leq C M^{\frac{1}{2}} \norm{\Ncal_{\delta}^{1/2} T_\lambda \psi}\;.
 \end{align*}
Hence for all vectors $X,Y \in \fock$ we have
\begin{align*}
\Big\lvert  \sum_{\alpha,\beta \in \Ik} K(k)_{\alpha,\beta} \langle X , c^*_\alpha(k) c^*_\beta(k)  Y \rangle \Big\rvert \leq C \norm{\Ncal_\delta^{1/2}X} \norm{(\Ncal_\delta+1)^{1/2}Y}\;.
 \end{align*}
 Using this bound we find
\begin{align*}& \left\lvert \frac{\di}{\di\lambda} \langle T_\lambda \psi , (\Ncal +4)^m T_\lambda \psi  \rangle \right\rvert\\
 & = \Big\lvert \Re \sum_{k \in \north } \sum_{\alpha,\beta \in \Ik} K(k)_{\alpha,\beta} \langle T_\lambda \psi , [(\Ncal +4)^m, c^*_\alpha(k)c^*_\beta(k)  ] T_\lambda \psi \rangle \Big\rvert \\
 &= \Big\lvert \Re \sum_{k \in \north } \sum_{\alpha,\beta \in \Ik} K(k)_{\alpha,\beta}\; \times\\
 &\quad \times \Big\langle T_\lambda \psi , \Big( (\Ncal +4)^m-\Ncal^m \Big)^{1/2} c^*_\alpha(k)c^*_\beta(k)  \Big( (\Ncal +8)^m-(\Ncal+4)^m \Big)^{1/2}  T_\lambda \psi \Big\rangle \Big\rvert\\
 &\leq  C \norm{ \Ncal_\delta^{1/2}   \Big( (\Ncal +4)^m-\Ncal^m \Big)^{1/2}   T_\lambda \psi} \norm{(\Ncal_\delta+1)^{1/2} \Big( (\Ncal +8)^m-(\Ncal+4)^m \Big)^{1/2} T_\lambda \psi}\\
 &\leq C_m \norm{(\Ncal +4)^{m/2} T_\lambda \psi}^2\;. \tagg{gronwall}
 \end{align*}
 In the last estimate we used that $\Ncal_\delta$ commutes with $\Ncal$, that $0\leq \Ncal_\delta \leq \Ncal$, and that 
 \begin{align*}
 0&\leq (\Ncal +4)^m-\Ncal^m  \leq C_m (\Ncal +4)^{m-1}\;,\\
 0&\leq (\Ncal +8)^m-(\Ncal+4)^m \leq C_m (\Ncal+4)^{m-1}\;. 
 \end{align*}
 The estimate \cref{eq:gronwall} closes a Gr\"onwall bound for $\langle T_\lambda \psi , (\Ncal +4)^m T_\lambda \psi  \rangle$ and \cref{eq:furtherbound1} follows.  
 
\paragraph{Proof of \cref{eq:stability1}.}
 In view of definition \cref{def:weightedpair} we can write 
 \begin{equation}
  \label{eq:Nccommute}
 [\Ncal_\delta, c^*_{\alpha}(k)] = c^{g*}_{\alpha}(k)
 \end{equation}
 for some weight function $g$ with $g(p,k)\in \{0,1,2\}$.
Similarly to the above calculation
\begin{align*}
&[(\Ncal_\delta +1)(\Ncal +4)^m, c^*_\alpha(k)c^*_\beta(k)  ] \\
&= [\Ncal_\delta, c^*_\alpha(k)c^*_\beta(k)] (\Ncal +4)^m + (\Ncal_\delta +1)[(\Ncal +4)^m, c^*_\alpha(k)c^*_\beta(k)]\\
&=\Big( c^{g*}_{\alpha}(k) c^*_{\beta}(k) +  c^{*}_{\alpha}(k) c^{g*}_{\beta}(k) \Big) (\Ncal +4)^m  + (\Ncal_\delta +1) \Big( (\Ncal +4)^m - \Ncal^m\Big) c^*_\alpha(k)c^*_\beta(k)\\
&= \Ncal^{m/2} \Big( c^{g*}_{\alpha}(k) c^*_{\beta}(k) +  c^{*}_{\alpha}(k) c^{g*}_{\beta}(k) \Big) (\Ncal +4)^{m/2}\\
&\quad + (\Ncal_\delta +1)  (\Ncal+1)^{-1/2} \Big( (\Ncal +4)^m - \Ncal^m \Big)^{1/2}  \times \\ 
&\hspace{9em}  \times  c^*_\alpha(k)c^*_\beta(k) \Big( (\Ncal +8)^m - (\Ncal+4)^m \Big)^{1/2} (\Ncal+5)^{1/2}\;. 
\end{align*}
Similar to the bounds used above, \cref{lem:bosonic-number-weight} provides us with  
\begin{align*}
% \Bignorm{ \sum_{\beta\in \Ik} K(k)_{\alpha,\beta} c^*_{\beta}(k) T_\lambda \psi} &\leq CM^{-\frac{1}{2}}  \norm{(\Ncal_{\delta} + 1)^{1/2} T_\lambda\psi}\;,\\
 \Bignorm{ \sum_{\beta\in \Ik} K(k)_{\alpha,\beta} c^{g*}_{\beta}(k) T_\lambda \psi} &\leq CM^{-\frac{1}{2}}  \norm{(\Ncal_{\delta} + 1)^{1/2} T_\lambda\psi} \;,\\
%   \sum_{\alpha \in \Ik} \norm{ c_{\alpha}(k)T_\lambda \psi}  &\leq C M^{\frac{1}{2}} \norm{\Ncal_{\delta}^{1/2} T_\lambda \psi} \;.
%   \\
    \sum_{\alpha \in \Ik} \norm{ c^g_{\alpha}(k)T_\lambda \psi}  &\leq C M^{\frac{1}{2}} \norm{\Ncal_{\delta}^{1/2} T_\lambda \psi}\;.
 \end{align*}
We then get
\begin{align*}& \left\lvert \frac{\di}{\di\lambda} \langle T_\lambda \psi , (\Ncal_\delta +1)(\Ncal +4)^m T_\lambda \psi  \rangle \right\rvert\\
 & = \Big\lvert \Re \sum_{k \in \north } \sum_{\alpha,\beta \in \Ik} K(k)_{\alpha,\beta} \langle T_\lambda \psi , [(\Ncal_\delta +1)(\Ncal +4)^m, c^*_\alpha(k)c^*_\beta(k)  ] T_\lambda \psi \rangle \Big\rvert \displaybreak[0]\\
 &\leq  \Big\lvert  \sum_{k \in \north } \sum_{\alpha,\beta \in \Ik} K(k)_{\alpha,\beta} \Big\langle T_\lambda \psi ,  \Ncal^{m/2} \Big( c^{g*}_{\alpha}(k) c^*_{\beta}(k) +  c^{*}_{\alpha}(k) c^{g*}_{\beta}(k) \Big) (\Ncal +4)^{m/2}  T_\lambda \psi \Big\rangle \Big\rvert \\
 &\quad+ \Big\lvert  \sum_{k \in \north } \sum_{\alpha,\beta \in \Ik} K(k)_{\alpha,\beta} \Big\langle T_\lambda \psi , (\Ncal_\delta +1)  (\Ncal+1)^{-1/2} \Big( (\Ncal +4)^m - \Ncal^m \Big)^{1/2}  \times \\
 &\hspace{12em}  \times  c^*_\alpha(k)c^*_\beta(k) \Big( (\Ncal +8)^m - (\Ncal+4)^m \Big)^{1/2} (\Ncal+5)^{1/2} T_\lambda \psi \Big\rangle \Big\rvert \displaybreak[0]\\
 &\leq C \norm{(\Ncal_\delta+1)^{1/2} \Ncal^{m/2}  T_\lambda \psi} \norm{(\Ncal_\delta+1)^{1/2} (\Ncal +4)^{m/2}  T_\lambda \psi} \\
 &\quad + C \norm{ (\Ncal_\delta+1)^{1/2} (\Ncal_\delta +1)  (\Ncal+1)^{-1/2} \Big( (\Ncal +4)^m -\Ncal^m \Big)^{1/2}  T_\lambda \psi} \times\\
 &\hspace{10em} \times \norm{ (\Ncal_\delta+1)^{1/2} \Big( (\Ncal +8)^m - (\Ncal+4)^m \Big)^{1/2} (\Ncal+5)^{1/2}  T_\lambda \psi}  \\
 &\leq C_m \norm{(\Ncal_\delta+1)^{1/2} (\Ncal +1)^{m/2}  T_\lambda \psi}^2\;.
 \end{align*}
 Thus  \cref{eq:stability1} follows by Gr\"onwall's inequality. 
\end{proof}

\begin{proof}[Proof of \cref{lem:Bog-T}.] Let us denote the exponent of the Bogoliubov transformation by $B$, i.\,e., $T_\lambda = \exp(\lambda B)$. As in the proof of \cite[Prop.~4.4]{BNP+20}, we pick $n_0 \in \Nbb$ and iterate $n_0$ times the Duhamel expansion  
\[
T^*_\lambda c_\gamma(l) T_\lambda = c_\gamma(l) + \int_0^\lambda \di\tau T^*_\tau [c_\gamma(l),B] T_\tau
\]
and use the commutator formula 
\begin{equation}\label{eq:bla}
[c_\gamma(l),B]= \sum_{\alpha \in \Il} K(l)_{\gamma,\alpha}c_\alpha^*(l) + \frac{1}{2}\sum_{k \in \north} \sum_{\beta \in \Ik} K(k)_{\gamma,\beta} \left(\Ecal_\gamma(k,l) c^*_\beta(k) + c^*_\beta(k) \Ecal_\gamma(k,l)\right)
\end{equation}
which follows from \cref{lem:ccr}. When iterating the expansion, the term $K(l)_{\gamma,\alpha}c_\alpha^*$ in the commutator formula gives rise to the power series of the $\sinh$ and $\cosh$; the rest of \cref{eq:bla} will be considered as an error term. The conclusion is that \cref{eq:Bog_Tlambda} holds true, with the error term summed over all iteration steps being  
\begin{align} \label{eq:e-Bog-00}
\mathfrak{E}_\gamma(\lambda,l) & =  \sum_{n=0}^{n_0-1} \int_0^\lambda \di \tau \frac{(\lambda-\tau)^n}{n!}  \sum_{k\in \north}\sum_{\alpha\in\Il \cap \Ik}  \sum_{\beta\in \Ik}\left(K(l)^n\right)_{\gamma,\alpha} K(k)_{\alpha,\beta}  \times \\
&\hspace{16em}  \times T^*_{\tau}\frac{1}{2}\left(\Ecal_\alpha(k,l) c^*_\beta(k) + c^*_\beta(k) \Ecal_\alpha(k,l) \right)^{\natural} T_{\tau} \nonumber\\
& \quad +  \int_0^\lambda \di \tau \frac{(\lambda-\tau)^{n_0-1}}{(n_0-1)!} \sum_{\alpha\in\Il} \left( K(l)^{n_0} \right)_{\gamma,\alpha} T^*_{\tau} c^\natural_\alpha(l) T_{\tau} - \sum_{\alpha \in \Il} \sum_{n = n_0}^\infty \frac{\lambda^n (K(l)^n)_{\gamma,\alpha}}{n!} c^\natural_{\alpha} (l)  \nonumber
\end{align}
for any $n_0 \geq 1$. (The two terms on the last line are the non--explicit integral term from the Duhamel formula and the powers missing to complete the series of the $\sinh$ and $\cosh$ from $n=n_0$ to $+\infty$.) Here in every summand $X^\natural$ means  $X$ for $n$ even and $X^*$ for $n$ odd. 

Recall that as a consequence of \cref{lem:K} we have
\[ \lvert \left(K(l)^n\right)_{\gamma,\alpha}\rvert \leq \left\{
                \begin{array}{ll}
                   \delta_{\gamma,\alpha}  &\text{ if }n=0\;,\\
                   C^n M^{-1}  &\text{ if }n \geq 1\;.
                \end{array}
                \right.
\]
Therefore, by the triangle inequality and summing over $\gamma$,  
\begin{align*}
\sum_{\gamma\in \Il} \norm{\mathfrak{E}_\gamma(\lambda,l) \psi} & \leq \sum_{n=0}^{n_0-1} \frac{C^n}{n! M} \int_0^\lambda \di \tau \sum_{k\in \north}\sum_{\alpha\in\Il \cap \Ik}  \sum_{\beta\in \Ik} \Big[ \norm{\Ecal_\alpha(k,l) c^*_\beta(k) T_{\tau} \psi}  +  \nonumber \\
&\hspace{4.5em}+ \norm{c^*_\beta(k) \Ecal_\alpha(k,l) T_{\tau} \psi} + \norm{ c_\beta(k) \Ecal^*_\alpha(k,l) T_{\tau} \psi}  + \norm{ \Ecal^*_\alpha(k,l) c_\beta(k)  T_{\tau} \psi} \Big]\nonumber \\
& \quad + \frac{C^{n_0}}{(n_0-1)!}   \int_0^\lambda \di \tau  \sum_{\alpha\in\Il} \Big( \norm{c_\alpha(l) T_{\tau}  \psi} + \norm{c^*_\alpha(l) T_{\tau}  \psi}  \Big) \nonumber\\
& \quad+ \sum_{n = n_0}^\infty \frac{C^n}{n!} \sum_{\alpha \in \Il} \Big( \norm{ c_\alpha(l)  \psi} + \norm{ c^*_\alpha(l)  \psi}  \Big)\; . 
\end{align*}
Using \cref{lem:bosonic-number} we can bound the operators in the last two lines; then taking $n_0\to \infty$ we obtain 
\begin{align}
\sum_{\gamma\in \Il} \norm{\mathfrak{E}_\gamma(\lambda,l) \psi} & \leq C M^{-1} \int_0^\lambda \di \tau \sum_{k\in \north}\sum_{\alpha\in\Il \cap \Ik}  \sum_{\beta\in \Ik} \Big[ \norm{\Ecal_\alpha(k,l) c^*_\beta(k) T_{\tau} \psi}  \nonumber \\
&\hspace{4em}+ \norm{c^*_\beta(k) \Ecal_\alpha(k,l) T_{\tau} \psi} + \norm{ c_\beta(k) \Ecal^*_\alpha(k,l) T_{\tau} \psi}  + \norm{ \Ecal^*_\alpha(k,l) c_\beta(k)  T_{\tau} \psi} \Big] \nonumber \\& =:I_1+I_2+I_3+I_4\;. \label{eq:sum-Ebog-01}
\end{align}

It remains to bound \cref{eq:sum-Ebog-01} term by term. For the first term, using \cref{lem:ccr}, \cref{lem:bosonic-number} and \cref{lem:stability}  we can estimate 
\begin{align}
I_1 &\leq \sup_{\tau \in [-1,1]} C M^{-1}\sum_{k\in \north }\sum_{\alpha\in\Il \cap \Ik} \sum_{\beta\in \Ik}  \norm{\Ecal_\alpha(k,l)  c^*_\beta(k)  T_{\tau} \psi} \nonumber \\
&\leq \sup_{\tau \in [-1,1]} CM^{-1} \sum_{k\in \north } \sum_{\beta\in \Ik}  M^{\frac{3}{2}}N^{-\frac{2}{3}+\delta}  \norm{\Ncal c^*_\beta(k)  T_{\tau} \psi} \nonumber\\
&= \sup_{\tau \in [-1,1]} CM^{\frac{1}{2}}N^{-\frac{2}{3}+\delta}  \sum_{k\in \north}  \sum_{\beta\in \Ik} \norm{c^*_\beta(k) (\Ncal +2) T_{\tau} \psi} \nonumber\\
&\leq \sup_{\tau \in [-1,1]} C M N^{-\frac{2}{3}+\delta}   \norm{(\Ncal_\delta+M)^{1/2}(\Ncal+2) T_{\tau} \psi} \nonumber\\
&\leq C M N^{-\frac{2}{3}+\delta}   \norm{(\Ncal_\delta+M)^{1/2}(\Ncal+1) \psi}\; . \label{eq:sum-Ebog-01-1}
\end{align}
For $I_2$, if $\beta\ne \alpha$, then $c^*_\beta(k)\Ecal_\alpha(k,l)=\Ecal_\alpha(k,l) c^*_\beta(k)$ which is just what we bounded in \cref{eq:sum-Ebog-01-1}. If $\beta=\alpha$ we have the decomposition 
\begin{align} \label{eq:c-eps-com}
c^*_\alpha(k)\Ecal_\alpha(k,l) = \Ecal_\alpha(k,l)  c^*_\alpha(k) -   \Ecal_{\alpha}(k,k) c^*_\alpha(l) + c^*_\alpha(l) \Ecal_{\alpha}(k,k)\; .  
\end{align}
(This decomposition can be verified by noticing that \cref{eq:c-eps-com} consists of two commutators,  recalling that $\Ecal_\alpha(k,l) = [c_\alpha(k),c^*_\alpha(l)] - \delta_{k,l}$,  and using the Jacobi identity.)
With this decomposition we proceed to
\begin{align}
M^{-1}\sum_{\alpha\in \Il \cap \Ik}\norm{\Ecal_\alpha(k,l)  c^*_\alpha(k)T_\tau \psi} &\leq \sum_{\alpha\in \Il \cap \Ik} CN^{-\frac{2}{3}+\delta} \norm{ \Ncal c^*_\alpha(k) T_\tau \psi} \nonumber\\
&\leq C M^{\frac{1}{2}} N^{-\frac{2}{3}+\delta} \norm{ (\Ncal_\delta+M)^{1/2} (\Ncal+1) \psi} \;.
\end{align}
Here we have used the first bound from \cref{lem:ccr} (without the summation), \cref{lem:bosonic-number} and then proceeded similarly to \cref{eq:sum-Ebog-01-1}. The second term of \cref{eq:c-eps-com} is treated similarly. For the last term of \cref{eq:c-eps-com}, we bound $\norm{c^*_\alpha(l)\xi} \leq \norm{(\Ncal_\delta+1)^{1/2}\xi}$ (which follows from \cref{eq:bosebound1} with $f=\delta_\alpha$) and then use the fact that  $\Ecal_{\alpha}(k,k)$ commutes with $\Ncal_\delta$ (this is clear from \cref{eq:Ekk} with $k =l$) to get
\begin{align*}
M^{-1}\sum_{\alpha\in\Il \cap \Ik}  \norm{ c^*_\alpha(l) \Ecal_{\alpha}(k,k)  T_{\tau} \psi} 
&\leq M^{-1} \sum_{\alpha\in\Il \cap \Ik}  \norm{(\Ncal_\delta +1)^{1/2} \Ecal_{\alpha}(k,k)  T_\tau \psi}  \\
&= M^{-1} \sum_{\alpha\in\Il \cap \Ik}  \norm{\Ecal_{\alpha}(k,k)  (\Ncal_\delta +1)^{1/2}  T_\tau \psi} \\
&\leq CM^{-1} M^{\frac{3}{2}}  N^{-\frac{2}{3}+\delta}  \norm{ \Ncal (\Ncal_\delta +1)^{\frac{1}{2}}  T_{\tau} \psi} \\
&\leq C M^{\frac{1}{2}} N^{-\frac{2}{3}+\delta}     \norm{ (\Ncal_\delta+1)^{\frac{1}{2}} (\Ncal +1) \psi}\;.
\end{align*}
In the last estimate we used  \cref{lem:stability}  again.
% Thus from \cref{eq:c-eps-com} we have
% \begin{align}\label{eq:c-eps-com-final}
% M^{-1} \sum_{\alpha\in\Il \cap \Ik}  \norm{ c^*_\alpha(k)\Ecal_\alpha(k,l)  T_{\tau} \psi} \leq  C M^{\frac{1}{2}} N^{-\frac{2}{3}+\delta} \norm{ (\Ncal_\delta+M)^{1/2} (\Ncal+1) \psi}\;.  
% \end{align}
Therefore
% , similarly to  \cref{eq:sum-Ebog-01-1} we have
\begin{align}
I_2 \leq C M N^{-\frac{2}{3}+\delta}   \norm{(\Ncal_\delta+M)^{1/2}(\Ncal+1) \psi}\; . \label{eq:sum-Ebog-02}
\end{align}
By the same argument, we obtain similar bounds for $\Ecal^*_\alpha(k,l) c_\beta(k)$ and $c_\beta(k) \Ecal^*_\alpha(k,l)$ (recall that $\Ecal^*_\alpha(k,l) = \Ecal_\alpha(l,k)$ to reduce to the previous estimates). Collecting the estimates for the terms of \cref{eq:sum-Ebog-01}, this completes the proof of \cref{lem:Bog-T}. 
\end{proof}

\section{Linearization of the Kinetic Energy}

In this section we prove that the fermionic kinetic energy $\Hbb_0$ behaves similarly to the bosonized kinetic energy
\begin{equation}
\Dbb_\B :=  2\hbar \kappaf \sum_{k\in \north} \sum_{\alpha\in \Ik} \lvert k \cdot \hat{\omega}_\alpha\rvert  c^*_\alpha(k) c_\alpha(k)\;.
\label{eq:DBdef}
\end{equation}
Recall the definition of the Fermi momentum $k_\F=\kappaf \hbar^{-1}=(\frac{3}{4\pi})^{\frac{1}{3}}N^{\frac{1}{3}}$. 
The main result of the section is that the difference of fermionic and approximately bosonic kinetic energy is almost invariant under the approximate Bogoliubov transformation.
\begin{lem}[Comparing Fermionic and Bosonized Kinetic Energy]\label{lem:lin-2}
Let $T_\lambda$ be the approximate Bogoliubov transformation defined in \cref{eq:Bog-T}. Then for all $\psi\in \fock$ we have
\begin{align} 
% &\Big\langle \Psi, (\Hbb_0 - \Dbb_\B) \Psi \Big\rangle -  \Big\langle \psi, (\Hbb_0 - \Dbb_\B) \psi \Big\rangle  \nonumber \\
& \Big\langle T_1\psi, (\Hbb_0 - \Dbb_\B) T_1\psi \Big\rangle  - \Big\langle  \psi, (\Hbb_0 - \Dbb_\B)  \psi \Big\rangle  \nonumber\\
&\geq  - C\hbar \Big[  M^{-\frac{1}{2}} \norm{  (\Ncal_\delta+1)^{1/2}  \psi  }^2 + M N^{-\frac{2}{3}+\delta}  \norm{ (\Ncal_\delta+1)^{1/2} (\Ncal+1) \psi} \norm{ (\Ncal_\delta+1)^{1/2}   \psi } \Big]\;. \nonumber
% \\
% &\geq - C\hbar \Big[  M^{-\frac{1}{2}} N^\delta + M^{\frac{3}{2}}N^{-\frac{1}{3}+2\delta +\bourgain } \Big]\;.\label{eq:ferror-3}
\end{align}
\end{lem}
The first of the two error terms in the lemma is the main reason for taking $M$ large. More precisely, the first error (involving $M^{-\frac{1}{2}}$) comes from linearizing the fermionic kinetic energy on each patch, for which the size of the patch must not be too large. The linearization error of the fermionic kinetic energy is estimated in the following lemma.

\begin{lem}[Linearization of Kinetic Energy]\label{lem:lin}
 For all $k \in \north$ and all $\alpha \in \Ik$ we have
\[
 [\Hbb_0,c^*_{\alpha}(k)] = 2\hbar \kappaf \lvert k \cdot \hat{\omega}_\alpha\rvert   c^*_{\alpha}(k) + \hbar \Efrak^{\textnormal{lin}}_\alpha(k)^*
\]
where the error term satisfies, for all $\psi\in \fock$,  
\[
  \sum_{\alpha\in \Ik} \norm{ \Efrak^{\textnormal{lin}}_\alpha(k)  \psi} \leq C \norm{\Ncal_\delta^{1/2} \psi}\;, \qquad \sum_{\alpha\in \Ik} \norm{ \Efrak^{\textnormal{lin}}_\alpha(k)^*  \psi} \leq C \norm{(\Ncal_\delta+M)^{1/2} \psi}\;,
\]
and 
\[
   \norm{  \sum_{\alpha\in \Ik} f_\alpha \Efrak^{\textnormal{lin}}_\alpha(k)^*  \psi} \leq C M^{-\frac{1}{2}} \norm{f}_{\ell^2}  \norm{(\Ncal_\delta+1)^{1/2} \psi} \qquad \forall f\in \ell^2(\Ik)\;.
\]
\end{lem}
All bounds here gain a factor $M^{-\frac{1}{2}}$ compared to the bounds in \cref{lem:bosonic-number}.
\begin{proof} Let us consider $\alpha \in \Ikp$, the other case is similar. By the CAR \cref{eq:car} we have $[a_i^* a_i, a^*_p] = \delta_{i,p} a^*_p$. Therefore
 \begin{align*}
 [\Hbb_0,c^*_{\alpha}(k)] & =  \Big[ \sum_{i\in \mathbb{Z}^3} e(i) a_i^* a_i, \frac{1}{n_{\alpha}(k)} \sum_{\substack{p\colon p\in \BFc \cap B_\alpha\\p-k \in \BF \cap B_\alpha}} a^*_p a^*_{p-k} \Big]\\
%  & = \frac{1}{n_{\alpha}(k)} \sum_{\substack{p\colon p\in \BFc \cap B_\alpha\\p-k \in \BF \cap B_\alpha}}  \sum_{i\in \mathbb{Z}^3} e(i) \Big( [a_i^* a_i, a^*_p]a_{p-k}^* + a^*_p [a_i^* a_i, a_{p-k}^*]  \Big) \\
  & = \frac{1}{n_{\alpha}(k)}   \sum_{\substack{p\colon p\in \BFc \cap B_\alpha\\p-k \in \BF \cap B_\alpha}} (e(p)+e(p-k))a^*_p a^*_{p-k} \quad  = 2\hbar \kappaf \lvert k\cdot \hat{\omega}_\alpha\rvert c_\alpha^*(k) + \hbar \Efrak^{\textnormal{lin}}_\alpha(k)^* \;,   \end{align*}
  where, 
%   \[
%   \Efrak^{\textnormal{lin}}_\alpha(k) = \frac{1}{n_{\alpha}(k)}   \sum_{\substack{p\colon p\in \BFc \cap B_\alpha\\p-k \in \BF \cap B_\alpha}} \hbar^{-1} (e(p)+e(p-k) - 2\hbar \kappaf \lvert k\cdot \hat{\omega}_\alpha\rvert) a_{p-k} a_p\;.
%   \]
using definition \cref{def:weightedpair}, we can write $\Efrak^{\textnormal{lin}}_\alpha(k) = c^{g}_{\alpha}(k)$  with the weight function
% (for the case $\alpha\in \Ikp$)
\begin{align*}
 g(p,k) = \hbar^{-1} \Big[e(p)+e(p-k) - 2\hbar  \kappaf \lvert k\cdot \hat{\omega}_\alpha\rvert \Big]  = \hbar \Big[ 2 k\cdot (p -k_\F \hat{\omega}_\alpha) -\lvert k\rvert^2 \Big]\;.
\end{align*}
Since $\diam(B_\alpha)\leq CN^{\frac{1}{3}}M^{-\frac{1}{2}}$ and $\lvert k\rvert\leq C$ we can bound 
\begin{align} \label{eq:lin-g-diam}
\lvert g(p,k)\rvert \leq C \hbar N^{\frac{1}{3}}M^{-\frac{1}{2}} = CM^{-\frac{1}{2}}\;. 
\end{align}
The claimed error estimates now follow from \cref{lem:bosonic-number-weight}.
\end{proof}

\begin{proof}[Proof of \cref{lem:lin-2}]
Recall that $T_0 = \id$. We will show that for all $\lambda \in [0,1]$ we have
\begin{align*}& \Big\lvert\frac{\di}{\di\lambda} \Big\langle T_\lambda \psi, (\Hbb_0 - \Dbb_\B) T_\lambda \psi \Big\rangle \Big\rvert \tagg{final81}\\
&\leq C\hbar \Big[ M^{-\frac{1}{2}} \norm{  (\Ncal_\delta+1)^{1/2}  \psi  }^2 + MN^{-\frac{2}{3}+\delta}  \norm{ (\Ncal_\delta+1)^{1/2} (\Ncal+1) \psi} \norm{ (\Ncal_\delta+1)^{1/2}   \psi } \Big]\;.
\end{align*}
The claim then follows by integration over $\lambda \in [0,1]$.

Consider the approximately bosonic operator $\Dbb_\B$. We have
   \begin{align}
 [\Dbb_\B,c^*_{\alpha}(k)] & = 2\hbar \kappaf \sum_{l \in \north} \sum_{\gamma \in \Il} \lvert k \cdot \hat \omega_\gamma\rvert c_\gamma^*(l) \big[ c_\gamma(l),  c_\alpha^*(k)\big] \nonumber\\
 & = 2\hbar \kappaf \sum_{l \in \north} \sum_{\gamma \in \Il} \lvert k \cdot \hat \omega_\gamma\rvert c_\gamma^*(l) \delta_{\gamma,\alpha} (\delta_{k,l} + \Ecal_\alpha(l,k)) \nonumber\\
  & = 2\hbar \kappaf \lvert k \cdot \hat{\omega}_\alpha\rvert  c_\alpha^*(k) + \hbar \Efrak^{\linb}_\alpha(k)^*
  \end{align}
  where, with an indicator function $\chi(\alpha \in \Il)$,  the error term is
\begin{align} \label{eq:e-lin-d}
  \Efrak^{\linb}_\alpha(k) := 2 \kappaf \sum_{l \in \north}  \lvert k \cdot \hat{\omega}_\alpha\rvert   \Ecal^*_\alpha(l,k) c_\alpha(l) \chi(\alpha \in \Il)\;. 
  \end{align}
For all $\psi\in \fock$, by the non--summed first bound from \cref{lem:ccr} and by \cref{lem:bosonic-number}
\begin{align} \label{eq:lin-bos}
  \sum_{\alpha\in \Ik} \norm{\Efrak^{\linb}_\alpha(k)  \psi} &\leq C \sum_{l \in \north} \sum_{\alpha \in \Ik\cap \Il} \norm{ \Ecal^*_\alpha(l,k) c_\alpha(l) \psi} \nonumber\\
  &\leq C \sum_{l \in \north} \sum_{\alpha \in \Ik\cap \Il}  CMN^{-\frac{2}{3}+\delta} \norm{ \Ncal c_\alpha(l) \psi} \nonumber \\
  &= CMN^{-\frac{2}{3}+\delta} \sum_{l \in \north}\sum_{\alpha \in \Ik\cap \Il}  \norm{  c_\alpha(l) (\Ncal-2) \psi} \nonumber\\
  &\leq CMN^{-\frac{2}{3}+\delta} M^{\frac{1}{2}} \norm{  (\Ncal_\delta+1)^{1/2}(\Ncal-2) \psi} \nonumber\\
  &\leq CM^{\frac{3}{2}} N^{-\frac{2}{3}+\delta} \norm{ (\Ncal_\delta+1)^{1/2}(\Ncal+1) \psi}\;.  
  \end{align}
Now using
\begin{align*}
[\hbar^{-1}(\Hbb_0 - \Dbb_\B), c^*_\alpha(k) c^*_\beta(k) ] &= \hbar^{-1}[\Hbb_0 - \Dbb_\B, c^*_\alpha(k)]  c^*_\beta(k)  + c^*_\alpha(k) [\hbar^{-1}(\Hbb_0 - \Dbb_\B),  c^*_\beta(k) ]    \\
&=  \Big(\Efrak^{\textnormal{lin}}_\alpha(k) - \Efrak^{\linb}_\alpha(k) \Big)^* c^*_\beta(k) 
 + c^*_\alpha(k) \Big(\Efrak^{\textnormal{lin}}_\beta(k) - \Efrak^{\linb}_\beta(k) \Big)^* \end{align*}
we can decompose and estimate
\begin{align}
&\hbar^{-1} \Big\lvert \frac{\di}{\di\lambda} \Big\langle T_\lambda \psi, (\Hbb_0 - \Dbb_\B) T_\lambda \psi \Big\rangle \Big\rvert \\
& = \hbar^{-1} \Big\lvert \Re \sum_{k\in \north} \sum_{\alpha,\beta\in \Ik} K(k)_{\alpha,\beta}\Big\langle T_\lambda \psi, [\Hbb_0 - \Dbb_\B, c^*_\alpha(k) c^*_\beta(k) ] T_\lambda \psi \Big\rangle \Big\rvert  \nonumber \\
& \leq  \sum_{k\in \north}  \Big\lvert \sum_{\alpha,\beta\in \Ik} K(k)_{\alpha,\beta} \Big\langle T_\lambda \psi, \Efrak^{\textnormal{lin}}_\alpha(k)^*c^*_\beta(k)  T_\lambda \psi \Big\rangle \Big\rvert  \nonumber \\
&\quad +  \sum_{k\in \north}  \Big\lvert \sum_{\alpha,\beta\in \Ik} K(k)_{\alpha,\beta} \Big\langle T_\lambda \psi, c^*_ \alpha(k) \Efrak^{\textnormal{lin}}_\beta(k)^*  T_\lambda \psi \Big\rangle \Big\rvert   \nonumber\\
&\quad + \sum_{k\in \north}  \Big\lvert \sum_{\alpha,\beta\in \Ik} K(k)_{\alpha,\beta} \Big\langle T_\lambda \psi, \Efrak^{\linb}_\alpha(k)^*c^*_\beta(k)  T_\lambda \psi \Big\rangle \Big\rvert \nonumber  \\
&\quad+  \sum_{k\in \north}  \Big\lvert  \sum_{\alpha,\beta\in \Ik} K(k)_{\alpha,\beta} \Big\langle T_\lambda \psi, c^*_ \alpha(k) \Efrak^{\linb}_\beta(k)^*  T_\lambda \psi \Big\rangle \Big\rvert \nonumber\\
&=: I_1 + I_2 + I_3 + I_4\;. \label{eq:H0-D}
\end{align}

It remains to estimate the right side of \cref{eq:H0-D}, term by term. For $I_1$, since $\lvert K(k)_{\alpha,\beta}\rvert \leq CM^{-1}$, using \cref{lem:bosonic-number},  \cref{lem:lin}  and \cref{lem:stability} we have
\begin{align*}
I_1 & \leq \sum_{k\in \north} \sum_{\alpha\in \Ik}  \Big\lvert \Big\langle \Efrak^{\textnormal{lin}}_\alpha(k) T_\lambda \psi,  \sum_{\beta \in \Ik} K(k)_{\alpha,\beta} c^*_\beta(k)  T_\lambda \psi \Big\rangle \Big\rvert \\
&\leq \sum_{k\in \north}  \sum_{\alpha \in \Ik}  \norm{ \Efrak^{\textnormal{lin}}_\alpha(k)  T_\lambda \psi}  \norm{ \sum_{\beta \in \Ik} K(k)_{\alpha,\beta} c^*_\beta(k)  T_\lambda \psi}\\
&\le\sum_{k\in \north} \sum_{\alpha \in \Ik}  \norm{ \Efrak^{\textnormal{lin}}_\alpha(k)  T_\lambda \psi} C M^{-\frac{1}{2}} \norm{  (\Ncal_\delta+1)^{1/2}  T_\lambda \psi  }\\
&\leq CM^{-\frac{1}{2}} \norm{  (\Ncal_\delta+1)^{1/2}  T_\lambda \psi  }^2 \leq CM^{-\frac{1}{2}} \norm{  (\Ncal_\delta+1)^{1/2}  \psi }^2\;.  
\end{align*}
We can bound $I_2$ similarly to $I_1$, simply exchanging the roles of the $\Efrak^{\textnormal{lin}}$--operator with the $c$--operator. For $I_3$ using \cref{eq:lin-bos} instead of \cref{lem:lin} we have
\begin{align*}
I_3 &\leq \sum_{k\in \north} \sum_{\alpha \in \Ik}  \norm{  \Efrak^{\linb}_\alpha(k)  T_\lambda \psi}  \norm{ \sum_{\beta \in \Ik} K(k)_{\alpha,\beta} c^*_\beta(k)  T_\lambda \psi}\\
&\leq C M^{\frac{3}{2}} N^{-\frac{2}{3}+\delta}   \norm{ (\Ncal_\delta+1)^{1/2} (\Ncal+1) T_\lambda \psi}\, M^{-\frac{1}{2}} \norm{  (\Ncal_\delta+1)^{1/2}  T_\lambda \psi  }\\
&\leq C M N^{-\frac{2}{3}+\delta}  \norm{ (\Ncal_\delta+1)^{1/2} (\Ncal+1) \psi}   \norm{ (\Ncal_\delta+1)^{1/2}   \psi }\; .
\end{align*}
For $I_4$, we split the sum over $\alpha,\beta\in \Ik$ into two parts. If $\alpha\ne \beta$, then 
\[
c^*_ \alpha(k) \Efrak^{\linb}_\beta(k)^*=\Efrak^{\linb}_\beta(k)^* c^*_ \alpha(k)
\]
and this part can be treated similarly to $I_3$. When $\alpha= \beta$, the corresponding contribution is 
\begin{align*}  
I_4'
% &:= \sum_{k\in \north}  \Big\lvert \sum_{\alpha,\beta\in \Ik} \delta_{\alpha,\beta}K(k)_{\alpha,\beta} \Big\langle T_\lambda \psi, c^*_ \alpha(k) \Efrak^{\linb}_\alpha(k)^*  T_\lambda \psi \Big\rangle \Big\rvert \\
&=\sum_{k\in \north}  \Big\lvert \sum_{\alpha\in \Ik}  K(k)_{\alpha,\alpha} \Big\langle T_\lambda \psi, c^*_ \alpha(k) \Efrak^{\linb}_\alpha(k)^*  T_\lambda \psi \Big\rangle \Big\rvert \\
&= \sum_{k\in \north}  \Big\lvert   \sum_{l \in \north} \sum_{\alpha\in \Ik \cap \Il}    2 \kappaf \lvert k \cdot \hat{\omega}_\alpha\rvert  K(k)_{\alpha,\alpha}  \times\\
&\qquad\qquad \times \Big\langle T_\lambda \psi, c^*_ \alpha(k) \Big( \Ecal_\alpha(l,k) c_\alpha^*(l) + c_\alpha^*(k) \Ecal_\alpha (l,l) - \Ecal_\alpha (l,l) c^*_\alpha (k)  \Big)   T_\lambda \psi \Big\rangle \Big\rvert\;.
\end{align*}
Here we have inserted the definition \cref{eq:e-lin-d} and used  \cref{eq:c-eps-com} to obtain
\begin{align*}
c^*_ \alpha(k) \Efrak^{\linb}_\alpha(k)^* &= 2 \kappaf \sum_{l\in \north} \lvert k \cdot \hat{\omega}_\alpha\rvert   c^*_ \alpha(k) c_\alpha^*(l)\Ecal_\alpha(l,k)  \\
&= 2 \kappaf \sum_{l\in \north} \lvert k \cdot \hat{\omega}_\alpha\rvert   c^*_ \alpha(k) \Big( \Ecal_\alpha(l,k) c_\alpha^*(l) + c_\alpha^*(k) \Ecal_\alpha (l,l) - \Ecal_\alpha (l,l) c^*_\alpha (k)  \Big)\;.
\end{align*}
The advantage of the last expression is that $\Ecal_\alpha (l,l)$ commutes with $\Ncal_{\delta}$. Therefore, using the Cauchy--Schwarz inequality together with \cref{lem:ccr}, \cref{lem:bosonic-number} and \cref{lem:stability} we obtain 
\begin{align*}
I_4' &\leq CM^{-1} \sum_{k,l\in \north}  \sum_{\alpha\in \Ik \cap \Il}  \norm{c_ \alpha(k) T_\lambda \psi} \Big[ \norm{ \Ecal_\alpha(l,k) c_\alpha^*(l)  T_\lambda \psi } + \\
&\hspace{16em}+  \norm{ c_\alpha^*(k) \Ecal_\alpha (l,l) T_\lambda \psi} + \norm{\Ecal_\alpha (l,l) c^*_\alpha (k) T_\lambda \psi}   \Big] \\
&\leq CM^{-1} \sum_{k,l\in \north}  \sum_{\alpha\in \Ik \cap \Il} \norm{\Ncal_\delta^{1/2} T_\lambda \psi} \Big[ M N^{-\frac{2}{3}+\delta} \norm{ \Ncal c_\alpha^*(l)  T_\lambda \psi} +\\
&\hspace{15em} + \norm{(\Ncal_{\delta}+1)^{1/2}\Ecal_\alpha (l,l) T_\lambda \psi} + M N^{-\frac{2}{3}+\delta} \norm{ \Ncal c_\alpha^*(k)  T_\lambda \psi}  \Big] \\
&\leq CM^{-1} \sum_{l\in \north}  \sum_{\alpha\in  \Il} \norm{\Ncal_\delta^{1/2} T_\lambda \psi} \Big[ M N^{-\frac{2}{3}+\delta} \norm{ c_\alpha^*(l) (\Ncal +2) T_\lambda \psi} +\norm{\Ecal_\alpha (l,l) (\Ncal_{\delta}+1)^{1/2} T_\lambda \psi}  \Big] \\
&\leq CM^{-1} \sum_{l\in \north} \norm{\Ncal_\delta^{1/2} T_\lambda \psi} \Big[ M N^{-\frac{2}{3}+\delta} M^{\frac{1}{2}}\norm{ (\Ncal_{\delta}+M)^{1/2} (\Ncal +2) T_\lambda \psi} \\
& \hspace{12em} +M^{\frac{3}{2}}N^{-\frac{2}{3}+\delta} \norm{\Ncal  (\Ncal_{\delta}+1)^{1/2} T_\lambda \psi}  \Big] \\
&\leq CM N^{-\frac{2}{3}+\delta} \norm{ (\Ncal_{\delta}+1)^{1/2} (\Ncal +1)  \psi} \norm{ \Ncal_\delta^{1/2}  \psi}\;.
\end{align*}
% This bound is even better than that of $I_3$.
Adding up the contributions to \cref{eq:H0-D} we arrive at the claimed bound \cref{eq:final81}.
\end{proof}

\section{Controlling Non--Bosonizable Terms}
\label{sec:nonboson}

In this section we consider the non--bosonizable terms $\Ecal_1+\Ecal_2$. \Cref{lem:remove-corridors} allows us to replace $\Ecal_2$ by $\Ecal_2^{\Rcal}$, hence we estimate $\mathcal{E}_1 + \mathcal{E}_2^{\Rcal}$. Recall that after the expansion into patches \cref{eq:bR-c} we have 
\begin{equation}
\label{eq:E1pE2R}
\begin{split}
 \mathcal{E}_1 + \mathcal{E}_2^{\Rcal} &= \frac{1}{2N} \sum_{k \in \north } \hat{V}(k) \Big[  \Dfrak(k)^*\Dfrak (k) + \Dfrak(-k)^*\Dfrak (-k) \Big]\\
 & \quad + \frac{1}{N} \sum_{k \in \north } \hat{V}(k) \Big[  \Dfrak(k)^* \sum_{\alpha\in \Ikp} n_\alpha(k) c_\alpha(k) + \Dfrak(-k)^* \sum_{\alpha\in \Ikm} n_\alpha(k) c_\alpha(k) + \hc \Big]\;. 
\end{split}\end{equation}
Now we prove that after the Bogoliubov transformation, the non--bosonizable terms can be bounded from below using the fermionic kinetic energy. This step relies on the assumption $\hat{V}(k) \geq 0$.

\begin{lem}[Non--Bosonizable Terms]\label{lem:non--bosonizable}
Let $T_\lambda$ be the approximate Bogoliubov transformation defined in \cref{eq:Bog-T}. Then for all $\lambda\in [-1,1]$ and for all $\psi\in \fock$ we have
\begin{equation} \label{eq:nonbos} \begin{split}
\langle \psi, T_\lambda^* (\Ecal_1+ \Ecal_2^{\Rcal}) T_\lambda  \psi\rangle &\geq - C \norm{\hat V}_{\ell^1} \langle \psi, \Hbb_0 \psi\rangle - CN^{-\frac{1}{2}} \norm{\psi}\norm{ \Hbb_0^{1/2} T_\lambda \psi} \\
&\quad - C N^{-\frac{5}{3}+2\delta} M \norm{(\Ncal_\delta+M)^{1/2} (\Ncal +1)\psi}^2\;.
\end{split} \end{equation}
\end{lem}
Note that in the lemma $\Hbb_0$ acts once on $\psi$ and once on $T_\lambda \psi$. The sharp bound on the matrix elements of the Bogoliubov kernel from \cref{lem:K} is crucial to the proof of this lemma.

The smallness assumption on $\hat{V}$ is important to control the term $-C\norm{\hat{V}}_{\ell^1} \Hbb_0$ appearing on the right hand side in this lemma.

\begin{proof} Take $k,l\in \north$. We denote $\tD(l):=T_\lambda^* \Dfrak(l) T_\lambda$. By \cref{lem:Bog-T}
% \begin{align}\label{eq:c-T-lambda}
% T_\lambda^* c_\alpha(k) T_\lambda =\mathfrak{E}_\alpha(\lambda,k) + \sum_{\beta\in\Ik} \cosh( \lambda K(k))_{\alpha,\beta} c_\beta(k) +  \sum_{\beta\in\Ik} \sinh(\lambda K(k))_{\alpha,\beta} c^*_\beta(k)
% \end{align}
we can write
\begin{align*}
T_\lambda^*  \Dfrak(l)^* c_\alpha(k) T_\lambda = \tD^*(l) \Big( \mathfrak{E}_\alpha(\lambda,k) + \sum_{\beta\in\Ik} \cosh(\lambda K(k))_{\alpha,\beta} c_\beta(k) +  \sum_{\beta\in\Ik} \sinh(\lambda K(k))_{\alpha,\beta} c^*_\beta(k) \Big). 
\end{align*}
Therefore, by the triangle inequality, the contribution of $\mathcal{E}_2^{\Rcal}$ to \cref{eq:nonbos} can be bounded by 
\begin{align}
&N^{-1} \sum_{\alpha \in \Ik}  n_\alpha(k) \big\lvert \big\langle  \psi, T_\lambda^* \Dfrak(l)^* c_\alpha(k) T_\lambda \psi \big\rangle\big\rvert \\
& \leq N^{-1}\sum_{\alpha \in \Ik} n_\alpha(k) \big\lvert \big\langle  \psi, \tD^*(l)  \mathfrak{E}_\alpha(\lambda,k) \psi \big\rangle\big\rvert \nonumber\\
&\quad + N^{-1}\sum_{\alpha,\beta\in\Ik} n_\alpha(k) \lvert\cosh( \lambda  K(k))_{\alpha,\beta}\rvert \big\lvert \big\langle  \psi, \tD^*(l) c_\beta(k) \psi \big\rangle\big\rvert \nonumber\\
&\quad + N^{-1} \sum_{\alpha,\beta\in\Ik} n_\alpha(k) \lvert\sinh(\lambda K(k))_{\alpha,\beta}\rvert \big\lvert \big\langle  \psi, \tD^*(l) c^*_\beta(k) \psi \big\rangle\big\rvert =: I_1+I_2+I_3\;. \label{eq:D-b-T-01} 
\end{align}
We proceed to bound the right side of \cref{eq:D-b-T-01} term by term. The first term can be bounded using $n_\alpha(k)\leq CN^{\frac{1}{3}}M^{-\frac{1}{2}}$ and \cref{lem:Bog-T}:\begin{align*}
I_1 &\leq CN^{-1} \sum_{\alpha\in \Ik} N^{\frac{1}{3}}M^{- 1/ 2} \norm{ \tD(l) \psi} \norm{\mathfrak{E}_\alpha(\lambda,k) \psi}\\
% &\leq CN^{-1} N^{\frac{1}{3}}M^{-\frac{1}{2}} \norm{ \tD(l) \psi}   M N^{-\frac{2}{3} +\delta} \norm{(\Ncal_\delta+M)^{1/2} (\Ncal +1)\psi} \\
&\leq CN^{-\frac{4}{3}+\delta} M^{\frac{1}{2}} \ \norm{ \tD(l) \psi} \norm{(\Ncal_\delta+M)^{1/2} (\Ncal +1)\psi}\;. 
\end{align*}
Using \cref{eq:ch-sh} we have 
\begin{align}  \label{eq:ch-sh-sum}
\sum_{\alpha \in \Ik} n_\alpha(k) \Big( \lvert\cosh(\lambda K(k))_{\alpha,\beta}\rvert +  \lvert\sinh(\lambda K(k))_{\alpha,\beta}\rvert \Big) \leq C n_\beta(k)\;.
\end{align}
Consequently, using \cref{lem:kinetic} we get
\begin{align*}
I_2 &\leq CN^{-1} \sum_{\beta \in \Ik}  n_\beta(k) \big\lvert \big\langle  \psi, \tD^*(l) c_\beta(k) \psi \big\rangle\big\rvert \\
&\leq    CN^{-1} \norm{\tD(l)\psi} \sum_{\beta  \in \Ik}  n_\beta(k) \norm{c_\beta(k) \psi} \leq CN^{-\frac{1}{2}} \norm{\tD(l)\psi} \norm{\Hbb_0^{1/2} \psi}\;. 
\end{align*}

The third term is more difficult. We have
\begin{align*}
I_3 &\leq CN^{-1} \sum_{\alpha \in \Ik}  n_\alpha(k) \big\lvert \big\langle  \psi, \tD^*(l) c_\alpha^*(k) \psi \big\rangle\big\rvert \\
&\leq CN^{-1} \sum_{\alpha \in \Ik}  n_\alpha(k) \big\lvert \big\langle  \psi, c_\alpha^*(k) \tD^*(l)  \psi \big\rangle\big\rvert + CN^{-1} \sum_{\alpha \in \Ik}  n_\alpha(k) \big\lvert \big\langle  \psi, [\tD^*(l), c_\alpha^*(k)]  \psi \big\rangle\big\rvert \\
&=: I_4+I_5\;. 
\end{align*}
The term $I_4$ can be bounded again by \cref{lem:kinetic} as
\begin{align*}
I_4 \leq CN^{-1}  \sum_{\alpha \Ik}  n_\alpha(k) \norm{c_\alpha(k) \psi}   \norm{\tD^*(l)\psi} \leq CN^{-\frac{1}{2}} \norm{\Hbb_0^{1/2} \psi} \norm{\tD(-l)\psi}\;. 
\end{align*} 
The commutator in $I_5$ can be computed by undoing the approximate Bogoliubov transformation, 
\[
[\tD^*(l), c_\alpha^*(k)] = [T_\lambda^* \Dfrak(l)^* T_\lambda, c_\alpha^*(k)] = T_\lambda^* [\Dfrak(l)^*, T_\lambda c^*_\alpha(k) T_\lambda^*] T_\lambda\;,
\]
and using \cref{eq:Bog_Tlambda} with $T_\lambda=T^*_{-\lambda}$, i.\,e.,
\[
T_\lambda c^*_\alpha(k) T_\lambda^*= \mathfrak{E}^*_\alpha(-\lambda,k) + \sum_{\beta\in\Ik} \cosh(\lambda K(k))_{\alpha,\beta} c^*_\beta(k) -   \sum_{\beta\in\Ik} \sinh(\lambda K(k))_{\alpha,\beta} c_\beta(k)\;.
\]
Therefore, by the triangle inequality and \cref{eq:ch-sh-sum} we can decompose 
\begin{align*}
I_5 &\leq CN^{-1} \sum_{\alpha \in  \Ik}  n_\alpha(k) \big\lvert \big\langle T_\lambda \psi, \big[\Dfrak(l)^*, \mathfrak{E}^*_\alpha(-\lambda,k)\big] T_\lambda \psi \big\rangle\big\rvert \\
&\quad +  CN^{-1} \sum_{\alpha \in  \Ik} n_\alpha(k) \big\lvert \big\langle T_\lambda \psi, \big[\Dfrak(l)^*, c_\alpha^*(k) \big] T_\lambda \psi \big\rangle\big\rvert \\
&\quad + CN^{-1} \sum_{\alpha \in  \Ik} n_\alpha(k) \big\lvert \big\langle T_\lambda \psi, \big[\Dfrak(l)^*, c_\alpha(k) \big] T_\lambda \psi \big\rangle\big\rvert =: I_6 + I_7+ I_8\;. 
\end{align*}
For $I_6$, we simply expand the commutator and use  \cref{lem:Bog-T} similarly as done for $I_1$ to get
\begin{align*}
I_6 &\leq CN^{-1} \sum_{\alpha\in \Ik} N^{\frac{1}{3}}M^{- 1/ 2} \Big( \norm{  \Dfrak(l) T_\lambda \psi} \norm{\mathfrak{E}^*_\alpha(-\lambda,k) T_\lambda \psi} + \norm{\mathfrak{E}_\alpha(-\lambda,k) T_\lambda \psi} \norm{  \Dfrak(-l) T_\lambda \psi}\Big)\\
&\leq CN^{-\frac{4}{3}+\delta} M^{\frac{1}{2}}  \Big( \norm{ \tD(l) \psi} + \norm{ \tD(-l) \psi} \Big) \norm{(\Ncal_\delta+M)^{1/2} (\Ncal +1)\psi}\;. 
\end{align*}
Here we used $\norm{  \Dfrak(l) T_\lambda \psi}= \norm{ \tD(l) \psi}$ as $T_\lambda$ is unitary.

For $I_7$ we compute the commutator explicitly. We decompose the operator $\Dfrak(l)^*$ as
\[
\Dfrak(l)^* = \sum_{p  \in \BFc \cap (\BFc+l)} a^*_{p} a_{p-l} - \sum_{h  \in \BF \cap (\BF+l)} a^*_{h} a_{h-l} =: \Dfrak^\textnormal{p}(l)^* - \Dfrak^\textnormal{h}(l)^*\;.
\]
In the case $\alpha \in \Ikp$ we can then compute 
\[
 n_\alpha(k) [\Dfrak^\textnormal{p}(l)^*,c_{\alpha}^*(k)]  = \big[ \sum_{q \in \BFc \cap (\BFc+l)} a^*_{q} a_{q-l}, 
   \sum_{\substack{p\colon p \in \BFc \cap B_\alpha\\p-k\in \BF \cap B_\alpha}}  a^*_p a^*_{p-k} \big]  =  \sum_{\substack{q\colon q-l \in \BFc \cap B_\alpha \\q-l-k \in \BF \cap B_\alpha \\ q \in \BFc }} a^*_{q} a^*_{q-l-k}\; . 
\]
By the Cauchy--Schwarz inequality and the kinetic energy estimate in \cref{lem:kinetic} we obtain
\begin{align*}
N^{-1} \sum_{\alpha\in \Ikp} n_\alpha(k) \left\lvert \left\langle T_\lambda \psi, [\Dfrak^\textnormal{p}(l)^*,  c^*_\alpha(k)] T_\lambda \psi \right\rangle \right\rvert & \leq N^{-1}  \sum_{q \in \BFc \cap (\BF + k+l) } \norm{T_\lambda \psi}  \norm{ a_{q} a_{q-k-l} T_\lambda \psi} \\
&\leq CN^{-\frac{1}{2}} \norm{\psi} \norm{ \Hbb_0^{1/2} T_\lambda \psi}\;. 
\end{align*} 
For $\alpha\in \Ikm$ and $\Dfrak^\textnormal{h}(l)$ we get similar estimates. The commutator in $I_8$ can be bounded exactly the same way, using $\Dfrak(l)^*=\Dfrak(-l)$. Thus  
\[
I_7 + I_8 \leq  CN^{-\frac{1}{2}} \norm{\psi}\norm{ \Hbb_0^{1/2} T_\lambda \psi}\;. 
\]
Collecting all estimates for $I_1,\ldots, I_8$ we conclude from \cref{eq:D-b-T-01} that
\begin{align} \label{eq:key-D-b}
&N^{-1} \sum_{\alpha \in \Ik}  n_\alpha(k) \big\lvert \big\langle  \psi, T_\lambda^* \Dfrak(l)^* c_\alpha(k) T_\lambda \psi \big\rangle\big\rvert \nonumber\\
&\leq CN^{-\frac{4}{3}+\delta} M^{\frac{1}{2}} \Big( \norm{ \tD(l) \psi} + \norm{ \tD(-l) \psi} \Big)\norm{(\Ncal_\delta+M)^{1/2} (\Ncal +1)\psi} \nonumber\\
&\quad + CN^{-\frac{1}{2}} \Big( \norm{ \tD(l) \psi} + \norm{ \tD(-l) \psi} \Big) \norm{\Hbb_0^{1/2} \psi} + CN^{-\frac{1}{2}} \norm{\psi}\norm{ \Hbb_0^{1/2} T_\lambda \psi}\;. 
\end{align}

The bound \cref{eq:key-D-b} holds true for all $k,l\in \north$. In particular, by the Cauchy--Schwarz inequality we deduce that 
\begin{align*}
& |\langle \psi, T_\lambda^* \Ecal_2^{\Rcal} T_\lambda  \psi\rangle| \leq  2 \sum_{k\in \north} \sum_{l= \pm k}  \frac{\hat V(k)}{N} \sum_{\alpha \in \Ik}  n_\alpha(k) \big\lvert \big\langle  \psi, T_\lambda^* \Dfrak(l)^* c_\alpha(k) T_\lambda \psi \big\rangle\big\rvert \\
&\leq   \sum_{k\in \north}\sum_{l= \pm k}  C\hat V(k) \Big[ N^{-\frac{4}{3}+\delta} M^{\frac{1}{2}} \ \norm{ \tD(l) \psi} \norm{(\Ncal_\delta+M)^{1/2} (\Ncal +1)\psi} + \\
&\hspace{9em} + CN^{-\frac{1}{2}} \norm{\tD(l)\psi} \norm{\Hbb_0^{1/2} \psi} + CN^{-\frac{1}{2}} \norm{\psi}\norm{ \Hbb_0^{1/2} T_\lambda \psi} \Big]\\
&\leq  \sum_{k\in \north}  \frac{\hat V(k)}{4N}  \Big[ \norm{ \tD(k) \psi}^2+\norm{ \tD(-k) \psi}^2\Big] \\
& \qquad +   \sum_{k\in \north} C \hat V(k) N \Big[ N^{-\frac{4}{3}+\delta} M^{\frac{1}{2}}  \norm{(\Ncal_\delta+M)^{1/2} (\Ncal +1)\psi} + N^{-\frac{1}{2}}\norm{\Hbb_0^{1/2} \psi} \Big]^2  \\
&\qquad + CN^{-\frac{1}{2}} \norm{\psi}\norm{ \Hbb_0^{1/2} T_\lambda \psi} \\
& \leq  \langle \psi, T_\lambda^* \Ecal_1 T_\lambda  \psi\rangle  + C \norm{\hat V}_{\ell^1} \langle \psi, \Hbb_0 \psi\rangle  \\
&\qquad + C N^{-\frac{5}{3}+2\delta} M \norm{(\Ncal_\delta+M)^{1/2} (\Ncal +1)\psi}^2 +CN^{-\frac{1}{2}} \norm{\psi}\norm{ \Hbb_0^{1/2} T_\lambda \psi}\;. 
\end{align*} 
This concludes the proof of \cref{eq:nonbos}.
\end{proof}    
  
\section{Diagonalization of Approximately Bosonic Hamiltonian} \label{sec:bosonic-Hamiltonian} 

We now focus on the approximately bosonic Hamiltonian $\Dbb_\B + Q_\B^{\Rcal}$, with $\Dbb_\B$ and $Q_\B^{\Rcal}$ as defined in \cref{eq:DBdef} and \cref{eq:QBNr3}. With $h_\textnormal{eff}(k)$ being the effective Hamiltonian introduced in \cref{eq:heff}, we can write
\begin{equation} \label{eq:DplusQ}
\Dbb_\B + Q_\B^{\Rcal} = \sum_{k\in \north} 2\hbar \kappaf \lvert k\rvert h_\textnormal{eff}(k)  \;.
\end{equation}
% \[
% h_\textnormal{eff}(k)  := \sum_{\alpha,\beta \in \Ik} \Big[ \big( \D(k) + \W(k) \big)_{\alpha,\beta} c_\alpha^*(k) c_\beta(k) +  \frac{1}{2} \Wt(k)_{\alpha,\beta} \big( c^*_\alpha(k) c^*_\beta(k) + c_\beta(k)  c_\alpha(k) \big)  \Big]\;.
% \]

The main result of this section is the following lemma in which we approximately diagonalize the effective Hamiltonian, extract its ground state energy $E_N^\textnormal{RPA}$, and then bound the excitation spectrum below by $\Dbb_\B - C \norm{\hat{V}}_{\ell^1} \Hbb_0$.
\begin{lem}[Diagonalization of Bosonized Hamiltonian] \label{lem:diag-final} Let $T_1$ be the approximate Bogoliubov transformation defined in \cref{eq:Bog-T}. For all normalized $\psi\in \fock$ we have 
 \begin{align*}
\Big\langle \psi, T_1^* \Big(\Dbb_\B + Q_\B^{\Rcal} \Big) T_1 \psi \Big\rangle & \geq E_N^\textnormal{RPA}  +  \left\langle \psi,  \Dbb_\B \psi \right\rangle - C \norm{\hat V}_{\ell^1} \langle \psi, \Hbb_0 \psi\rangle \\
&\quad -C \hbar \Big[ N^{-\frac{2}{3}+\delta} \norm{\Ncal^{1/2}\psi}^2 + M^{\frac{1}{4}}N^{-\frac{1}{6}+\frac{\delta}{2}} + N^{-\frac{\delta}{2}} + M^{-\frac{1}{4}}N^{\frac{\delta}{2}} \nonumber\\
& \qquad \qquad +  \Big( M N^{-\frac{2}{3}+\delta}
% \norm{(\Ncal_\delta+1)^{1/2} \psi} 
\norm{(\Ncal_\delta+M)^{1/2} (\Ncal +1) \psi} \Big)^2\\
& \qquad  \qquad   + MN^{-\frac{2}{3}+\delta}  \norm{(\Ncal_\delta+1)^{1/2} \psi} \norm{(\Ncal_\delta+M)^{1/2} (\Ncal +1) \psi}\Big]\; ,
\end{align*} 
where $E_N^\textnormal{RPA}$ is the RPA correlation energy defined in \cref{eq:rpa_energy}.
\end{lem}

The smallness assumption on $\hat{V}$ is important to control the term $-C\norm{\hat{V}}_{\ell^1} \Hbb_0$ appearing on the right hand side in this lemma.

\begin{proof} \textbf{Error Terms.} Recall that $K(k)$ is a real symmetric matrix, and hence also $\sinh(K(k))$ and $\cosh(K(k)$ are real symmetric matrices.
Defining
\begin{align} \label{eq:tc-def}
\tc_\alpha(k):= \sum_{\beta\in\Ik}\cosh(K(k))_{\alpha,\beta} c_\beta(k) + \sum_{\beta\in\Ik} \sinh( K(k))_{\alpha,\beta} c^*_\beta(k)
\end{align}
we have, according to \cref{lem:Bog-T}, 
\begin{align} \label{eq:Bog-error-f}
T_1^* c_\alpha(k)  T_1 = \tc_\alpha(k)  + \mathfrak{E}_\alpha(1,k) \;.
\end{align} 
In the first step we are going to control the contribution of $\mathfrak{E}_\alpha(1,k)$. By \cref{eq:ch-sh} we have  
\[
\lvert\cosh( K(k))_{\alpha,\beta}-  \delta_{\alpha,\beta}\rvert + \lvert\sinh( K(k))_{\alpha,\beta}\rvert \leq \frac{C}{M}
\]
and thus using \cref{lem:bosonic-number} (to treat the contribution of $\delta_{\alpha,\beta}$, recall that $\norm{c_\alpha(k) \psi} \leq \norm{\Ncal_\delta^{1/2}\psi}$ for all $\psi \in \fock$ follows from the first bound in the lemma)  we get 
\begin{align} \label{eq:Bog-error-f-2}
\norm{\tc_\alpha(k) \psi} +  \norm{\tc_\alpha^*(k) \psi} \leq C \norm{ (\Ncal_\delta+1)^{\frac{1}{2}} \psi}\;.
\end{align}
Now we expand
\begin{align} \label{eq:THbT-exp}
T_1^* h_\textnormal{eff}(k) T_1&= \sum_{\alpha,\beta \in \Ik}  \Big( \D(k) + \W(k)\Big)_{\alpha,\beta}  \Big(\tc_\alpha^*(k)+ \mathfrak{E}^*_\alpha(1,k) \Big) \Big(\tc_\beta(k) + \mathfrak{E}_\beta(1,k) \Big) \nonumber \\
&\quad\ +  \frac{1}{2}\sum_{\alpha,\beta \in \Ik}  \Big[ \Wt(k)_{\alpha,\beta} \Big( \tc^*_\alpha(k) + \mathfrak{E}^*_\alpha(1,k) \Big) \Big( \tc^*_\beta(k) + \mathfrak{E}^*_\beta(1,k)) \Big)+\hc \Big]\; .
\end{align}
The main contribution is
\begin{align*} 
h_\textnormal{eff}^\textnormal{diag}(k) &:= \sum_{\alpha,\beta \in \Ik} \Big[ \big( \D(k) + \W(k)\big)_{\alpha,\beta}  \tc_\alpha^*(k)\tc_\beta(k)  + \frac{1}{2} \Wt(k)_{\alpha,\beta} ( \tc^*_\alpha(k) \tc^*_\beta(k) + \tc_\beta(k)\tc_\alpha(k) \Big]\;.
\end{align*}
To bound the error terms in \cref{eq:THbT-exp}, first observe that  
\[
\lvert \D(k)_{\alpha,\beta} \rvert \leq \delta_{\alpha,\beta}\;, \quad \lvert \W(k)_{\alpha,\beta}\rvert + \lvert \Wt(k)_{\alpha,\beta}\rvert\leq CM^{-1}\;.
\]
Now using  the Cauchy--Schwarz inequality together with \cref{eq:Bog-error-f-2} and \cref{lem:Bog-T} we get 
\begin{align}
& \Big\lvert\langle \psi, (T_1^* h_\textnormal{eff}(k) T_1 - h_\textnormal{eff}^\textnormal{diag}(k) ) \psi  \rangle\Big\rvert \nonumber\\
&\leq \sum_{\alpha,\beta\in \Ik} \Big(\D(k) + \W(k)\Big)_{\alpha,\beta} \Big[  2\norm{\tc_\alpha(k) \psi} \norm{\mathfrak{E}_\beta(1,k) \psi} + \norm{ \mathfrak{E}_\alpha(1,k) \psi} \norm{ \mathfrak{E}_\beta(1,k) \psi} \Big] \nonumber\\
& \quad + 2\sum_{\alpha,\beta\in \Ik} \Wt(k)_{\alpha,\beta} \Big[ \norm{\tc_\alpha(k) \psi} \norm{ \mathfrak{E}^*_\beta(1,k) \psi} + \norm{ \mathfrak{E}_\alpha(1,k) \psi} \norm{ \mathfrak{E}^*_\beta(1,k) \psi}\Big]\nonumber\\
& \quad +  2\sum_{\alpha\in \Ik} \norm{ \mathfrak{E}_\alpha(1,k) \psi} \norm{ \sum_{\beta\in \Ik} \Wt(k)_{\alpha,\beta} \tc^*_\beta(k) \psi}\nonumber\\
&\leq CMN^{-\frac{2}{3}+\delta}  \norm{(\Ncal_\delta+1)^{1/2} \psi} \norm{(\Ncal_\delta+M)^{1/2} (\Ncal +1) \psi} \nonumber\\
&\quad + C\Big( MN^{-\frac{2}{3}+\delta}
% \norm{(\Ncal_\delta+1)^{1/2} \psi}
\norm{(\Ncal_\delta+M)^{1/2} (\Ncal +1) \psi} \Big)^2\;. \label{labeq:THbT-exp-err} 
\end{align}
In this calculation we used the crude bound $\sum_{\alpha \in \Ik} \norm{ \mathfrak{E}_\alpha(1,k) \psi}^2 \leq \big(\sum_{\alpha \in \Ik} \norm{ \mathfrak{E}_\alpha(1,k) \psi}\big)^2$.

\paragraph{Bosonic Terms.} Now we compute $h_\textnormal{eff}^\textnormal{diag}(k)$ by inserting the transformation \cref{eq:tc-def}.   In this step let us suppress the $k$--dependence in the notation.  We have
\begin{align*}
\tc_\alpha^* \tc_\beta 
% &= \Big[ \sum_{\alpha'\in\Ik} \cosh(K)_{\alpha,\alpha'} c^*_{\alpha'} + \sum_{\alpha' \in\Ik} \sinh( K)_{\alpha,\alpha'} c_{\alpha'} \Big] \times \\
% &\qquad \times \Big[ \sum_{\beta'\in\Ik}\cosh(K)_{\beta,\beta'} c_{\beta'} + \sum_{\beta'\in\Ik} \sinh( K)_{\beta,\beta'} c^*_{\beta'} \Big] \\
&= \sum_{\alpha',{\beta'} \in\Ik} \Big[ \cosh(K)_{\alpha,\alpha'} \cosh(K)_{\beta,{\beta'}} c^*_{\alpha'}c_{\beta'}  + \cosh(K)_{\alpha,\alpha'}  \sinh( K)_{\beta,{\beta'}} c^*_{\alpha'} c^*_{\beta'}   \\
&\hspace{5em}  + \sinh( K)_{\alpha,\alpha'}  \cosh(K)_{\beta,{\beta'}} c_{\beta'}c_{\alpha'}   \\
&\hspace{5em}  + \sinh( K)_{\alpha,\alpha'} \sinh( K)_{\beta,{\beta'}}  \Big( c^*_{\beta'}c_{\alpha'}   +  \delta_{\alpha',\beta'} +  \delta_{\alpha',\beta'} \Ecal_{\alpha'} (k,k)  \Big)  \Big]\;,
\end{align*}
where we used the approximate CCR \cref{eq:approximateCCR} to achieve bosonic Wick--normal order.
% Note that 
% \begin{align*}
% \sum_{\alpha,\beta \in \Ik} \Big(D+\W\Big)_{\alpha,\beta}  \cosh(K)_{\alpha,\alpha'} \cosh(K)_{\beta,{\beta'}}  =  \Big( \cosh(K) \Big(D+\W\Big)  \cosh(K) \Big)_{\alpha',\beta'}\;.
% \end{align*}
Moreover, using $[c_\alpha^*,c_\beta^*]=[c_\alpha,c_\beta]=0$ we symmetrize the coefficients of $c_\alpha^* c_\beta^*$ and $c_\alpha c_\beta$; thus 
\begin{align*}
&\sum_{\alpha,\beta \in \Ik} (D+\W)_{\alpha,\beta} \tc_\alpha^* \tc_\beta\\
&= \sum_{\alpha,\beta \in\Ik} \Big( \cosh(K) (D+W)  \cosh(K) + \sinh(K) (D+W)  \sinh(K)  \Big)_{\alpha,\beta} c_\alpha^* c_\beta  \\
&\quad + \frac{1}{2}\sum_{\alpha,\beta \in\Ik} \Big( \cosh(K) (D+W) \sinh( K) + \sinh(K) (D+W) \cosh( K) \Big)_{\alpha,\beta} (c^*_\alpha c^*_\beta +   c_\beta c_\alpha )\\
&\quad + \sum_{\alpha\in \Ik} \Big( \sinh(K) (D+W)  \sinh(K) \Big)_{\alpha,\alpha} \Big(1+ \Ecal_\alpha (k,k)\Big)\;. 
\end{align*}
Moreover 
\begin{align*}
& \sum_{\alpha,\beta \in \Ik} \Wt_{\alpha,\beta} \tc_\alpha^* \tc^*_\beta \\ & = \sum_{\alpha,\beta \in\Ik} \Big( \cosh(K) \Wt  \sinh(K) + \sinh(K) \Wt  \cosh(K)  \Big)_{\alpha,\beta} c_\alpha^* c_\beta  \\
&\quad + \sum_{\alpha,\beta \in\Ik} \Big( \cosh(K) \Wt \cosh( K) \Big)_{\alpha,\beta} c^*_\alpha c^*_\beta  + \sum_{\alpha,\beta \in\Ik} \Big( \sinh(K) \Wt \sinh(K) \Big)_{\alpha,\beta} c_\alpha c_\beta  \\
&\quad + \sum_{\alpha\in \Ik} \Big( \sinh(K) \Wt  \cosh(K) \Big)_{\alpha,\alpha} \Big(1+  \Ecal_\alpha (k,k)\Big)\;. 
\end{align*}
Adding both terms, we thus have
\begin{align} \label{eq:H-diag-final}
h_\textnormal{eff}^\textnormal{diag}(k) &= \sum_{\alpha,\beta \in\Ik}  \Big[ \cosh(K) \Big(D +W \Big)  \cosh(K)  + \sinh(K) \Big(D+W\Big)  \sinh(K)  \nonumber \\
&\qquad \qquad + \cosh(K) \Wt  \sinh(K) + \sinh(K) \Wt  \cosh(K) \Big]_{\alpha,\beta} c_\alpha^* c_\beta\nonumber  \\
&\quad + \frac{1}{2}\sum_{\alpha,\beta \in\Ik}    \Big[   \cosh(K) \Big(D+W\Big) \sinh( K) + \sinh(K) \Big(D+W\Big) \cosh( K) \nonumber  \\
&\hspace{6em} + \cosh(K) \Wt \cosh( K)  + \sinh(K) \Wt \sinh( K) \Big]_{\alpha,\beta} (c_\alpha^* c^*_\beta + c_\beta c_\alpha)  \nonumber  \\
&\quad + \frac{1}{2}\sum_{\alpha\in \Ik}   \Big[ 2\sinh(K) \Big(D+W\Big)  \sinh(K) \nonumber \\
&\hspace{5em} + \cosh(K) \Wt \sinh(K)    + \sinh(K) \Wt  \cosh(K) \Big]_{\alpha,\alpha} \Big(1+ \Ecal_\alpha (k,k)\Big)\;.
\end{align}
To simplify \cref{eq:H-diag-final} further,
recall from \cref{eq:defS1} that
\[
S_1 = (D+W-\tilde{W})^{1/2} E^{-1/2}
\]
and set
\[
S_2 := (S_1^\intercal)^{-1} \;.
\]
We use the polar decomposition to write $S_1^\intercal = O \lvert S_1^\intercal \rvert$ with an orthogonal matrix $O$. Inserting our choice $K = \log \lvert S_1^\intercal\rvert$ into the exponentials defining the $\cosh$ and the $\sinh$, and noting that each exponential by itself is a symmetric matrix, we obtain
\begin{equation}
\cosh(K) = \frac{1}{2}(S_1 + S_2)O\;, \qquad \sinh(K) = \frac{1}{2}(S_1-S_2)O\;. \label{eq:bogbog1}
\end{equation}
% \footnote{The appearance of the matrix $O$ means that with $\cosh(K)$ and $\sinh(K)$ we have not implemented the Bogoliubov transformation defined by $S_1$ and $S_2$ but rather one that differs by a one--particle unitary $\Gamma(O)$. In the bosonic approximation, a one--particle unitary does not change the energy of the state.
% % This is inevitable since in the explicit expression \cref{eq:Bog-T} only the symmetric part of $K(k)$ is relevant.
% }
% Recall that $K$ is a symmetric matrix, so that also $\cosh(K)$ and $\sinh(K)$ are symmetric.

%(The reader may wonder if we cannot find $K$ such that $\cosh(K) = \frac{1}{2}(S_1+S_2)$, avoiding the matrix $O$. However, since $c^*_\alpha(k)$ and $c^*_\beta(k)$ commute exactly, only the symmetric part $K^\textnormal{symm}_{\alpha,\beta} = \frac{1}{2}(K_{\alpha,\beta} + K_{\beta,\alpha})$ of any matrix $K$ contributes in \cref{eq:bogformula}. In general, $\cosh(K)$ in \cref{eq:tc-def} would become $\cosh(K^\textnormal{symm})$. As a symmetric matrix, this cannot equal the non--symmetric $\frac{1}{2}(S_1+S_2)$. With $\log\lvert S_1^\intercal \rvert$ we have an explicitly symmetric $K$, and $\lvert S_1^\intercal \rvert$ being positive also avoids any discussion of the branch of the logarithm. The additional matrix $O$ corresponds to a change of basis in the one--boson Hilbert space, which, at least in the bosonic approximation, does not change the energy of the many--body state.)  

Note that since $K$ is symmetric, it cannot equal the non--symmetric matrix $\frac{1}{2}(S_1+S_2)$. Thus the inclusion of the unitary operator $O$ in \eqref{eq:bogbog1} is inevitable. However, $O$ corresponds to a change of basis in the one--boson Hilbert space, which, at least in the bosonic approximation, does not change the energy of the many--body state. 

%The reader may wonder if we cannot find $K$ such that $\cosh(K) = \frac{1}{2}(S_1+S_2)$, avoiding the matrix $O$. However, since $c^*_\alpha(k)$ and $c^*_\beta(k)$ commute exactly, only the symmetric part $K^\textnormal{symm}_{\alpha,\beta} = \frac{1}{2}(K_{\alpha,\beta} + K_{\beta,\alpha})$ of any matrix $K$ contributes in \cref{eq:bogformula}. In general, $\cosh(K)$ in \cref{eq:tc-def} would become $\cosh(K^\textnormal{symm})$. As a symmetric matrix, this cannot equal the non--symmetric $\frac{1}{2}(S_1+S_2)$. With $\log\lvert S_1^\intercal \rvert$ we have an explicitly symmetric $K$, and $\lvert S_1^\intercal \rvert$ being positive also avoids any discussion of the branch of the logarithm. The additional matrix $O$ corresponds to a change of basis in the one--boson Hilbert space, which, at least in the bosonic approximation, does not change the energy of the many--body state.)  

\paragraph{Ground state energy.}
As in \cite[Proof of Theorem 2.1 and Appendix A.2]{BNP+20}, with the matrix $E$ defined in \cref{eq:defE}, we have
\[ 
S_1^\intercal (D+W+\tilde{W}) S_1 = S_2 ^\intercal (D+W-\tilde{W}) S_2 =E\;.
\]
Using the fact that $\cosh(K)$ and $\sinh(K)$ are symmetric matrices, the constant term of \cref{eq:H-diag-final} thus simplifies to
\begin{align*}  
&\frac{1}{2}\sum_{\alpha\in \Ik}  \Big[ 2\sinh(K) \Big(D+W\Big)  \sinh(K)+ \sinh(K) \Wt  \cosh(K) +\cosh(K)\Wt \sinh(K)     \Big]_{\alpha,\alpha}\\
% &=\frac{1}{2}\tr \Big[ 2\sinh(K) \Big(D+W\Big)  \sinh(K)+ \sinh(K) \Wt  \cosh(K) +\cosh(K)\Wt \sinh(K)     \Big] \\
& = \frac{1}{8}\tr \Big[ 2 (S_1^\intercal-S_2^\intercal)\Big(D+W\Big) (S_1-S_2)+ (S_1^\intercal-S_2^\intercal)  \Wt  (S_1+S_2) + (S_1^\intercal+S_2^\intercal)\Wt (S_1-S_2) \Big]\\
% &= \frac{1}{4} \tr \Big[ S_1^\intercal \Big(D+W +\Wt \Big) S_1+ S_2^\intercal (D+W-  \Wt) S_2 - S_1^\intercal (D+W) S_2 - S_2^\intercal (D+W) S_1   \Big] \\
& =  \frac{1}{4} \tr \Big[ 2 E - (S_2 S_1^\intercal + S_1  S_2^\intercal) (D+W)   \Big] = \frac{1}{2} \tr(E-D-W)\;. 
\end{align*}
This is the term giving rise to the ground state energy we aim to derive.

To estimate the error term proportional to $\Ecal_\alpha (k,k)$ in \cref{eq:H-diag-final}, note first that
\begin{equation}    \label{eq:step1}
\Big\lvert \Big[ 2\sinh(K) \Big(D+W\Big)  \sinh(K) + \cosh(K) \Wt \sinh(K)    + \sinh(K) \Wt  \cosh(K) \Big]_{\alpha,\alpha} \Big\rvert \leq \frac{C}{M}\;.
\end{equation}
From \cref{eq:Ekk} one easily derives the bound
\[   \label{eq:new}
\sum_{\alpha \in \Ik} \lvert \langle \psi, \Ecal_\alpha(k,k) \psi\rangle \rvert \leq \frac{1}{\inf_{\alpha \in \Ik} n_{\alpha}(k)^2} \langle\psi,\Ncal \psi\rangle   \qquad \forall \psi \in \fock\;.
\] 
This bound in combination with \cref{eq:step1} and the fact that, due to the equator cut--off $\lvert k\cdot \hat{\omega}_\alpha\rvert \geq N^{-\delta}$ in the definition of the index set $\Ik$,
\[\inf_{\alpha \in \Ik} n_{\alpha}(k)^2 = \inf_{\alpha \in \Ik} 4\pi k_\F^2 M^{-1} \lvert k \cdot \hat\omega_\alpha \rvert  \left( 1 + o(1) \right)\geq C N^{\frac{2}{3}-\delta}/M\;,\] implies
\begin{align} 
\pm \frac{1}{2} \sum_{\alpha\in \Ik} \Ecal_\alpha (k,k) \Big[&  2\sinh(K) \Big(D+W\Big)  \sinh(K) \nonumber \\
& + \cosh(K) \Wt \sinh(K)    + \sinh(K) \Wt  \cosh(K) \Big]_{\alpha,\alpha} \leq CN^{-\frac{2}{3}+\delta}\Ncal\;. \label{eq:error-ccr-easy}
\end{align}

\paragraph{Vanishing of the off--diagonal terms.} Next we show that the terms in $h_\textnormal{eff}^\textnormal{diag}(k)$ proportional to $(c_\alpha^* c^*_\beta + c_\beta c_\alpha)$ vanish, as intended with the Bogoliubov transformation. Indeed
\begin{align*}
& \cosh(K)(D+W)\sinh(K) + \sinh(K)(D+W)\cosh(K) \\
& \quad + \cosh(K) \tilde{W}\cosh(K) + \sinh(K)\tilde{W}\sinh(K) \\
% & = \frac{1}{4} O^\intercal \Big[ (S_1^\intercal + S_2^\intercal) (D+W) (S_1-S_2)+ (S_1^\intercal-S_2^\intercal) (D+W) (S_1+S_2)\\
% &\hspace{3.5em} + (S_1^\intercal + S_2^\intercal)\tilde{W}(S_1 + S_2) + (S_1^\intercal - S_2^\intercal)\tilde{W}(S_1- S_2) \Big] O \\
& = \frac{1}{4}O^\intercal \Big[ 2S_1^\intercal (D+W+\tilde{W}) S_1 - 2S_2^\intercal (D+W - \tilde{W}) S_2\Big]O \\
& = \frac{1}{4}O^\intercal \Big[ 2E-2E\Big]O=0\;. \tagg{intercal}
\end{align*}
Summarizing we get the lower bound
\begin{equation} \label{eq:summarizing}
h_\textnormal{eff}^\textnormal{diag}(k) \geq \frac{1}{2}\tr\big(E(k)-D(k)-W(k)\big) + \sum_{\alpha,\beta\in \Ik} \Kfrak(k)_{\alpha,\beta}\, c_\alpha^*(k) c_\beta(k) - CN^{-\frac{2}{3}+\delta}{\Ncal} 
\end{equation}
with the matrix\footnote{A computation shows that $\Kfrak(k)=O(k)^\intercal E(k) O(k)$
 but we are not going to use this formula.
%  because the analysis of the unitary $O$ is complicated.
} 
\begin{align*}
\Kfrak(k)& := \cosh(K(k))\big(D(k)+W(k)\big)\cosh(K(k)) + \sinh(K(k))\big(D(k)+W(k)\big)\sinh(K(k)) \\
 & \quad + \cosh(K(k)) \Wt(k)\sinh(K(k)) + \sinh(K(k)) \Wt(k) \cosh(K(k))\;.
\end{align*}

\paragraph{Bogoliubov--diagonalized effective Hamiltonian.} Summing \cref{eq:summarizing} over $k\in \north$ and including the pre--factors as in \cref{eq:DplusQ} (in particular an $\hbar$) we conclude that  
\begin{align} \label{eq:final-quad-diag-01}
&\langle \psi, T_1^* (\Dbb_\B + Q_\B^\Rcal) T_1 \psi\rangle \\
& =  \sum_{k\in \north} 2\hbar \kappaf \lvert k\rvert \langle \psi, T_1^* h_\textnormal{eff}(k) T_1\psi\rangle \nonumber\\
& \geq \sum_{k\in \north} \hbar \kappaf \lvert k\rvert \tr\big(E(k)-D(k)-W(k)\big)  +\sum_{k\in \north} \sum_{\alpha,\beta \in \Ik} 2\hbar \kappaf \lvert k\rvert \big\langle \psi, \Kfrak(k)_{\alpha,\beta} c_\alpha^*(k) c_\beta(k) \psi \big\rangle \nonumber\\
&\quad - C\hbar \Big[N^{-\frac{2}{3}+\delta} \norm{\Ncal^{1/2}\psi}^2 + MN^{-\frac{2}{3}+\delta}  \norm{(\Ncal_\delta+1)^{1/2} \psi} \norm{(\Ncal_\delta+M)^{1/2} (\Ncal +1) \psi} \nonumber\nonumber\\
&\qquad\qquad + \big(\, MN^{-\frac{2}{3}+\delta} \norm{(\Ncal_\delta+M)^{1/2} (\Ncal +1) \psi} \,\big)^2 \Big]\;. \nonumber 
\end{align} 
The error terms on the last two lines of \cref{eq:final-quad-diag-01}  come from  \cref{labeq:THbT-exp-err} and \cref{eq:error-ccr-easy}. 
In \cite[Eq.~(5.15)]{BNP+20} we already showed that the constant term on the right hand side of \cref{eq:final-quad-diag-01} gives rise to the correlation energy $E_N^\textnormal{RPA}$ we aimed to derive,
\begin{align}  \label{eq:GSE-app}
 \sum_{k\in \north} \hbar \kappaf \lvert k\rvert \tr (E(k) - \D(k)-\W(k))  = E_N^\textnormal{RPA} + \hbar \mathcal{O}\big(M^{\frac{1}{4}}N^{-\frac{1}{6}+\frac{\delta}{2}} + N^{-\frac{\delta}{2}} + M^{-\frac{1}{4}}N^{\frac{\delta}{2}}\big) \;.
\end{align} 
This is the correlation energy as given in the main theorem in \cref{eq:rpa_energy}.
\paragraph{Controlling $-\Dbb_\B$ by the diagonalized effective Hamiltonian.} To make use of the positive contribution of $\sum_{\alpha,\beta \in \Ik}\Kfrak(k)_{\alpha,\beta} c_\alpha^*(k) c_\beta(k)$, it is convenient to subtract $D(k)$. We can then expand $\cosh(K(k)) = (\cosh(K(k))-\id) + \id$; the term in $\Kfrak(k)$ where $D(k)$ is multiplied from both sides by the identity cancels with the explicitly subtracted $D(k)$. For the remaining terms we can use \cref{eq:ch-sh} and $\lvert W(k)_{\alpha,\beta}\rvert \leq \frac{C}{M} u_\alpha(k) u_\beta(k)$ to get
\begin{equation}
|(\Kfrak(k) -D(k) )_{\alpha,\beta}| \leq \frac{\hat V(k)}{M} u_\alpha(k) u_\beta(k)\;, \qquad \forall \alpha,\beta\in \Ik\;.   
\end{equation}
Thus, by \cref{lem:kinetic} and recalling $u_\alpha(k) \leq C M^{\frac{1}{2}} N^{-\frac{1}{3}} n_\alpha(k)$, we conclude that for all $\xi\in \fock$
\begin{align*}
&\Big\lvert \Big\langle \xi, \sum_{\alpha,\beta \in \Ik} ( \Kfrak(k) -D(k) )_{\alpha,\beta} c^*_{\alpha}(k) c_{\beta}(k) \xi \Big\rangle  \Big\rvert \leq \sum_{\alpha,\beta \in \Ik} \lvert( \Kfrak(k) -D(k) )_{\alpha,\beta}\rvert \norm{ c_\alpha(k) \xi} \norm{c_\beta(k)\xi} \\
&\leq \sum_{\alpha,\beta \in \Ik} \frac{C \hat{V}(k)}{N^{\frac{2}{3}}} n_\alpha(k) n_\beta(k) \norm{ c_\alpha(k) \xi} \norm{c_\beta(k)\xi} = \frac{C \hat{V}(k)}{N^{\frac{2}{3}}}  \Big(\sum_{\alpha\in \Ik} n_\alpha(k)  \norm{ c_\alpha(k) \xi}\Big)^2 \\
& \leq  \frac{C \hat{V}(k)}{N^{\frac{2}{3}}} N \norm{\Hbb_0^{1/2}\xi}^2 = C  \hat{V}(k) \hbar^{-1} \langle \xi, \Hbb_0 \xi\rangle\;. 
\end{align*} 
Therefore 
\[
\sum_{\alpha,\beta \in \Ik} \Big( \Kfrak(k)  -D(k)  \Big)_{\alpha,\beta} c^*_{\alpha}(k) c_{\beta}(k) \geq  - C  \hat{V}(k) \hbar^{-1} \Hbb_0\;. 
\]
Summing over $k\in \north$
% and moving the contribution of $\sum_{\alpha,\beta \in \Ik} D(k)_{\alpha,\beta} c^*_{\alpha}(k) c_{\beta}(k)$ to the other side
we conclude that
\begin{align} \label{eq:kinetic-Dbb0-final}
\sum_{k\in \north} \sum_{\alpha,\beta \in \Ik} 2\hbar \kappaf \lvert k\rvert  \Kfrak(k)_{\alpha,\beta} c_\alpha^*(k) c_\beta(k) 
\geq  \Dbb_\B   - C \norm{  \hat{V}}_{\ell^1} \Hbb_0\;. 
\end{align}
Inserting the last bound together with \cref{eq:GSE-app} in \cref{eq:final-quad-diag-01}, the proof is complete.
\end{proof}

\section{Proof of the Main Result}

\begin{proof}[Proof of \cref{thm:main}] Recall the definition \cref{eq:Hcorr} of the correlation Hamiltonian $\Hcal_\textnormal{corr}$,
\[
\Hcal_\textnormal{corr} = R^* \Hcal_N R - E^\textnormal{HF}_N =  \Hbb_0 + Q_\B + \mathcal{E}_1 + \mathcal{E}_2 + \Xbb\;.
\]
Let $\Psi \in \fock$ be the approximate ground state constructed by the particle number localization, \cref{lem:loc}, from some exact ground state $\psi_\textnormal{gs}$ of $\Hcorr$ (i.\,e., from a minimizer of the expectation value on the left hand side of \cref{eq:Hcorr-lb}). By the localization we have
\[
\langle \psi_\textnormal{gs}, \Hcal_\textnormal{corr} \psi_\textnormal{gs}\rangle \geq \langle \Psi, \Hcal_\textnormal{corr} \Psi \rangle - CN^{-1}\;.
\]and  
\begin{align} \label{eq:apriori-last}
\langle \Psi, (\Hbb_0 + \Ecal_1) \Psi\rangle \leq C\hbar,  \quad \norm{ (\Ncal_\delta+1)^{1/2} (\Ncal+1)^m \Psi} \leq CN^{\frac{1}{3}m+\frac{\delta}{2} }
\end{align}
for all $m\geq 0$. Here we have used $\Ncal_\delta \leq CN^{\frac{1}{3}+\delta} \Hbb_0$ to estimate the gapped number operator. 

Let $T_1$ be the approximate Bogoliubov transformation defined in \cref{eq:Bog-T} with the Bogoliubov kernel $K(k)$ from \cref{eq:Kk}. Since $T_1$ is unitary, we can define $\psi$ by setting
\[
\Psi=T_1\psi\;.
\]
From \cref{eq:apriori-last} and \cref{lem:stability} we also have
\begin{align} \label{eq:apriori-last-01}
\langle \psi, \Ncal_\delta \psi \rangle \leq CN^\delta, \quad \norm{ (\Ncal_\delta+1)^{1/2} (\Ncal+1) \psi} \leq CN^{\frac{1}{3}+\frac{\delta}{2}}\;, 
\end{align}
and because of $M \gg N^{2\delta}$ we have
\begin{align*}
 \norm{(\Ncal_\delta+M)^{1/2} (\Ncal +1) \psi} & \leq \sqrt{ \langle (\Ncal+1)\psi,(\Ncal_\delta+1) (\Ncal+1)\psi\rangle + \langle (\Ncal+1)\psi,M (\Ncal+1)\psi\rangle }\\
 & \leq \sqrt{CN^{\frac{2}{3}+\delta} + C N^{\frac{2}{3}} M}\leq C M^{\frac{1}{2}} N^{\frac{1}{3}} \;. \tagg{end_err}
\end{align*}

Now we collect the main bounds and estimate the error terms using \cref{eq:apriori-last}, \cref{eq:apriori-last-01}, and \cref{eq:end_err}. First, by  \cref{lem:exchange} we can bound the quadratic operator $\Xbb$ below by 
\begin{align} \label{eq:ferror-1}
\langle \Psi, \Xbb \Psi\rangle \geq - C N^{-\frac{1}{3}} \langle \Psi, \Hbb_0 \Psi\rangle \geq -  CN^{-\frac{2}{3} }\;. 
\end{align}
By  \cref{lem:remove-corridors}, we have
\begin{align} 
\Big\langle \Psi,  (Q_\B + \mathcal{E}_2 - Q_\B^{\Rcal} - \mathcal{E}_2^{\Rcal})  \Psi \Big\rangle  &\geq - C\Big( N^{ - \frac{\delta}{2}} + C N^{ -\frac{1}{6} + \frac{\delta}{2} } M^{\frac{1}{4}}\Big) \Big\langle \Psi, (\Hbb_0+\mathcal{E}_1+\hbar)\Psi\Big\rangle \nonumber\\
&\geq - C\hbar \Big( N^{ - \frac{\delta}{2}} + C N^{ -\frac{1}{6} + \frac{\delta}{2} } M^{\frac{1}{4}}\Big)\;. 
\label{eq:ferror-2}
\end{align}
Next we can use  \cref{lem:lin-2}  to deduce that 
\begin{align}  
% &
\Big\langle \Psi, (\Hbb_0 - \Dbb_\B) \Psi \Big\rangle -  \Big\langle \psi, (\Hbb_0 - \Dbb_\B) \psi \Big\rangle 
% \nonumber 
% \\
% &= \Big\langle T_1\psi, (\Hbb_0 - \Dbb_\B) T_1\psi \Big\rangle  - \Big\langle T_0 \psi, (\Hbb_0 - \Dbb_\B) T_0 \psi \Big\rangle  \nonumber\\
% &\geq  - C\hbar \Big[  M^{-\frac{1}{2}} \norm{  (\Ncal_\delta+1)^{1/2}  \psi  }^2 + CM^{\frac{3}{2}}N^{-\frac{2}{3}+\delta}  \norm{ (\Ncal_\delta+1)^{1/2} (\Ncal+1) \psi} \norm{ (\Ncal_\delta+1)^{1/2}   \psi } \Big]  \nonumber\\
% &
\geq - C\hbar \Big[  M^{-\frac{1}{2}} N^\delta + MN^{-\frac{1}{3}+2\delta } \Big]\label{eq:ferror-3} 
\end{align} 
and, because of the estimate for the matrix elements of the Bogoliubov kernel derived  in \cref{lem:K}, \cref{lem:non--bosonizable} shows that
\begin{align}
&\Big\langle \Psi, (\Ecal_1+ \Ecal_2^{\Rcal}) \Psi\Big\rangle = \Big\langle \psi, T_1^* (\Ecal_1+ \Ecal_2^{\Rcal}) T_1  \psi\Big\rangle \nonumber \\
&\geq - C \norm{\hat V}_{\ell^1} \langle \psi, \Hbb_0 \psi\rangle - CN^{-\frac{1}{2}} \norm{\psi}\norm{\Hbb_0^{1/2} \Psi}  - C N^{-\frac{5}{3}+2\delta} M \norm{(\Ncal_\delta+M)^{1/2} (\Ncal +1)\psi}^2 \nonumber\\
&\geq - C \norm{\hat V}_{\ell^1} \langle \psi, \Hbb_0 \psi\rangle - C \hbar \Big[ N^{-\frac{1}{3}}+ M^2 N^{-\frac{2}{3}+2\delta}\Big]\;. \label{eq:ferror-4}
\end{align}
(Note that here, in the second summand on the second line, in $\norm{\Hbb_0^{1/2}\Psi}$ we have a single instance of the approximate ground state vector $\Psi$ and not $\psi$, so that this norm can be estimated by \cref{eq:apriori-last}.)
Finally, by \cref{lem:diag-final}, the bosonized effective Hamiltonian $\Dbb_\B + Q_\B^{\Rcal}$ yields 
 \begin{align}
&\Big\langle \Psi , (\Dbb_\B + Q_\B^{\Rcal} ) \Psi \Big\rangle- \left( E_N^\textnormal{RPA} +  \big\langle \psi,  \Dbb_\B \psi \big\rangle - C \norm{\hat V}_{\ell^1} \langle \psi, \Hbb_0 \psi\rangle \right) \nonumber\\
&\geq - C \hbar \Big[ N^{-\frac{2}{3}+\delta} \norm{\Ncal^{1/2}\psi}^2 + MN^{-\frac{2}{3}+\delta}  \norm{(\Ncal_\delta+1)^{1/2} \psi} \norm{(\Ncal_\delta+M)^{1/2} (\Ncal +1) \psi} \nonumber\nonumber\\
& \qquad \qquad + \Big( MN^{-\frac{2}{3}+\delta}
% \norm{(\Ncal_\delta+1)^{1/2} \psi} 
\norm{(\Ncal_\delta+M)^{1/2} (\Ncal +1) \psi} \Big)^2 + M^{\frac{1}{4}}N^{-\frac{1}{6}+\frac{\delta}{2}} + N^{-\frac{\delta}{2}} + M^{-\frac{1}{4}}N^{\frac{\delta}{2}} \Big] \nonumber\\
&\geq - C \hbar \Big[ N^{-\frac{1}{3}+\delta} + M^{\frac{3}{2}} N^{-\frac{1}{3}+\frac{3\delta}{2} } \nonumber \\
&\qquad \qquad + \Big( M^{\frac{3}{2}} N^{-\frac{1}{3}+\delta  } \Big)^2  + M^{\frac{1}{4}}N^{-\frac{1}{6}+\frac{\delta}{2}} + N^{-\frac{\delta}{2}} + M^{-\frac{1}{4}}N^{\frac{\delta}{2}} \Big]\;. \label{eq:ferror-5}
\end{align} 

To conclude, we sum all error bounds from \cref{eq:ferror-1}--\cref{eq:ferror-5}. The quantities 
\[
N^{-\frac{\delta}{2}}, \quad M^{-\frac{1}{4}}N^{\frac{\delta}{2}}\quad \text{ and }\quad M^{\frac{3}{2}} N^{-\frac{1}{3}+\frac{3\delta}{2}}
\]
from the error terms suggest us to take
\[
M=N^{4\delta}, \quad \delta= \frac{1}{24}\;. 
\]
With this choice, collecting all of \cref{eq:ferror-1}--\cref{eq:ferror-5}, we conclude that
\begin{align}
\Big\langle \Psi , \Hcal_\textnormal{corr} \Psi \Big\rangle \geq E_N^\textnormal{RPA} + \Big( 1 - C \norm{\hat V}_{\ell^1} \Big)\langle \psi, \Hbb_0 \psi\rangle - C \hbar N^{-\frac{1}{48}}\;. 
\end{align} 
The contribution of $\big( 1 - C \norm{\hat V}_{\ell^1} \big)\langle \psi, \Hbb_0 \psi\rangle$ is non--negative thanks to the smallness assumption for the potential. Thus it remains the error of order $\hbar N^{-\frac{1}{48}} = \hbar^{1+\frac{1}{16}}$, which completes the proof of the main result.
\end{proof}

\appendix

\section{Hartree--Fock Theory} \label{app:B}

In this appendix we consider the Hartree--Fock energy
\begin{align*}
E_N^\textnormal{HF} &:= \inf \Big\{ \langle \psi, H_N \psi\rangle :  \psi=\bigwedge_{j=1}^N u_j \text{ with } \{u_j\}_{j=1}^N \text{ an orthonormal family in } L^2(\mathbb{T}^3) \Big\}
\end{align*}
with  the Hamiltonian $H_N$ in \eqref{eq:HN}. We assume that $N=|\BF|$, namely the Fermi ball is completely filled. It is clear that if the interaction vanishes, i.\,e., $V=0$, then the Slater determinant of plane waves $\psi_\textnormal{pw}$ as in \eqref{eq:plane-waves} is the unique minimizer for $E_N^\textnormal{HF}$. However, it is less trivial that the Hartree--Fock minimizer is unchanged if the interaction is sufficiently weak.

\begin{thm}%[Triviality of Hartree--Fock theory] 
\label{thm:HF} Consider the Hamiltonian $H_N$ in \eqref{eq:HN} with 
\[
0\le \widehat V \in \ell^1(\Zbb^3)\qquad\textnormal{with }0\le \lambda \|\widehat V \|_{\ell^1} < \frac{\hbar^2}{2} \textnormal{ and } N=|\BF|\;.
\]
Then the Slater determinant of plane waves $\psi_\textnormal{pw}$ in \eqref{eq:plane-waves} is the unique minimizer (up to a phase) for $E_N^\textnormal{HF}$.
\end{thm}

Note that  \cref{thm:HF} does not require any specific choice of parameters $\lambda$ and $\hbar$. In our semiclassical mean--field scaling, $\hbar= N^{-1/3}$ and $\lambda=N^{-1}$, the condition $0\le \lambda \|\widehat V \|_{\ell^1} < \hbar^2/2$ holds for $N$ large provided that $\widehat V \in \ell^1$. 

We follow the argument in \cite{GHL19} where the Hartree--Fock energy of the electron gas is studied. The main difference is that in our finite volume setting, the spectral gap of the Laplacian is strong enough to dominate the interaction, ensuring the exact equality $E_N^\textnormal{HF}=\langle \psi_\textnormal{pw}, H_N \psi_\textnormal{pw}\rangle$ instead of just an exponential closeness as in \cite{GHL19}.

\begin{proof} 
Let $\Psi=\bigwedge_{j=1}^N u_j$ be a Slater determinant. Then a straightforward computation shows that 
\begin{align} \label{eq:app-e-diff}
&\langle \Psi, H_N\Psi \rangle -  \langle \psi_\textnormal{pw}, H_N \psi_\textnormal{pw}\rangle \nonumber \\
&=  \tr (-\hbar^2\Delta (\gamma-\gamma_\textnormal{pw})) -  \frac{\lambda}{2}\int_{\Tbb^3}\int_{\Tbb^3}  [|\gamma(x,y)|^2 -|\gamma_\textnormal{pw}(x,y)|^2]   V(x-y) \di x \di y  \nonumber\\
&\quad + \frac{\lambda}{2}\int_{\Tbb^3}\int_{\Tbb^3}  [\rho_{\gamma}(x) \rho_\gamma(y) - \rho_{\gamma_\textnormal{pw}}(x)\rho_{\gamma_\textnormal{pw}}(x)]  V(x-y) \di x \di y
\end{align}
where $\gamma=\sum_{j=1}^N |u_j\rangle \langle u_j|$ is the one--body density matrix of $\Psi$, $\gamma(x,y)=\sum_{j=1}^Nu_j(x) \overline{u_j(y)}$ and $\rho_\gamma(x)=\gamma(x,x)$.  Note that the one--body density matrix $\gamma_\textnormal{pw}$ of the plane waves has the integral kernel 
\[
\gamma(x,y)= \gamma_\textnormal{pw}(x-y)=(2\pi)^{-3}\sum_{p\in \BF}  e^{ip \cdot (x-y)}\;. 
\]
In the following we estimate the right side of \eqref{eq:app-e-diff} term by term. 

\paragraph{Direct term.} Since $\rho_{\gamma_\textnormal{pw}}$ is constant, we have $\widehat \rho_{\gamma_\textnormal{pw}}(k)=0$ for all $0\ne k \in \Zbb^3$. Moreover, $\widehat \rho_{\gamma}(0)= \widehat \rho_{\gamma_\textnormal{pw}}(0)=(2\pi)^{-3}N$. Therefore, using  $\widehat V \ge 0$ we get
\begin{align} \label{eq:app-direct}
&\int_{\Tbb^3}\int_{\Tbb^3}  [\rho_{\gamma}(x) \rho_\gamma(y) - \rho_{\gamma_\textnormal{pw}}(x)\rho_{\gamma_\textnormal{pw}}(y)]  V(x-y) \di x \di y \nonumber\\
&=\sum_{k\in \Zbb^3} \widehat V(k) \Big( |\widehat \rho_{\gamma}(k)|^2 - |\widehat \rho_{\gamma_\textnormal{pw}}(k)|^2 \Big) = \sum_{0\ne k\in \Zbb^3} \widehat V(k)   |\widehat \rho_{\gamma}(k)|^2 \ge 0\;. 
\end{align}

\paragraph{Exchange term.} We decompose 
\[
|\gamma(x,y)|^2 -|\gamma_\textnormal{pw}(x-y)|^2 = |\gamma(x,y)-\gamma_\textnormal{pw}(x-y)|^2 + 2 \Re \Big[ (\gamma(x,y) -\gamma_\textnormal{pw}(x-y)) \gamma_\textnormal{pw}(y-x)  \Big]\;. 
\]
The first part can be estimated as 
\begin{align*}
& \int_{\Tbb^3}\int_{\Tbb^3}  |\gamma(x,y)-\gamma_\textnormal{pw}(x-y)|^2 V(x-y) \di x \di y \\
&\le   \|V\|_{L^\infty} \int_{\Tbb^3}\int_{\Tbb^3}  |\gamma(x,y)-\gamma_\textnormal{pw}(x-y)|^2 \di x \di y =  \|\widehat V\|_{\ell^1}  \tr [(\gamma-\gamma_\textnormal{pw})^2]\;. 
\end{align*}
For the second part of the exchange term, we write 
\[
\int_{\Tbb^3}\int_{\Tbb^3} (\gamma(x,y) -\gamma_\textnormal{pw}(x-y))  \gamma_\textnormal{pw}(y-x)  V(y-x) \di x \di y = \tr ( G (\gamma - \gamma_\textnormal{pw}))
\]
where $G$ is an operator on $L^2(\Tbb^3)$ with kernel $\gamma_\textnormal{pw}(y-x) V(y-x)$. Equivalently, $G$ is the multiplication operator in Fourier space with
\[
G(k) = \sum_{p\in \BF} \widehat V(k-p)\;. 
\]
In particular $G\ge 0$ and hence $\tr ( G (\gamma - \gamma_\textnormal{pw}))\in \Rbb$. Thus  the exchange term is bounded as
\begin{align} \label{eq:app-exchange}
&-  \frac{\lambda}{2}\int_{\Tbb^3}\int_{\Tbb^3}  [|\gamma(x,y)|^2 -|\gamma_\textnormal{pw}(x,y)|^2]   V(x-y) \di x \di y \nonumber \\
&\qquad \ge - \frac{\lambda \|\widehat V\|_{\ell^1}}{2}  \tr [(\gamma-\gamma_\textnormal{pw})^2] - \lambda \tr ( G (\gamma - \gamma_\textnormal{pw}))\;. 
\end{align}
Inserting \eqref{eq:app-direct} and \eqref{eq:app-exchange} in \eqref{eq:app-e-diff} we obtain
\begin{align} \label{eq:app:diff-1}
\langle \Psi, H_N\Psi \rangle -  \langle \psi_\textnormal{pw}, H_N \psi_\textnormal{pw}\rangle  &\ge \tr [ (-\hbar^2\Delta - \lambda G) (\gamma-\gamma_\textnormal{pw})] - \frac{\lambda \|\widehat V\|_{\ell^1}}{2}  \tr [(\gamma-\gamma_\textnormal{pw})^2] \;.
\end{align}

\paragraph{Kinetic term.} Finally, we prove that if $0\le \lambda  \|\widehat V\|_{\ell^1} \le \hbar^2$, then 
\begin{align} \label{eq:app:diff-3}
 \tr [ (-\hbar^2\Delta - \lambda G)  (\gamma-\gamma_\textnormal{pw})]  \ge \frac{\hbar^2 -\lambda \|\widehat V\|_{\ell^1} }{2}  \tr [ (\gamma-\gamma_\textnormal{pw})^2]\;.
\end{align}
To see that, we proceed similarly to \cite[Eq.~(5)]{GHL19}. More precisely, let us find a multiplication operator $A(k) \ge (\hbar^2 -\lambda \|\widehat V\|_{\ell^1} )/2$ on Fourier space such that 
\begin{align} \label{eq:app:diff-2}
 \tr [ (-\hbar^2\Delta - \lambda G)  (\gamma-\gamma_\textnormal{pw})] =  \tr [ A(\gamma-\gamma_\textnormal{pw})^2]\;. 
\end{align}
Since $\gamma$ and $\gamma_\textnormal{pw}$ are projections we can decompose
\[
(\gamma-\gamma_\textnormal{pw})^2= \gamma_\textnormal{pw}^\bot(\gamma-\gamma_\textnormal{pw})\gamma_\textnormal{pw}^\bot -  \gamma_\textnormal{pw}(\gamma-\gamma_\textnormal{pw})\gamma_\textnormal{pw}\;, \quad \gamma_\textnormal{pw}^\bot = 1- \gamma_\textnormal{pw}\;. 
\]
Hence, for any constant $C_0\in \Rbb$,
\begin{align*}
\tr [ A (\gamma-\gamma_\textnormal{pw})^2] &= \tr [ (\gamma_\textnormal{pw}^\bot A  \gamma_\textnormal{pw}^\bot -  \gamma_\textnormal{pw} A  \gamma_\textnormal{pw})  (\gamma-\gamma_\textnormal{pw})] \\
&= \tr [ (\gamma_\textnormal{pw}^\bot A  \gamma_\textnormal{pw}^\bot -  \gamma_\textnormal{pw} A  \gamma_\textnormal{pw} + C_0)  (\gamma-\gamma_\textnormal{pw})]\;.
\end{align*}
Here in the last equality we have used $\tr(\gamma)= \tr (\gamma_\textnormal{pw})=N$. 
Thus the desired equality \eqref{eq:app:diff-2} holds true if 
\[
\gamma_\textnormal{pw}^\bot A  \gamma_\textnormal{pw}^\bot -  \gamma_\textnormal{pw} A  \gamma_\textnormal{pw}  + C_0 =  -\hbar^2\Delta - \lambda G 
\]
which is equivalent to 
\[
A(k) {\mathds 1}(k\in \BF^c) - A(k) {\mathds 1}(k\in \BF) =  \hbar^2 \lvert k\rvert^2 - \lambda G(k) - C_0\;. 
\]
The latter equality holds true when 
\[
A(k)= |\hbar^2|k|^2 - \lambda G(k) - C_0| =  \begin{cases}
\hbar^2|k|^2 - \lambda G(k) - C_0\;, &\quad k\in \BFc \\
- (\hbar^2 |k|^2 - \lambda G(k) - C_0)\;, &\quad k\in \BF
\end{cases}
\]
provided that the constant $C_0$ satisfies 
\[
\sup_{k\in \BF} \Big( \hbar^2 |k|^2 - \lambda G(k) \Big) \le C_0 \le \inf_{k\in \BF^c} \Big( \hbar^2 |k|^2 - \lambda G(k) \Big)\;.
\]
Note that since the Fermi ball is completely filled we have the gap $|k_2|^2- |k_1|^2 \ge 1$ for all $k_1\in \BF$ and $k_2\in \BF^c$ (since $|k_2|^2 - |k_1|^2$ is positive and integer). Furthermore, 
\[
0\le G(k) = \sum_{p\in \BF} \widehat V(k-p) \le \|\widehat V\|_{\ell^1}\;. 
\]
When  $0\le \lambda  \|\widehat V\|_{\ell^1} \le \hbar^2$, we can choose 
\[
C_0 := \frac{1}{2}\inf_{k\in \BFc} \Big( \hbar^2 \lvert k\rvert^2 - \lambda G(k) \Big) +\frac{1}{2} \sup_{k\in \BF} \Big( \hbar^2 \lvert k\rvert^2 - \lambda G(k) \Big)
\]
and obtain 
\[
A(k)= |\hbar^2 |k|^2 - \lambda G(k) - C_0|  \ge \frac{\hbar^2 -\lambda \|\widehat V\|_{\ell^1} }{2} \qquad \forall k \in \Zbb^3\;. 
\]
The desired estimate \eqref{eq:app:diff-3} follows immediately.

%\eqref{eq:app:diff-2} with 
\paragraph{Conclusion.} Inserting \eqref{eq:app:diff-3} in \eqref{eq:app:diff-1} we find that 
\[
\langle \Psi, H_N\Psi \rangle -  \langle \psi_\textnormal{pw}, H_N \psi_\textnormal{pw}\rangle \ge \left ( \frac{\hbar^2}{2}  - \lambda \|\widehat V\|_{\ell^1} \right)  \tr [ (\gamma-\gamma_\textnormal{pw})^2]\;.
\]
Hence, under the condition $\lambda \|\widehat V\|_{\ell^1}<\hbar^2/2$, we conclude that $\psi_\textnormal{pw}$ is the unique Hartree--Fock minimizer. 
\end{proof}

\section{Kinetic Energy Estimates}\label{app:A}

In this appendix we provide a simplified proof of \cref{lem:kinetic}, which was first established in \cite[Lemma~4.7]{HPR20}. Like \cite{HPR20} we use the following special case of a result by \cite{Hux03}.
\begin{thm}[Integer points in ellipses]
\label{thm:ellipses}
Let $d_0 \in \mathbb{N}$. For every $R>0$  consider the ellipse
\[
E(R)= \{ (x,y)\in \Rbb^2: d_0 x^2 + y^2 \le R^2\}\;. 
\]
Then the number of points in $E(R) \cap \Zbb^3$ is $|E(R)| + \Ocal(R^{\gamma})$ for $R\to \infty$, with any $\gamma>131/208$. Here $|E(R)|= \pi d_0^{-1/2} R^2$ is the area of $E(R)$. 
\end{thm}
\noindent We do not need the full power of this theorem; for our purpose, any $\gamma<1$ is sufficient. With exponents $\gamma = 2/3$, it is a classic result due to Van der Corput's thesis \cite{Cor19} from 1919.

\begin{proof}[Proof of \cref{lem:kinetic}] We start by proceeding as in \cite[Lemma~4.7]{HPR20}. Using the Cauchy--Schwarz inequality we get
\begin{align*}
 & \sum_{p \in \BFc \cap (\BF +k)} \norm{a_p a_{p-k} \psi} \\
%  & = \sum_{p \in \BFc \cap (\BF +k)} \frac{1}{\sqrt{e(p)+e(p-k)}} \sqrt{e(p)+e(p-k)} \norm{a_p a_{p-k} \psi} \\
 & \leq \Bigg[\sum_{p \in \BFc \cap (\BF +k)} \frac{1}{e(p)+e(p-k)} \Bigg]^{1/2} \Bigg[\sum_{p \in \BFc \cap (\BF +k)}\left( e(p)+e(p-k) \right) \norm{a_p a_{p-k} \psi}^2 \Bigg]^{1/2}\;.
\end{align*}
The second factor is bounded by the kinetic energy as claimed,
\begin{align*}
 & \sum_{p \in \BFc \cap (\BF +k)}\left( e(p)+e(p-k) \right) \norm{a_p a_{p-k} \psi}^2 \\
 & \leq \sum_{p \in \BFc \cap (\BF +k)} e(p) \norm{a_p \psi}^2 + \sum_{p \in \BFc \cap (\BF +k)} e(p-k) \norm{a_{p-k} \psi}^2 \\
 & \leq \sum_{p \in \BFc} e(p) \langle \psi, a^*_p a_p \psi\rangle + \sum_{h \in \BF} e(h) \langle \psi, a^*_h a_h \psi \rangle = \langle \psi, \Hbb_0 \psi\rangle\;. 
\end{align*}
The hard part is to show that there exists some constant $C > 0$ such that
\[
 \sum_{p \in \BFc \cap (\BF +k)} \frac{1}{e(p)+e(p-k)} \leq C N \;,
\]
which is equivalent to 
\begin{equation}\label{eq:toprove}
 \sum_{p \in \BFc \cap (\BF +k)} \frac{1}{\lvert p\rvert^2 - \lvert p-k\rvert^2} \leq C N^{\frac{1}{3}} \;.
\end{equation}

Heuristically, one may understand this bound as follows: the sum is over lattice points in the grey area of \cref{fig:normalization}; since $\lvert k \rvert = \Ocal(1)$ and the area of the Fermi surface is of order $N^{\frac{2}{3}}$, the number of these points is also of order $N^{\frac{2}{3}}$. Since $\lvert p\rvert \sim N^{\frac{1}{3}}$ ($p$ has to be close to the Fermi surface) and $\lvert k\rvert \sim 1$, we expect that in average $\lvert p\rvert^2 - \lvert p-k\rvert^2 = 2 p \cdot k - k^2 \sim N^{\frac{1}{3}}$, leading us to the order $N^{\frac 13}$ in \cref{eq:toprove}. 

Strictly speaking, $\lvert p\rvert^2 - \lvert p-k\rvert^2$ can be much smaller than the average size.  Actually $\lvert p\rvert^2 - \lvert p-k\rvert^2$ can be of order $O(1)$ if both $p$ and $p-k$ are very close to the Fermi surface. Fortunately, integer points very close to the surface of the sphere are quite rare and the contribution from this part can be controlled by number theoretic results. 
%according to a famous number theoretic result\footnote{ See, e.\,g., \cite{BRS17}: for large $n\in \Nbb$ we have 
%   $r_3(n) := \lvert \{ x \in \Zbb^3: \lvert x \rvert^2 =n \} \rvert \ll n^{\frac{1}{2}+\bourgain }$ for any $\bourgain > 0$.  
%  This bound follows directly from the bound for the analogous two-dimensional problem, $r_2(m) := \lvert \{ x \in \Zbb^2 \colon \lvert x \rvert^2 =m \} \rvert \ll m^{\bourgain }$ for any $\bourgain > 0$, which can be deduced easily from the exact formula for $r_2(m)$ \cite[Theorem 9.1]{Tak18}. Alternatively, the bound for $r_3(n)$ can be proved by expressing $r_3(n)$ through the Dirichlet $L$--function $L(1,\chi_{-n})$ using a famous formula by Gauss \cite{Arn92} and the Dirichlet class number formula; the $L$--function can then be bounded by a logarithm as in \cite[Lemma 10.15]{MV06}.}, the number of lattice points exactly on the sphere is not much bigger than proportional to the radius $N^{\frac{1}{3}}$, again leading us to an overall order $N$. Even more precisely, one notices that the dangerous situation is restricted to a region around the equator (if we think of $k$ indicating the direction of the north pole) of the Fermi surface in \cref{fig:normalization}, so that the number of points is even smaller.   

This argument was formulated rigorously in \cite{HPR20} but the detailed proof is rather technical. Here we present a simplified proof of \eqref{eq:toprove} for the reader's convenience. 

% \begin{proof}[Proof of \cref{eq:toprove}] 
Note that for every $p\in \BFc \cap (\BF+k)$ we have
\[
1\le p^2-(p-k)^2 = 2 p\cdot k - k^2 \;.
\]
The lower bound $1$ follows from the fact that $|p| >|p-k|$ and that $p,p-k\in \Zbb^3$. Hence, 
\[\frac{1}{p^2-(p-k)^2} = \frac{1}{2p\cdot k} \left(1 + \frac{k^2}{2p \cdot k - k^2}\right) \leq \frac{1}{2p\cdot k} \left(1 + k^2 \right) \leq \frac{C}{p\cdot k}\;.\]
Moreover, when $p\in \BFc \cap (\BF+k)$ we have $|p|\le CN^{\frac{1}{3}}$, and hence $|p\cdot k|\le CN^{\frac{1}{3}}$. Recall also that $p,k \in \Zbb^3$, so that $p\cdot k\in \Zbb$. Thus we get
\begin{align} \label{eq:sum-1/ep-2}
 \sum_{p\in \BFc \cap (\BF+k)} \frac{1}{ p^2-(p-k)^2  } \le  \sum_{p\in \BFc \cap (\BF+k)} \frac{C}{p\cdot k} \le \sum_{s\in \Zbb \cap [(1+k^2)/2,CN^{1/3}]}  \frac{C|B_s|}{s}
\end{align}
where
\[
B_s := \{p\in \BFc \cap (\BF+k) : p\cdot k = s\}\;. 
\]
To count $|B_s|$ we use \cref{thm:ellipses}; we are going to show that for any $1>\gamma > 131/208$ we have
\begin{align} \label{eq:sum-1/ep-3}
|B_s| \le C_\gamma (|s| + N^{\frac{\gamma}{3}}) \qquad \forall s\in \Zbb \cap [(1+k^2)/2,CN^{1/3}]\;. 
\end{align}

\paragraph{Easy case.} Assume that $k=(k_1,0,0)$ with $k_1\ne 0$. Consider $p=(p_1,p_2,p_3)\in B_s$. Then $p_1=(p\cdot k)/k_1 = s/k_1$ and 
\[
p^2 > k_F^2 \geq (p-k)^2 \iff k_F^2 + 2 s - k^2 - p_1^2  \geq p_2^2 + p_3^2 > k_F^2 - p_1^2\;. 
\]
Thus, considered in the plane $\Rbb^2$, we have $(p_2,p_3)\in B(0,R_2)\backslash B(0,R_1)$ with the radii of the two centered balls being
\[
R_1= \sqrt{k_F^2 - p_1^2}\;, \quad R_2 = \sqrt{k_F^2 + 2 s - k^2 - p_1^2}\;. 
\]
Note that 
\[
R_1<R_2 \le C N^{\frac{1}{3}}\;, \quad R_2^2- R_1^2 \le C s\;.  
\]
Thus $|B_s|$ is bounded by the number of integer points in the annulus $B(0,R_2)\backslash B(0,R_1)$, which according to \cref{thm:ellipses} is bounded by
\[
|B_s|\le |B(0,R_2)|- |B(0,R_1)| + \Ocal(N^{\frac{\gamma}{3}}) = \pi (R_2^2-R_1^2) + \Ocal(N^{\frac{\gamma}{3}})  \le C (s + N^{\frac{\gamma}{3}})\;. 
\]
  
\paragraph{General case.} If $k$ is aligned with one of the basis vectors $(1,0,0), (0,1,0), (0,0,1) \in \Rbb^3$, then we can proceed as above. Otherwise, we can assume  $k=(k_1,k_2,k_3)$ with $k_1,k_2\ne 0$.  Consider the set of orthogonal vectors in $\Rbb^3$ given by
 \begin{align*}
k=(k_1,k_2,k_3)\;,\quad k_\bot=(0, -k_3,k_2)\;,\quad k_\bot' = (-k_2^2-k_3^2, k_1 k_2, k_1 k_3)\;. 
\end{align*} 
Every $p\in \Zbb^3$ is determined uniquely by $(n_1,n_2,n_3)\in \Zbb^3$ with 
\[
n_1= p\cdot k\;, \quad n_2= p\cdot k_\bot\;, \quad n_3= p \cdot k_\bot'\;; 
\]
in fact
\[
p = \Big(p\cdot \frac{k}{\lvert k\rvert} \Big) \frac{k}{\lvert k\rvert} + \Big(p\cdot \frac{k_\bot}{\lvert k_\bot\rvert} \Big) \frac{k_\bot}{\lvert k_\bot\rvert} + \Big(p\cdot \frac{k'_\bot}{\lvert k'_\bot\rvert} \Big) \frac{k'_\bot}{\lvert k'_\bot\rvert}\;.
\]
Then, using 
\[
p^2 = \left| p \cdot \frac{k}{|k|} \right|^2 + \left| p \cdot \frac{k_\bot}{|k_\bot|} \right|^2 + \left| p \cdot \frac{k_\bot'}{|k_\bot'|} \right|^2
\]
and the identity $|k_\bot'|=|k| |k_\bot|$ we find that
\[
|k_\bot'|^2 p^2 = |k_\bot|^2  \left| p \cdot k \right|^2 + |k|^2 \left| p \cdot  k_\bot  \right|^2 + \left| p \cdot k_\bot' \right|^2 = |k_\bot|^2  n_1^2  + |k|^2 n_2^2 + n_3^2\;. 
\]
Consequently, if $p\in B_s$, then $n_1=k\cdot p = s$ and  
\[
p^2 > k_F^2 \geq (p-k)^2  \iff |k_\bot'|^2 ( k_F^2 + 2 s - k^2) - |k_\bot|^2  s^2   \geq   |k|^2 n_2^2 + n_3^2  > |k_\bot'|^2 k_F^2 - |k_\bot|^2  s^2\;. 
\]
Thus $(n_2,n_3)\in E(R_2)\backslash E(R_1)$ where $E(R)$ is the ellipse 
\[
E(R)= \{(x,y)\in \Rbb^2: |k|^2 x^2+ y^2 \le R^2\}
\]
and 
\[
R_1= \sqrt{|k_\bot'|^2 k_F^2 - |k_\bot|^2   s ^2}\;, \quad R_2= \sqrt{ |k_\bot'|^2 ( k_F^2 + 2 s - k^2)  - |k_\bot|^2  s ^2  }\;. 
\]
Note that $R_1<R_2 \le C N^{\frac{1}{3}}$ and $R_2^2 - R_1^2 \le C s$. Thus $|B_s|$ is bounded by the number of integer points in  $E(R_2)\backslash E(R_1)$, which is bounded by \cref{thm:ellipses} (with $d_0= \lvert k\rvert^2$) by: 
\begin{align*}
|B_s|\le |E(R_2)| -|E(R_1)| + \Ocal(N^{\gamma/3}) = \pi \frac{R_2^2}{|k|} - \pi \frac{R_1^2}{|k|} + \Ocal(N^{\frac{\gamma}{3}}) \le C(s + N^{\frac{\gamma}{3}})\;. 
\end{align*}

Thus in conclusion, we have proved that \eqref{eq:sum-1/ep-3} holds for every  $0\ne k\in \Zbb^3$. Inserting \eqref{eq:sum-1/ep-3} in \eqref{eq:sum-1/ep-2} we conclude that
\begin{align*}
 \sum_{p\in \BFc \cap (\BF+k)} \frac{1}{ p^2-(p-k)^2  }  &\le \sum_{s\in \Zbb \cap [(1+k^2)/2,CN^{1/3}]}  \frac{C |B_s|}{s} \\
 &\le \sum_{s= 1}^{CN^{1/3}}  \frac{C (s+N^{\frac{\gamma}{3}})}{s } \le CN^{\frac{1}{3}}+ N^{\frac{\gamma}{3}} \log N\;.
 \end{align*}
 Since $\gamma<1$, this implies \cref{eq:toprove}.
\end{proof}

The above proof can be adapted easily to give \cref{lem:kineticequator}. 

\begin{proof}[Proof of \cref{lem:kineticequator}] As in the proof of \cite[Lemma 4.7]{HPR20}, by Cauchy--Schwarz we get  
\begin{align*}
\sum_{\substack{p\colon p \in \BFc \cap (\BF +k)\\ e(p)+e(p-k) \leq 4 N^{-\frac{1}{3} -\delta}  }} \norm{a_p a_{p-k} \psi} &\leq 
 \Bigg[ \sum_{\substack{p\colon p \in \BFc \cap (\BF +k)\\ e(p)+e(p-k) \leq 4 N^{-\frac{1}{3} -\delta}  }}  (e(p)+e(p-k))^{-1} \Bigg]^{1/2} \times \nonumber\\
 &\quad \times  \Bigg[ \sum_{\substack{p\colon p \in \BFc \cap (\BF +k)\\ e(p)+e(p-k) \leq 4 N^{-\frac{1}{3} -\delta}  }}  (e(p)+e(p-k))  \norm{a_p  a_{p-k} \psi}^2 \Bigg]^{1/2}\;. 
\end{align*}
In the second factor we simply drop the constraint $e(p)+e(p-k) \leq 4 N^{-\frac{1}{3} -\delta}$ and bound it by the kinetic energy $\Hbb_0$ exactly as in the proof of \cref{lem:kinetic}.
It remains to prove  
\begin{align} \label{eq:IABC}
\sum_{\substack{p\colon p \in \BFc \cap (\BF +k)\\ \lvert p\rvert^2 - \lvert p -k \rvert^2 \leq 4 N^{\frac{1}{3}-\delta}}}  \frac{1}{\lvert p\rvert^2 - \lvert p-k\rvert^2} \leq CN^{\frac{1}{3}-\delta}\;.  
\end{align}
We proceed exactly as in the proof of  \cref{eq:toprove}. The only difference is that the condition $4 N^{\frac{1}{3} -\delta} \ge p^2 - (p-k)^2 = 2 p\cdot k -k^2 \geq 1$ implies that $\lvert p\cdot k\rvert\le C N^{\frac{1}{3} -\delta}$. Hence, using \eqref{eq:sum-1/ep-3} we have 
\begin{align*}
\sum_{\substack{p\colon p \in \BFc \cap (\BF +k)\\ \lvert p\rvert^2 - \lvert p-k\rvert^2 \leq 4 N^{1/3 -\delta}  }} \frac{1}{ p^2-(p-k)^2  } &\le  \sum_{s\in \Zbb \cap [(1+k^2)/2,CN^{1/3-\delta}]}  \frac{C |B_s|}{s}\\
&\le \sum_{s=1}^{4N^{\frac{1}{3}-\delta}}  \frac{C (s+N^{\frac{\gamma}{3}})}{s} \leq CN^{\frac{1}{3}-\delta} + N^{\frac{\gamma}{3}} \log N. 
\end{align*}
Since the latter bound holds true for every $\gamma>131/208$, we conclude that \eqref{eq:IABC} holds true for every $\delta< 77/624$. This completes the proof of Lemma 4.2.
\end{proof}

% \section{Plane Waves minimize the Hartree--Fock Functional}

\bibliographystyle{alpha}
\newcommand{\etalchar}[1]{$^{#1}$}

\end{document}